\documentclass[a4paper,11pt]{article}
\pdfoutput=1

\usepackage{jheppub}
\usepackage[T1]{fontenc}
\usepackage{amsmath,amssymb,amsthm,mathtools,xcolor}
\usepackage{tikz}
\usepackage{needspace}
\usetikzlibrary{arrows.meta,calc}
\newtheorem{colorprop}{Proposition}[section]
\newtheorem{colorlemma}{Lemma}[section]
\allowdisplaybreaks
\let\OriginalTableOfContents\tableofcontents
\renewcommand{\tableofcontents}{%
 \begingroup
 \renewcommand{\baselinestretch}{1}\small
 \OriginalTableOfContents
 \endgroup
}

\def\e{{\rm e}}
\def\R{{\cal R}}
\newcommand{\N}{{\cal N}}
\newcommand{\zh}[1]{\zeta_{q^2}(#1)}
\newcommand{\Ezero}{\widehat E_0}
\newcommand{\Ehat}{\widehat E}
\newcommand{\Pq}{\mathcal{P}_q}
\newcommand{\wt}{\operatorname{wt}}

\title{A discrete analogue of $\mathcal N=2$ super Weil--Petersson volumes}

\author[a]{Wasif Ahmed}
\author[b,c]{Joydeep Naskar}

\affiliation[a]{Physics Department, Broida Hall, University of California, Santa Barbara, CA 93106, USA}
\affiliation[b]{Beijing Institute of Mathematical Sciences and Applications, Huaibei Town, Huairou District, Beijing 101408, China}
\affiliation[c]{Kavli Institute for Theoretical Physics, Santa Barbara, CA 93106, USA}

\emailAdd{wasif@ucsb.edu}
\emailAdd{joydeepnaskar@bimsa.cn}

\abstract{
We introduce a discrete spectral curve from $\mathcal N=2$ double-scaled SYK whose topological recursion defines quasi-polynomial volumes in integer boundary lengths. The curve has two edge theories: a local $\mathcal N=2$ theory at the lower edge and a bosonic theory at the upper edge. Focusing on the local $\mathcal N=2$ theory, we compute its finite-$q$ volumes and independently check the first few using the string equation. In the $q\to1$ continuum limit, these recover the super Weil--Petersson volumes of Turiaci and Witten~\cite{Turiaci:2023jfa}. Finally, we give an all-genus intersection-theoretic interpretation of these local volumes on the moduli space of stable marked curves, in terms of a deformed spin class with a gap-dependent $\kappa$-class dressing.
}
\begin{document}
\maketitle
\clearpage
\flushbottom

%-------------------------------------------------
\section{Introduction}
\label{sec:intro}

Jackiw--Teitelboim (JT) gravity is one of the simplest models of quantum gravity in which the gravitational amplitudes can be computed using the geometry of moduli space~\cite{Jackiw:1984je,Teitelboim:1983ux,Saad:2019lba,Witten:2020volumes}. At fixed topology, the calculation involves the Weil--Petersson (WP) volumes of hyperbolic surfaces with prescribed boundary lengths. These volumes are polynomials in the squared lengths, and they satisfy Mirzakhani's recursion~\cite{Mirzakhani:2006fta,Mirzakhani:2006inter}. Their coefficients can also be written as intersection numbers on moduli space~\cite{Witten:1990hr,Kontsevich:1992ti}, so the same quantities can be studied using gravitational path integrals, recursion relations or intersection theory.

The matrix-integral description of JT gravity relates these volumes to a spectral problem. The disk density of states is encoded in a spectral curve, and once the two-boundary input is specified, Eynard--Orantin recursion computes all the higher amplitudes from it~\cite{Eynard:2007fi,EynardOrantin:2007wp}. This prescription does not require the density to come from JT gravity, so we can equally well use it for a density obtained from a microscopic model. In particular, it is interesting to understand the volume theory associated with the full microscopic density, before taking the limit that gives JT gravity.

For the Sachdev--Ye--Kitaev (SYK) model in its double-scaled limit (DSSYK), such a construction was carried out by Okuyama~\cite{Okuyama:2023byh}. There, the exact bosonic density gives a discrete spectral curve whose volumes are quasi-polynomials in integer boundary lengths, i.e., there is one polynomial for each choice of boundary parities. Discrete volumes of this type also appear in lattice-point counts on moduli space~\cite{Norbury:2008,DoNorbury:2011}. If we keep only the leading local terms of the SYK construction, we obtain the $q$-deformed volumes of Do and Norbury~\cite{DoNorbury:qrecursion,DoNorbury:DSSYK}, which reduce to the ordinary WP volumes after a further $q\to1$ scaling. At finite $q$, however, the global curve retains both endpoints of the energy band, and it defines a discrete theory that goes beyond this local limit.

Supersymmetric JT gravity has its own volume theories. For $\mathcal N=1$, the recursion and its relation to spin geometry have been studied extensively~\cite{Stanford:2019vob,Johnson:2020heh,Norbury:2017class,Norbury:2020enumerative,Norbury:2023super,AlexandrovNorbury:2024}. The $\mathcal N=2$ case is richer, since the positive-energy states come in supermultiplets labeled by a charge, with thresholds that depend on the charge, and there are also BPS states at zero energy. Turiaci and Witten derived the corresponding recursion for the super WP volumes~\cite{Turiaci:2023jfa}, and matrix-model and string-equation approaches to these volumes were developed in~\cite{Johnson:2023ofr,Johnson:2025power,Ahmed:2025extended}. Since the exact density of $\mathcal N=2$ DSSYK is known~\cite{Berkooz:2020uly}, we can now ask for a finite-$q$ extension of this construction.

In this work, we construct an $\mathcal N=2$ discrete spectral curve and compute its quasi-polynomial volumes. We start from the SYK Hamiltonian and use the averaged moments to obtain the density, and then its resolvent. A sector of fixed state charge receives positive-energy contributions from two adjacent values of the average multiplet charge, so we separate these contributions and work at one fixed average charge. The resulting density leads us to the global curve~\eqref{eq:global-curve}, which, for generic charge, has a pole at zero energy, away from the continuum band. In this way, the global curve keeps the charge-dependent threshold as well as both ends of the band.

Applying topological recursion to this curve, we obtain the discrete volumes~\eqref{eq:discrete-def}, and we prove that each parity polynomial has total degree at most $3g-3+n$ in the squared boundary lengths, where $g$ is the genus and $n$ is the number of boundaries. We also reconstruct a matrix potential for the disk density, and relate the residue of the zero-energy pole at generic charge to the weight of BPS states at a shifted charge, in the range of charges where this interpretation applies. We emphasize that only the disk input is fixed by the microscopic density. The two-boundary input and the recursion define the matrix-model extension that we consider here, and its agreement with the microscopic SYK correlators beyond the disk is an open question, which we discuss in Section~\ref{sec:discussion}.

Next, we study the two edge theories of the global curve. The upper edge gives a bosonic local theory, while the lower edge retains the $\mathcal N=2$ pole. We will focus on the latter, and compute its finite-$q$ local volumes together with their dependence on the gap. In a correlated $q\to1$ limit, in which the gap scales as $(1-q)^2$, these become the $\mathcal N=2$ super WP volumes of Turiaci and Witten~\cite{Turiaci:2023jfa}. We check the first few volumes independently using the string equation, and give an all-genus intersection-theoretic interpretation of them on the moduli space of stable marked curves. When the gap closes, the density develops a hard edge, and the intersection formula is then written in terms of Norbury's $\Theta$-class. Whether the finite-$q$ volumes can also be understood as a measure on a space of $\mathcal N=2$ super Riemann surfaces is an interesting open question.

The continuum limit of the global curve needs some care. The $\mathcal N=2$ and bosonic JT disk curves both follow directly by scaling near the respective endpoints. For the volumes, however, we must also take every recursion residue at the chosen endpoint, and only then do the lower volumes give the Turiaci--Witten theory~\cite{Turiaci:2023jfa} and the upper volumes give the ordinary WP volumes. The full amplitudes can contain mixed-endpoint terms that survive this limit, so in general they cannot be identified with either local theory. We also take $q\to0$ in the two edge theories at fixed gap. In this limit, the lower edge becomes an Airy--Bessel curve, whose intersection class is the deformed negative $2$-spin class, and the upper edge becomes the Airy curve. We summarize these relations in Figure~\ref{fig:curve-limits}.

\begin{figure}[htbp]
\centering
\begin{tikzpicture}[
  font=\small,
  line/.style={-{Stealth[length=2mm]},semithick,draw=black!70},
  block/.style={draw=black!60,rounded corners=2pt,align=center,inner sep=6pt},
  edge/.style={block,fill=black!3,text width=5.45cm,minimum height=1.1cm},
  result/.style={block,text width=2.75cm,minimum height=1.85cm,font=\footnotesize},
  tag/.style={font=\footnotesize,fill=white,inner sep=2pt}
]
\node[block,fill=black!7,text width=4.6cm,minimum height=1.35cm] (full) at (0,0)
  {\textbf{Global curve},\\$\mathcal N=2$ SYK multiplet,\\Discrete volumes $N^q_{g,n}$};

\node[edge] (lower) at (-3.5,-2.45)
  {\textbf{Lower edge: $z=+1$},\\Finite-$q$ $\mathcal N=2$ local theory};
\node[edge] (upper) at (3.5,-2.45)
  {\textbf{Upper edge: $z=-1$},\\Finite-$q$ bosonic local theory};
\coordinate (fork) at (0,-1.38);
\draw[semithick,draw=black!70] (full.south) -- node[tag,right] {edge expansion} (fork);
\draw[line] (fork) -| (lower.north);
\draw[line] (fork) -| (upper.north);

\node[result] (lowerone) at (-5.2,-5.0)
  {$\mathcal N=2$ JT disk,\\Turiaci--Witten\\super WP volumes};
\node[result] (lowerzero) at (-1.8,-5.0)
  {Airy--Bessel curve,\\Deformed negative\\$2$-spin class,\\Bessel at zero gap};
\node[result] (upperone) at (1.8,-5.0)
  {Bosonic JT disk,\\Ordinary WP volumes};
\node[result] (upperzero) at (5.2,-5.0)
  {Airy curve,\\$\psi$-class\\intersections};
\draw[line] (lower.south) -- ++(0,-0.42) -| node[tag,pos=0.8,right] {$q\to1$} (lowerone.north);
\draw[line] (lower.south) -- ++(0,-0.42) -| node[tag,pos=0.8,right] {$q\to0$} (lowerzero.north);
\draw[line] (upper.south) -- ++(0,-0.42) -| node[tag,pos=0.8,right] {$q\to1$} (upperone.north);
\draw[line] (upper.south) -- ++(0,-0.42) -| node[tag,pos=0.8,right] {$q\to0$} (upperzero.north);
\end{tikzpicture}
\caption{The global curve~\eqref{eq:global-curve}, its discrete volumes~\eqref{eq:discrete-def} and the two edge theories. For the volumes, each edge branch means that we retain the leading graded terms and take \emph{every} recursion residue at that endpoint. The $q\to1$ arrows use $\Ezero=(1-q)^2E_0$ and the edge normalizations~\eqref{eq:full-edge-JT}, and both JT disk limits can also be obtained directly from the global curve.}
\label{fig:curve-limits}
\end{figure}
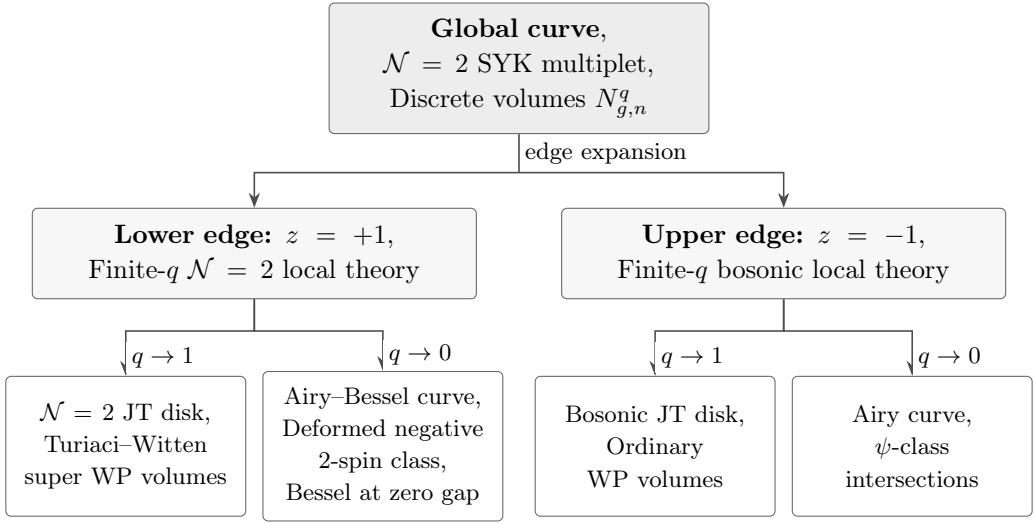

\textbf{Organization}: The paper is organized as follows. In Section~\ref{sec:global}, we derive the global curve from SYK, and in Section~\ref{sec:globalTR} we compute its discrete volumes. We study the two edges and their limits in Sections~\ref{sec:edges} and~\ref{sec:localvol}. In Section~\ref{sec:ODE}, we compute the local volumes again using the string equation, in Section~\ref{sec:BPS} we relate the pole of the curve to the BPS states, and in Section~\ref{sec:intersection} we give the intersection formula. We conclude with a discussion of future directions in Section~\ref{sec:discussion}.

\textbf{Reader's guide}: Readers interested mainly in the physics may read Sections~\ref{sec:global}--\ref{sec:BPS} and skip Section~\ref{sec:intersection}, while readers interested mainly in the mathematics may start from the spectral data~\eqref{eq:global-curve} and read Sections~\ref{sec:globalTR}--\ref{sec:localvol} and~\ref{sec:intersection}, skipping the SYK, string-equation and BPS derivations.

%-------------------------------------------------
\section{From the SYK Hamiltonian to the spectral curve}
\label{sec:global}

In this section, we derive the spectral curve from the microscopic model. We will first compute the disorder-averaged moments of the Hamiltonian at fixed charge, then read off the density of states from these moments, and finally continue the density away from the energy band to obtain the curve.

\subsection{Hamiltonian and fixed-charge moments}
\label{subsec:SYK-main}

Let us first recall the model. The $\mathcal N=2$ SYK model~\cite{Fu:2016vas} consists of $N$ complex fermions $\chi_i$, with $\bar\chi_i=\chi_i^\dagger$ and $\{\chi_i,\bar\chi_j\}=\delta_{ij}$. For odd $p$, we write $\chi_{\mathbf j}=\chi_{j_1}\cdots\chi_{j_p}$ for a multi-index $\mathbf j=(j_1,\ldots,j_p)$ with $j_1<\cdots<j_p$. The supercharge and the Hamiltonian, together with the statistics of the independent complex Gaussian couplings $C_{\mathbf j}$, are given by
\begin{equation}
\begin{gathered}
 Q=\sum_{\mathbf j} C_{\mathbf j}\chi_{\mathbf j},\qquad H=\frac12\{Q,Q^\dagger\},\\
 \langle C_{\mathbf j}\rangle=0,\qquad
 \langle C_{\mathbf j}C_{\mathbf j'}^*\rangle=\frac{2^p}{\binom{N}{p}}\,\delta_{\mathbf j\mathbf j'}.
\end{gathered}
\label{eq:syk-model}
\end{equation}
We normalize the trace by $\operatorname{tr}1=1$, so that $\langle\operatorname{tr}H\rangle=1$. We will label the states by the eigenvalue $s$ of the charge $\gamma=p^{-1}\sum_i(\bar\chi_i\chi_i-\tfrac12)$, which obeys $[\gamma,Q^\dagger]=Q^\dagger$, and work in the double-scaled limit
\begin{equation}
 N,p\longrightarrow\infty,\qquad
 \lambda=\frac{2p^2}{N}\quad\text{fixed},\qquad q=\e^{-\lambda}.
 \label{eq:DSSYK-limit}
\end{equation}

Since $Q^2=(Q^\dagger)^2=0$, only alternating words appear in the positive powers of the Hamiltonian,
\begin{equation}
 H^k=2^{-k}\big[(QQ^\dagger)^k+(Q^\dagger Q)^k\big].
\label{eq:syk-alternating-words}
\end{equation}
The Gaussian average pairs each $Q$ with a $Q^\dagger$, which gives the oriented chord diagrams of~\cite{Berkooz:2020uly}. After the charge projection, the sum over chord diagrams becomes a transfer-matrix problem, whose eigenfunctions are continuous $q$-Hermite polynomials with base $q^2$ and angular measure
\begin{equation}
 \mathfrak w_q(\varphi)=(q^2,\e^{2i\varphi},\e^{-2i\varphi};q^2)_\infty,
 \qquad \int_0^\pi\frac{d\varphi}{2\pi}\,\mathfrak w_q(\varphi)=1.
 \label{eq:main-chord-measure}
\end{equation}
Here $(v_1,\ldots,v_m;q)_\infty=\prod_{j=1}^{m}\prod_{n\geq0}(1-v_jq^n)$ is the $q$-Pochhammer symbol. There are two branches of positive-energy states, with average multiplet charges $r=s\pm\tfrac12$, and their dispersion relation is
\begin{equation}
 E_r(\varphi)=q^{-1/4}\big[\cosh(\lambda r)-\cos\varphi\big],
 \qquad 0<\varphi<\pi.
 \label{eq:dispersion}
\end{equation}
Let $\rho(E;s)$ be the disorder-averaged density of states per unit energy and charge, in the normalized trace. Its moments $m_k(s)=\int dE\,E^k\rho(E;s)$ for $k\geq1$ are given by
\begin{equation}
 m_k(s)=\frac{\sqrt\lambda}{4\pi^{3/2}}
 \sum_{\pm}q^{(s\pm1/2)^2}
 \int_0^\pi d\varphi\,\mathfrak w_q(\varphi)E_{s\pm1/2}(\varphi)^{k-1}.
\label{eq:syk-energy-moments}
\end{equation}
Note that the energy appears with the power $E_r^{k-1}$, rather than $E_r^k$. This is because the squared norm of a supercharge-exact transfer-matrix eigenstate is proportional to its eigenvalue, so the spectral resolution supplies an inverse power of the energy.\footnote{Here $\mathsf Q$ and $\mathsf Q^\dagger$ are the chord operators representing $Q$ and $Q^\dagger$ on open-chord states, and $\mathsf T=\{\mathsf Q,\mathsf Q^\dagger\}$ is the transfer matrix. For an eigenvector with $\mathsf T\Psi=\mathcal E\Psi$ and $\mathsf Q^\dagger\Psi=0$, we have $\|\mathsf Q\Psi\|^2=\mathcal E\|\Psi\|^2$. The transfer matrix, its $q$-Hermite measure and the vacuum overlap are derived in Appendix~\ref{app:SYK-transfer}.} As we will see, this is the origin of the $1/E$ factor that distinguishes the supersymmetric curve from the bosonic one. The chord calculation and the normalization checks can be found in Appendix~\ref{app:SYK}.

\subsection{The density of one supermultiplet}

We can now read off the continuum density from the moments, by moving the extra factor of $1/E$ into the measure. Since a possible contribution at $E=0$ is invisible to all the positive moments, we have to allow for it separately, and write
\begin{equation}
\begin{aligned}
 \rho(E;s)&=\rho_{s-1/2}(E)+\rho_{s+1/2}(E)+\rho_{\rm BPS}(s)\delta(E),\\
 \rho_r(E)&=\frac{\sqrt\lambda}{4\pi^{3/2}}q^{r^2}
 \int_0^\pi d\varphi\,\mathfrak w_q(\varphi)
 \frac{\delta(E-E_r(\varphi))}{E}.
\end{aligned}
\label{eq:syk-sector-density}
\end{equation}
The BPS weight $\rho_{\rm BPS}(s)$ is fixed by the zeroth moment, $\int dE\,\rho(E;s)=\sqrt{\lambda/\pi}\,q^{s^2}$, and we will discuss it in Section~\ref{sec:BPS}. Using $dE_r/d\varphi=q^{-1/4}\sin\varphi$, the density of one branch becomes
\begin{equation}
 \rho_r(E)=\frac{\sqrt\lambda\,q^{1/4}}{4\pi^{3/2}}
 \frac{q^{r^2}}{E}\frac{\mathfrak w_q(\varphi)}{\sin\varphi}.
\label{eq:syk-multiplet-density}
\end{equation}
This density counts one component of a positive-energy multiplet in each of the adjacent sectors $r\pm\tfrac12$, and the two components have the same density, as required by supersymmetry. So a generic sector of fixed $s$ contains two bands with four distinct endpoints, and the two bands coincide only at $s=0$. Note, however, that this decomposition of the one-point function does not by itself imply that the microscopic charge spectra are statistically independent (see also Section~\ref{sec:discussion}).

It is useful to note that the Gaussian factor $q^{r^2}$ and the threshold in~\eqref{eq:syk-energy-moments} carry the same charge shift. In Appendix~\ref{app:SYK-density}, we check this pairing and the prefactor against the first two moments of the Hamiltonian, and in doing so we correct the crossed pairing in Eq.~(5.18) of~\cite{Berkooz:2020uly}.

\subsection{From the density to the global curve}
\label{subsec:global-curve}

We now fix the multiplet charge $r$ and define
\begin{equation}
 c=\cosh(\lambda r),\qquad
 \Ezero=\frac{c-1}{2}=\sinh^2\!\frac{\lambda r}{2},
 \label{eq:gap-from-charge}
\end{equation}
and introduce the variable $z=\e^{i\varphi}$ together with the affine energy coordinate
\begin{equation}
 x(z)=z+z^{-1},\qquad
 \Ehat=\frac{q^{1/4}E}{2}=\frac{2c-x}{4}.
 \label{eq:global-x}
\end{equation}
In these variables the cut is at $x\in[-2,2]$, or equivalently $\Ehat\in[\Ezero,1+\Ezero]$, so that $\Ezero$ is the threshold. Splitting off the zero-index Pochhammer factors and using the identity
\begin{equation}
 (1-q^{2n}z^2)(1-q^{2n}z^{-2})
 =(1-q^{2n})^2-q^{2n}(z-z^{-1})^2
 \label{eq:product-factorization}
\end{equation}
gives
\begin{equation}
 \frac{\mathfrak w_q(\varphi)}{\sin\varphi}
 =4A_q^3\sin\varphi\,\Pq(z),\qquad A_q=(q^2;q^2)_\infty,
 \label{eq:poch-product}
\end{equation}
where
\begin{equation}
 \Pq(z)=\prod_{n\geq1}
 \left[1-\frac{q^{2n}(z-z^{-1})^2}{(1-q^{2n})^2}\right].
 \label{eq:Pq-global}
\end{equation}
Since $\sin\varphi=2\sqrt{(\Ehat-\Ezero)(1+\Ezero-\Ehat)}$, the density~\eqref{eq:syk-multiplet-density} becomes
\begin{equation}
\begin{aligned}
 \rho_r(E)\,dE&=\mathcal W_r\,\rho_0^{\rm glob}(\Ehat)\,d\Ehat,\\
 \mathcal W_r&=4\hbar\sqrt{\frac\lambda\pi}\,
 q^{r^2+1/4}A_q^3,\\
 \rho_0^{\rm glob}(\Ehat)
 &=\frac{\sqrt{(\Ehat-\Ezero)(1+\Ezero-\Ehat)}}{2\pi\hbar\Ehat}
   \Pq(\e^{i\varphi}).
\end{aligned}
\label{eq:syk-global-density}
\end{equation}
Here $\hbar$ is the genus-counting parameter of the matrix integral~\eqref{eq:pot-matrix-integral}, so that an amplitude of genus $g$ with $n$ boundaries is of order $\hbar^{2g-2+n}$. In~\eqref{eq:syk-global-density}, we have separated the charge-dependent weight $\mathcal W_r$ from the shape of the density, and removing this weight fixes the disk normalization that we use below. The factors of $\hbar$ in $\mathcal W_r$ and $\rho_0^{\rm glob}$ cancel in the microscopic measure $\rho_r(E)\,dE$.

The last step is to encode the density in the discontinuity of a resolvent, following~\cite[Sections 2--3]{Okuyama:2023byh}. In the matrix coordinate, the density is $\varrho_{q,c}(x)=\hbar\rho_0^{\rm glob}((2c-x)/4)/4$. Let $\mathcal V_{q,c}$ be the matrix potential, and define
\begin{equation}
 R_{q,c}(x)=\int_{-2}^{2}\frac{\varrho_{q,c}(t)}{x-t}\,dt,
 \qquad
 y(x)=\frac12\mathcal V'_{q,c}(x)-R_{q,c}(x).
 \label{eq:main-resolvent-y}
\end{equation}
The saddle-point equation $\mathcal V'_{q,c}=R_{q,c,+}+R_{q,c,-}$, where $R_{q,c,\pm}(x)=R_{q,c}(x\pm i0)$, then gives $y(x+i0)=i\pi\varrho_{q,c}(x)$. We reconstruct the potential $\mathcal V_{q,c}$ itself in Appendix~\ref{app:potential}. If we continue the square root with $\sqrt{x^2-4}\sim x$ at infinity, the Joukowsky map~\eqref{eq:global-x} gives $\sqrt{x^2-4}=z-z^{-1}$, and we arrive at
\begin{equation}
 \boxed{x(z)=z+z^{-1},\qquad
 y(z)=\frac18\frac{z-z^{-1}}{2c-x(z)}\Pq(z),\qquad
 B(z_1,z_2)=\frac{dz_1dz_2}{(z_1-z_2)^2}}.
 \label{eq:global-curve}
\end{equation}
We will refer to~\eqref{eq:global-curve} as the global curve. It keeps both ends of the energy band, and it is the input for the recursion in the next section, with the Bergman kernel $B$ supplying the two-boundary input. The coefficient $1/8$ is fixed by the energy Jacobian.

The map $z\mapsto z^{-1}$ exchanges the two sheets. Since the energy decreases as $x$ increases, $z=+1$ is the lower endpoint of the energy band and $z=-1$ is the upper one. Compared with the bosonic curve, we have replaced $q$ by $q^2$ and changed the disk normalization, and in addition we have acquired the factor $(2c-x)^{-1}$, which is nothing but the original $1/E$. For $c>1$, its pole lies at zero physical energy,
\begin{equation}
 x=2c,\qquad z_\ast^\pm=c\pm\sqrt{c^2-1}.
 \label{eq:global-pole-location}
\end{equation}
For generic charge, this remains a pole of $y$. There are two exceptions. At nonzero integer $|r|=n$, the $n$th factor of $\Pq$ cancels the pole on both sheets, while at $r=0$ the pole meets the lower endpoint and produces a hard edge. We treat these exceptions, as well as the BPS residue, in Appendix~\ref{app:SYK-BPS}.

For the recursion near either endpoint, it will be convenient to write
\begin{equation}
 \Pq(z)=\e^{-\Phi_q(z)},\qquad
 \Phi_q(z)=\sum_{m\geq1}\frac{Z_m}{m}(z-z^{-1})^{2m},
 \label{eq:Phi-global}
\end{equation}
with
\begin{equation}
 Z_m=\zh{2m}=\sum_{n\geq1}\frac{q^{2nm}}{(1-q^{2n})^{2m}}.
 \label{eq:qzeta-intro}
\end{equation}
The series converges for $|z-z^{-1}|<2\sinh\lambda$, and the product gives its continuation outside this region. This is enough for our purposes, since at each topology the recursion uses only finitely many of these coefficients. In the next section, we construct the discrete volumes from these global spectral data.

%-------------------------------------------------
\section{From the curve to discrete volumes}
\label{sec:globalTR}
%-------------------------------------------------

\subsection{The recursion and its length expansion}

With the disk and two-boundary data in hand, we can now define the higher amplitudes by Eynard--Orantin recursion~\cite{Eynard:2007fi,Eynard:2014zxa}, where each step takes residues at the two endpoints of~\eqref{eq:global-curve}. For $n\geq1$ and $2g-2+n>0$, the recursion in our sign convention reads
\begin{equation}
\begin{split}
 \omega_{g,n}(z_1,J)=\sum_{\pm}\mathop{\rm Res}_{z\to\pm1}
 \big(-K_{\rm glob}(z_1,z)\big)
 \Big[&\omega_{g-1,n+1}(z,z^{-1},J)\\
 &+\sum_{g_1+g_2=g}^{\prime}
 \sum_{I\sqcup I^c=J}
 \omega_{g_1,|I|+1}(z,I)
 \omega_{g_2,|I^c|+1}(z^{-1},I^c)
 \Big],
\end{split}
\label{eq:global-TR}
\end{equation}
where $J=(z_2,\ldots,z_n)$ and $\omega_{0,2}=B$, and the prime means that we exclude the $(0,1)$ factors but keep the Bergman-kernel contributions $(0,2)$. The recursion kernel is given by
\begin{equation*}
 K_{\rm glob}(z_1,z)=\frac{\int_{z^{-1}}^{z}B(z_1,\cdot)}{2\,[y(z)-y(z^{-1})]\,dx(z)}.
\end{equation*}
Since $y(z^{-1})=-y(z)$, we can write it more explicitly as
\begin{equation}
 K_{\rm glob}(z_1,z)=
 \frac{2\,[2c-x(z)]\,z^3\e^{\Phi_q(z)}}
 {(z^2-1)(z_1-z)(z_1z-1)}\frac{dz_1}{dz}.
\label{eq:global-kernel}
\end{equation}
The kernel vanishes at a genuine pole of $y$ away from the cut. Using this, one can show by induction on $2g-2+n$, starting from the Bergman kernel, that every stable correlator is rational in each external variable, regular at the origin, and has poles only at $z_i=\pm1$.

These two poles are responsible for the parity dependence of the volumes, since the pole at $-1$ contributes an alternating sign to the Taylor coefficients. More precisely, inversion symmetry gives a basis of differentials $\sum_{b\geq1}(\pm1)^b b^{2k+1}z^{b-1}dz$ with $k\geq0$, as we show in Appendix~\ref{app:coloring-parity}. The integer $b$ will become a boundary length $b_i$ below, and after dividing out one factor of $b_i$ for each boundary, the remaining dependence is polynomial in the $b_i^2$.

We therefore define the discrete volumes $N^q_{g,n}$ by
\begin{equation}
 \omega_{g,n}(z_1,\ldots,z_n)
 =\sum_{b_1,\ldots,b_n\geq1}
 N^q_{g,n}(b_1,\ldots,b_n)
 \prod_{i=1}^{n}b_i z_i^{b_i-1}dz_i.
\label{eq:discrete-def}
\end{equation}
This is Okuyama's discrete-length expansion~\cite{Okuyama:2023byh}, which we take around $z_i=0$; it is related by inversion to the physical sheet $|z|>1$ of Appendix~\ref{app:potential-Joukowsky}. For each fixed parity pattern, $N^q_{g,n}$ is a polynomial of total degree at most $3g-3+n$ in the $b_i^2$, with coefficients in
\begin{equation}
 \mathbb{Q}\big[c,Z_1,Z_2,\ldots,Z_{3g-3+n}\big].
\label{eq:ring}
\end{equation}
The full quasi-polynomial is invariant under permutations of the boundaries, while a polynomial for a fixed parity pattern is symmetric only under the permutations that preserve that pattern. We prove the total-degree bound, and the finite range of $Z_m$ in~\eqref{eq:ring}, in Appendix~\ref{app:global-degree}.

The discrete volumes are called quasi-polynomials for this reason, i.e., one polynomial suffices for each parity pattern of the boundaries, but not for all patterns at once. Such objects also appear in the lattice-point models of moduli space~\cite{Norbury:2008,DoNorbury:2011} and in bosonic DSSYK~\cite{Okuyama:2023byh,DoNorbury:DSSYK}. As the first examples will show, the supersymmetric pole changes which parity sectors contribute.

\subsection{The first examples and the two endpoints}
\label{subsec:edge-split}

At $(g,n)=(1,1)$, we find
\begin{equation}
 N^q_{1,1}(b)=
 \begin{cases}
 4c\left[\dfrac{b^2-4}{24}+Z_1\right],& b\ \text{even},\\[6pt]
 -4\left[\dfrac{b^2-1}{24}+Z_1\right],& b\ \text{odd}.
 \end{cases}
\label{eq:N11}
\end{equation}
Here both parities contribute. The reason can be seen from the kernel~\eqref{eq:global-kernel}, since $2c-x(z)=2c-z-z^{-1}$ has a part that is even under $z\mapsto-z$ and a part that is odd under it. The even part preserves the parity selection of the bosonic curve, while the odd part flips it. For the bosonic curve, the two edges are equivalent, and only one of these sectors is allowed~\cite{Okuyama:2023byh,DoNorbury:DSSYK}. In our case, however, the two edges carry different weights, and both sectors survive.

For three boundaries, the result takes a particularly simple form,
\begin{equation}
 N^q_{0,3}(b_1,b_2,b_3)=
 \begin{cases}
 8c,& b_1+b_2+b_3\ \text{even},\\
 -8,& b_1+b_2+b_3\ \text{odd}.
 \end{cases}
\label{eq:N03}
\end{equation}
For four boundaries, let $|\mathbf b|^2=\sum_{i=1}^{4}b_i^2$ and let $n_{\rm even}$ be the number of even boundary lengths. Then we find
\begin{equation}
 N^q_{0,4}=
 \begin{cases}
 16[(c^2+1)(|\mathbf b|^2+24Z_1)-4c^2+2],&n_{\rm even}=0,4,\\
 -32c(|\mathbf b|^2+24Z_1-1),&n_{\rm even}=1,3,\\
 16[(c^2+1)(|\mathbf b|^2+24Z_1)-2c^2],&n_{\rm even}=2.
 \end{cases}
\label{eq:N04}
\end{equation}
The $(1,2)$ result, which we give in Appendix~\ref{app:globaldata}, contains all four parity sectors and depends on $Z_1$ and $Z_2$. The two parity polynomials at $(2,1)$, which are quartic in $b^2$ with coefficients in $\mathbb{Q}[c,Z_1,Z_2,Z_3,Z_4]$, are also collected there.

In order to see the two local theories separately, we will keep track of the endpoint used at each step of the recursion, labeling a residue at $+1$ by $+$ and a residue at $-1$ by $-$. This is the colored-recursion description of~\cite{Eynard:2007fi,Eynard:2014zxa}, where the colors simply record the choices of endpoint. The full answer is a sum over all choices, and if we use one endpoint throughout, we get a \emph{monochromatic} contribution.

Already at $(1,1)$, we can keep the two residues separate. Writing $N^{(+)}$ for the contribution from $z=+1$ and $N^{(-)}$ for that from $z=-1$, we find
\begin{equation}
 N^{(+)}_{1,1}(b)=-\frac14+\Ezero
 \left[\frac{b^2-4}{6}+4Z_1\right],
\label{eq:N11plus}
\end{equation}
\begin{equation}
 N^{(-)}_{1,1}(b)=2(-1)^b\left[
 \frac{b^2}{12}+2Z_1+\Ezero\left(\frac{b^2-4}{12}+2Z_1\right)-\frac{5}{24}
 \right].
\label{eq:N11minus}
\end{equation}
Their sum gives back~\eqref{eq:N11}. As we will see in the next section, the two pieces have very different edge limits.

%-------------------------------------------------
\section{Two inequivalent edge theories}
\label{sec:edges}
%-------------------------------------------------

The factor of $1/E$ already suggests that the two edges are different, since the upper edge stays away from zero energy, while the lower edge can approach the pole. To make this precise, we will keep only the leading terms in an expansion near each edge. For this purpose, we first treat the $q$-zeta values $Z_m$ as independent parameters and assign the weights

\begin{equation}
 \wt(b_i)=1,\qquad \wt(Z_m)=2m,\qquad \wt(\Ezero)=-2,
\label{eq:grading}
\end{equation}
which match the scalings $b_i\sim(1-q)^{-1}$, $Z_m\sim(1-q)^{-2m}$ and $\Ezero\sim(1-q)^2$ in the limit of Section~\ref{sec:localvol}. By the \emph{top} part of a quantity, we mean its part of highest total weight in all three types of variables, which is more than just the highest power of the boundary lengths.

\subsection{The upper edge: a bosonic theory}

Let us start with the upper edge. Near $z=-1$, we define
\begin{equation}
 \tilde w^2=-\frac{(z+1)^2}{4z}.
\label{eq:upper-u}
\end{equation}
In terms of $\tilde w$, we have the exact relations
\begin{equation}
 x=-2-4\tilde w^2,
 \qquad
 (z-z^{-1})^2=16\tilde w^2(1+\tilde w^2).
\label{eq:upper-exact}
\end{equation}
The physical energy is then
\begin{equation}
 \Ehat=1+\Ezero+\tilde w^2.
\end{equation}
At $\tilde w=0$, the pole factor is $2c-x=4(1+\Ezero)$, which is regular. If we expand its inverse, we get a constant followed by powers of $\tilde w^2$ and $\Ezero$, and every nonconstant term has a lower weight. Similarly,
\begin{equation}
 (z-z^{-1})^{2m}=(4\tilde w)^{2m}(1+\tilde w^2)^m=(4\tilde w)^{2m}+\text{lower graded terms}.
\end{equation}
So if we keep only the terms of highest weight, we get a quadratic Airy coordinate and
\begin{equation}
 y^{\rm top}\propto y_{\rm up}(\tilde w)=
 \tilde w\exp\!\left[-\sum_{m\geq1}\frac{Z_m}{m}(4\tilde w)^{2m}\right].
\label{eq:upper-top}
\end{equation}
After a constant rescaling, this is the bosonic $q$-deformed top curve of Do and Norbury~\cite{DoNorbury:qrecursion,DoNorbury:DSSYK}, with the replacement
\begin{equation}
 \zeta_q(2m)\longrightarrow4^mZ_m.
\label{eq:upper-substitution}
\end{equation}
Here $\zeta_q(2m)=\sum_{k\geq1}q^{km}/(1-q^k)^{2m}$ are the bosonic $q$-zeta values. We write $V^{q,\mathrm{bos}}_{g,n}$ for the bosonic local volumes of Do and Norbury in these variables, and define their upper-edge specialization by
\begin{equation*}
 V^{q,\mathrm{up}}_{g,n}(\mathbf b)
 :=\left.V^{q,\mathrm{bos}}_{g,n}(\mathbf b)
   \right|_{\zeta_q(2m)\mapsto4^mZ_m}.
\end{equation*}
The all-upper recursion gives $V^{q,\mathrm{up}}_{g,n}$. We fix its normalization and prove the corresponding edge identity in Appendix~\ref{app:coloring-upper}.

\subsection{The lower edge: the pole survives}

Now let us turn to the lower edge. Near $z=+1$, we define
\begin{equation}
 w^2=\frac{(z-1)^2}{4z}.
\label{eq:lower-w}
\end{equation}
In terms of $w$, we then have
\begin{equation}
 x=2+4w^2,
 \qquad
 (z-z^{-1})^2=16w^2(1+w^2),
\label{eq:lower-exact}
\end{equation}
and the physical energy becomes
\begin{equation}
 \Ehat=\Ezero-w^2.
\label{eq:lower-energy}
\end{equation}
Here the denominator is $2c-x=4(\Ezero-w^2)$, and in the graded scaling the gap shrinks at the same rate as $w^2$, so we have to keep both terms. As a result, the supersymmetric pole survives in the lower local curve. Moreover,
\begin{equation}
 (z-z^{-1})^{2m}=(4w)^{2m}(1+w^2)^m=(4w)^{2m}+\text{lower graded terms}.
\end{equation}
We will use $w$ itself as the local coordinate, with $x_{\rm loc}=w^2/2$. Matching $y\,dx$ as in Appendix~\ref{app:SYK-curve}, we then find the lower local curve,
\begin{equation}
 x_{\rm loc}(w)=\frac12w^2,\qquad
 y_{\rm loc}(w)=\frac{w}{\Ezero-w^2}
 \exp\!\left[-\sum_{m=1}^{\infty}\frac{Z_m}{m}(4w)^{2m}\right].
\label{eq:local-curve}
\end{equation}
For $\Ezero>0$, the density has a soft edge and vanishes as $\sqrt{\Ehat-\Ezero}$.

When the gap closes, i.e., at $\Ezero=0$, the curve degenerates to
\begin{equation}
 x_{\rm loc}(w)=\frac12w^2,
 \qquad
 y_{\rm loc}(w)=-\frac{1}{w}
 \exp\!\left[-\sum_{m\geq1}\frac{Z_m}{m}(4w)^{2m}\right].
\label{eq:gapless-curve}
\end{equation}
Now $y_{\rm loc}\sim-1/w$, and the density diverges as $\Ehat^{-1/2}$. This is a \emph{hard edge}, which is the familiar behavior in supersymmetric matrix ensembles~\cite{Altland:1997,Stanford:2019vob,Turiaci:2023jfa}.

\subsection{Why parity averaging is not enough}
\label{subsec:coloring}

One might try to isolate the lower theory by averaging over the boundary parities. This works for the external variables, but it does not control the choices of endpoint inside the recursion. To see this, let $\varsigma_i=(-1)^{b_i}$ and decompose the discrete volumes as
\begin{equation}
 N^q_{g,n}(b_1,\ldots,b_n)
 =\sum_{\mathcal I\subseteq\{1,\ldots,n\}}
 \left(\prod_{i\in\mathcal I}\varsigma_i\right)N_{\mathcal I}(b_1^2,\ldots,b_n^2).
\label{eq:parity-character}
\end{equation}
Averaging over the $2^n$ parity polynomials leaves $N_\varnothing$, which selects the poles at $+1$ in the external variables. However, an internal residue can still come from $-1$. At $(1,2)$, this gives an extra $-2\Ezero Z_1$ in the top part expressed in local lengths, and we derive this term, together with the next example at $(2,1)$, in Appendix~\ref{app:coloring-parity}.

Instead, we have to restrict the recursion itself. Let $N^{(+\cdots+)}_{g,n}$ denote the contribution to~\eqref{eq:discrete-def} in which every recursion vertex uses the residue at $z=+1$. Then the volumes $V^q_{g,n}$ of the lower local curve~\eqref{eq:local-curve}, which we define in~\eqref{eq:local-volume-def} below, are given by the graded lower-edge limit of this contribution,
\begin{equation}
 V^q_{g,n}(b_1,\ldots,b_n)=
 2^{-n}
 \left[N^{(+\cdots+)}_{g,n}\right]_{2(g-1)}
 \left(\frac{b_1}{2},\ldots,\frac{b_n}{2}\right).
\label{eq:colored-identity}
\end{equation}
Here the subscript selects the terms of weight $2(g-1)$, and the rescaled lengths and the factor $2^{-n}$ come from the local coordinate. This identity relates the lower local theory to the discrete volumes, and we prove it for every stable marked topology in Appendix~\ref{app:coloring-proof}, where we also give its upper-edge counterpart and compute the mixed-endpoint corrections.

%-------------------------------------------------
\section{Local volumes and limits of the global curve}
\label{sec:localvol}
%-------------------------------------------------

\subsection{A first example and the general pattern}

From now on, $V^q_{g,n}$ denotes the lower local volumes obtained from~\eqref{eq:local-curve}. Their correlators $\omega^{\rm loc}_{g,n}$ are generated by the recursion~\eqref{eq:global-TR}, now with a single residue at $w=0$, the involution $w\mapsto-w$, $B=dw_1dw_2/(w_1-w_2)^2$ and the kernel~\eqref{eq:local-kernel} given below. Let $\mathcal L$ be the Laplace transform $f(b)\mapsto\int_0^\infty\e^{-bw}f(b)\,db$ in each variable. The inverse transform then turns these correlators into polynomials in the boundary lengths,
\begin{equation}
 V^q_{g,n}(b_1,\ldots,b_n)=
 \frac{1}{b_1\cdots b_n}
 \mathcal{L}^{-1}_{w_1,\ldots,w_n}
 \left[
 \frac{\omega^{\rm loc}_{g,n}(w_1,\ldots,w_n)}{dw_1\cdots dw_n}
 \right].
\label{eq:local-volume-def}
\end{equation}

The simplest example is the genus-one volume with one boundary,
\begin{equation}
 V^q_{1,1}(b)=-\frac18+\frac{\Ezero}{48}(b^2+96Z_1).
\label{eq:V11-local}
\end{equation}
The constant $-1/8$ survives when the gap closes, while the remaining term is linear in the gap, with all its $q$-dependence carried by $Z_1$. Both features can be understood from the recursion kernel.

With the local involution $w\mapsto-w$, the recursion kernel is
\begin{equation}
 K_{\rm loc}(w_1,w)=
 \frac{\Ezero-w^2}{2w(w_1^2-w^2)}
 \exp\!\left[\sum_{m\geq1}\frac{Z_m}{m}(4w)^{2m}\right]\frac{dw_1}{dw}.
\label{eq:local-kernel}
\end{equation}
Each stable correlator contains $2g-2+n$ kernels, and each of them is linear in $\Ezero$. Hence the local volumes are polynomials in the gap,
\begin{equation}
 V^q_{g,n}=\sum_{j=0}^{2g-2+n}\Ezero^j V^{q,[j]}_{g,n},
\label{eq:gap-polynomial}
\end{equation}
which extends the threshold-energy polynomial of $\mathcal N=2$ JT supergravity to finite $q$~\cite{Turiaci:2023jfa,Ahmed:2025extended}. Moreover, the grading~\eqref{eq:grading} gives
\begin{equation}
 \wt(V^q_{g,n})=2(g-1).
\label{eq:homogeneity}
\end{equation}
At the highest power of the gap, the $-w^2$ term in the numerator of~\eqref{eq:local-kernel} is subleading, so the coefficient of $\Ezero^{2g-2+n}$, which we denote by $[\Ezero^{2g-2+n}]V^q_{g,n}$, is given by the upper-edge volumes $V^{q,\mathrm{up}}_{g,n}$ defined in Section~\ref{sec:edges},
\begin{equation}
 [\Ezero^{2g-2+n}]V^q_{g,n}=V^{q,\mathrm{up}}_{g,n}.
\label{eq:top-gap}
\end{equation}
The coefficient of the highest power of the gap is therefore bosonic, with the same substitution as in~\eqref{eq:upper-substitution}, which is also the case for the continuum $\mathcal N=2$ volumes~\cite{Turiaci:2023jfa,Ahmed:2025extended}.

When the gap closes, the stable volumes at genus zero vanish, since their recursion requires a factor of $\Ezero$. The genus-one volumes, on the other hand, survive and are independent of $q$. Writing $\mathbf b=(b_1,\ldots,b_n)$ and showing the gap as the last argument, we find
\begin{equation}
 V^q_{1,n}(\mathbf b;0)=\frac{(-1)^n(n-1)!}{8}.
\label{eq:gapless-genus1}
\end{equation}
The first $q$-dependence of the gapless theory appears at genus two. We list the gapped and gapless local volumes, including the complete $(2,1)$ polynomial, in Appendix~\ref{app:localvolumes}.

We also record the disk density of the local curve, which on the physical cut $\Ehat>\Ezero$ is given by
\begin{equation}
 \rho^q_0(\Ehat)=\frac{1}{2\pi\hbar}
 \frac{\sqrt{\Ehat-\Ezero}}{\Ehat}
 \exp\!\left[-\sum_{m\geq1}\frac{(-16)^mZ_m}{m}(\Ehat-\Ezero)^m\right].
\label{eq:local-density}
\end{equation}
Outside the disk of convergence of the logarithmic series, the density is defined by the entire product~\eqref{eq:app-ODE-density}.

\subsection{The global curve as $q\to1$}
\label{subsec:full-q-one}

The JT disk curves can be obtained directly from the global curve~\eqref{eq:global-curve}, without extracting the top part, by zooming into an edge and rescaling $y\,dx$. If we instead send $q\to1$ at fixed $x$ inside the cut, $\Pq$ diverges, so there is no finite unscaled interior curve in this normalization.

We set $\epsilon=1-q$ and scale the gap and the boundary lengths as
\begin{equation}
 \Ezero=\epsilon^2E_0,\qquad b_i\longmapsto\frac{b_i}{\epsilon}.
\label{eq:local-scaling}
\end{equation}
This is the scaling at fixed multiplet charge, since $\Ezero/\epsilon^2\to r^2/4$ at fixed $r$. For later use, we also note that, with $\zeta$ the Riemann zeta function,
\begin{equation}
 \lim_{q\to1^-}\epsilon^{2m}Z_m=\frac{\zeta(2m)}{4^m},\qquad m\geq1.
\label{eq:qzeta-limit}
\end{equation}
To see this, note that $\epsilon^2q^{2n}/(1-q^{2n})^2\leq1/(4n^2)$, and that each summand tends to this bound. Dominated convergence then proves~\eqref{eq:qzeta-limit}, and it also controls the tails of the products.

The exact coordinates
\begin{equation}
 z_\epsilon^\pm(\bar w)=\pm\big(\sqrt{1+\epsilon^2\bar w^2}+\epsilon\bar w\big)^2,
 \qquad x(z_\epsilon^\pm)=\pm(2+4\epsilon^2\bar w^2)
 \label{eq:full-edge-coordinates}
\end{equation}
resolve the lower and upper energy edges, respectively, and if we pull back the full one-form, we find
\begin{equation}
 \begin{aligned}
 \epsilon^{-1}(z_\epsilon^+)^*(y\,dx)&\longrightarrow
       \frac{\bar w\sin(2\pi\bar w)}{2\pi(E_0-\bar w^2)}\,d\bar w,\\
 \epsilon^{-3}(z_\epsilon^-)^*(y\,dx)&\longrightarrow
       \frac{\bar w\sin(2\pi\bar w)}{2\pi}\,d\bar w.
 \end{aligned}
 \label{eq:full-edge-JT}
\end{equation}
Here $\bar w=w/\epsilon$ at the lower edge and $\bar w=\tilde w/\epsilon$ at the upper edge. With $\bar x=\bar w^2/2$, these are the $\mathcal N=2$ and bosonic JT disk curves, and within either edge the Bergman kernel tends to $d\bar w_1\,d\bar w_2/(\bar w_1-\bar w_2)^2$. We give the exact pullbacks and the limit of the product in Appendix~\ref{app:SYK-continuum}. The lower limit also holds at $E_0=0$, where the curve is irregular. On the other hand, if we hold the unscaled gap $\Ezero>0$ fixed, the pole stays away from both endpoints, and after constant normalizations both edge limits are bosonic.

In particular, if we write $\Ehat=\epsilon^2E_c$, where $E_c$ is the continuum energy, and show the gap as a second argument, the \emph{global} density satisfies
\begin{equation}
 \lim_{q\to1^-}\epsilon\,\rho^{\rm glob}_0(\epsilon^2E_c;\epsilon^2E_0)
 =\frac{\sinh(2\pi\sqrt{E_c-E_0})}{4\pi^2\hbar E_c},
 \qquad E_c>E_0.
 \label{eq:full-density-limit}
\end{equation}
The local density has the same limit, which at zero gap reduces to
\begin{equation}
 \lim_{q\to1^-}(1-q)\rho^q_0((1-q)^2E_c)
 =\frac{\sinh(2\pi\sqrt{E_c})}{4\pi^2\hbar E_c}.
\label{eq:density-limit}
\end{equation}
These conventions include a rescaling of the disk normalization as well as of the energy, and we restore the microscopic SYK weight in Appendix~\ref{app:SYK-continuum}.

\subsection{Limits of the volume polynomials}
\label{subsec:volume-limits}
\label{subsec:q-zero}

For the volumes, we must also specify which residues are included. If we take every residue at the lower endpoint, we get the Turiaci--Witten limit. More precisely, the homogeneity~\eqref{eq:homogeneity} and the limit~\eqref{eq:qzeta-limit} give the Turiaci--Witten volumes $V^{\N=2}_{g,n}$ as the limit
\begin{equation}
 V^{\N=2}_{g,n}(\mathbf b;E_0)=
 \lim_{q\to1^-}\epsilon^{2(g-1)}
 V^q_{g,n}(\mathbf b/\epsilon;\epsilon^2E_0).
\label{eq:volume-limit}
\end{equation}
The limiting local curve is the $\mathcal N=2$ JT curve, and the coefficients that we have computed agree with~\cite{Turiaci:2023jfa,Ahmed:2025extended}. For example, at genus one we find
\begin{equation}
 V_{1,1}^{\N=2}(b)=-\frac18+\frac{E_0}{48}(b^2+4\pi^2).
\label{eq:V11-cont}
\end{equation}
At $(2,1)$, the comparison checks all four powers of $E_0$. Similarly, the all-upper local volumes, normalized as in~\eqref{eq:app-color-upper-identity}, reduce to the Weil--Petersson volumes $V^{\rm WP}_{g,n}$ in the bosonic limit,
\begin{equation}
 \lim_{q\to1^-}\epsilon^{6g-6+2n}
 V^{q,\mathrm{up}}_{g,n}(\mathbf b/\epsilon)
 =V^{\rm WP}_{g,n}(\mathbf b).
 \label{eq:upper-volume-limit}
\end{equation}
This is the bosonic scaling of~\cite{Okuyama:2023byh}, with the deformation parameters related by~\eqref{eq:upper-substitution}.

However, the full recursion remembers both endpoints. The upper-edge one-form vanishes two powers of $\epsilon$ faster than the lower one, but its inverse enters the kernel, and this can compensate for the suppression of the terms connecting the two edges. We can see this effect already at $(1,2)$, where
\begin{equation}
 \lim_{q\to1^-}\frac14
 N^q_{\varnothing;1,2}\!
 \left(\frac{b_1}{2\epsilon},\frac{b_2}{2\epsilon};\epsilon^2E_0\right)
 =V^{\N=2}_{1,2}(\mathbf b;E_0)-\frac{\pi^2E_0}{12}.
 \label{eq:full-mixed-continuum}
\end{equation}
Here the parity sectors are evaluated as polynomials in the same formal boundary lengths, and the extra term follows from~\eqref{eq:app-color-failure12} and $\epsilon^2Z_1\to\pi^2/24$. So the full amplitudes cannot be identified with the super WP volumes, even though their lower-edge disk curves agree.

The opposite limit is simpler. At fixed gap, $q\to0$ sets every $Z_m$ to zero in the local curves. The upper curve then becomes the Airy curve, whose volumes are $\psi$-class intersection numbers, while the lower curve becomes $x_{\rm loc}=w^2/2$, $y_{\rm loc}=w/(\Ezero-w^2)$, which gives the deformed negative $2$-spin class of Section~\ref{subsec:it-spin}. At zero gap, the latter reduces to the signed Bessel curve $y_{\rm loc}=-1/w$, with $\Theta$-class intersections.

%-------------------------------------------------
\section{The string equation and the ODE method}
\label{sec:ODE}

The local volumes can also be computed in a different way, using the string equation of the matrix model, an approach known as the ordinary differential equation (ODE) method~\cite{Banks:1990df,Johnson:2024bue,Ahmed:2025extended,JohnsonRodrigues:2026}. In this section, we follow the presentation of~\cite[Section V.1]{Johnson:2026boundary}, specify the finite-$q$ data of the string equation, and work out one example. The derivation can be found in Appendix~\ref{app:ODEdata}.

Consider the auxiliary Schr\"odinger operator $\mathcal H=-\hbar^2\partial_\xi^2+u(\xi)$, where $\xi$ is its coordinate and $u(\xi)$ is the potential, and primes denote $\xi$-derivatives. Let $R_k[u]$ be the Gel'fand--Dikii differential polynomials in $u$~\cite{Gelfand:1975rn}, normalized as $R_k[u]=u^k+O(\hbar^2)$. The string equation is then given by
\begin{equation}
 u\R^2-\frac{\hbar^2}{2}\R\R''+\frac{\hbar^2}{4}(\R')^2
 =(\widetilde\Gamma_q^{\rm str})^2,
 \qquad
 \R=\sum_{k\geq1}t_kR_k[u]+\xi.
\label{eq:string-equation}
\end{equation}
The times $t_k$, the Fermi level $\mu$ (the endpoint of the $\xi$-integration in~\eqref{eq:WKB-density}) and the signed coupling $\widetilde\Gamma_q^{\rm str}=\hbar\Gamma_q$ are all fixed by the lower local density~\eqref{eq:local-density}. In order to write them down, let us define
\begin{equation}
 \begin{aligned}
 \mathcal F_q(v)&=\prod_{j\geq1}\left(1+\frac{16q^{2j}v}{(1-q^{2j})^2}\right)
                 =\sum_{n\geq0}\mathfrak f_n v^n,\\
 \mathcal T_q(v)&=\sum_{n\geq0}\frac{(2n+1)!!}{2^n n!}\,\mathfrak f_n(-v)^n.
 \end{aligned}
 \label{eq:string-moment-data}
\end{equation}
For $\Ezero>0$, Abel inversion then gives the specialization
\begin{equation}
 \begin{aligned}
 t_k(\Ezero,q)&=\left.\frac{(-1)^k}{(2k+1)k!}
                         \frac{d^k\mathcal T_q(v)}{dv^k}\right|_{v=\Ezero},\qquad k\geq1,\\
 \mu&=\mathcal T_q(\Ezero),\qquad
 \widetilde\Gamma_q^{\rm str}=\sqrt{\Ezero}\,\mathcal F_q(-\Ezero).
 \end{aligned}
 \label{eq:string-specialization}
\end{equation}
Here the derivatives are taken at fixed $q$, and we hold the signed coupling fixed as $\hbar\to0$. We relate this density convention to the continuum normalization of~\cite{Johnson:2026boundary} in Appendix~\ref{subsec:ODE-continuum-data}.

The classical potential $u_0$, i.e., the $\hbar\to0$ limit of $u$, lies on the branch
\begin{equation}
 \xi=-\mathcal G(u_0),\qquad
 \mathcal G(u)=\sum_{k\geq1}t_ku^k-\frac{\widetilde\Gamma_q^{\rm str}}{\sqrt u},
 \qquad u_0(\mu)=\Ezero.
 \label{eq:string-classical-data}
\end{equation}
Differentiating this inverse relation with the background parameters held fixed, we obtain the endpoint data that we need for the volumes, e.g., $u_0'(\mu)=-2\Ezero$ and $u_0''(\mu)=6\Ezero-96Z_1\Ezero^2$. The universal genus-one formula, which follows from the resolvent of Appendix~\ref{subsec:ODE-resolvent}, then gives
\begin{equation}
 V^q_{1,1}(b;\Ezero)
 =\left.\left(\frac{u_0''}{24u_0'}-\frac{b^2u_0'}{96}\right)\right|_{\xi=\mu}
 =-\frac18+\frac{\Ezero}{48}(b^2+96Z_1),
 \label{eq:string-V11-check}
\end{equation}
which reproduces~\eqref{eq:V11-local}. Volumes with more boundaries are obtained by inserting the additional boundaries before evaluating at $\xi=\mu$. We give the resolvent prescription, the normalized loop operator and further checks, including the finite hard-edge limits, in Appendices~\ref{subsec:ODE-resolvent} and~\ref{sec:ODE-adding-boundaries}.

%-------------------------------------------------
\section{The BPS information in the pole}
\label{sec:BPS}
%-------------------------------------------------

In this section, we show that the global curve connects the continuum and BPS parts of the SYK spectrum. Recall that the continuum density came from the positive moments, which cannot see the states at zero energy. Nevertheless, at generic charge, its analytic continuation has a pole whose residue gives a BPS weight. As we will see, this relation involves a charge shift, i.e., the BPS states belong to an adjacent charge sector. Throughout this section, $E$ is the microscopic SYK energy, and the derivation can be found in Appendix~\ref{app:SYK-BPS}.

Let $r$ be the average charge of a non-BPS multiplet. BPS states occur only at state charges $|s|<\tfrac12$, and we will call this range the \emph{BPS window} (see Appendix~\ref{app:SYK-BPS}). For $0<|r|<1$, exactly one of the two component charges lies in this window, namely
\begin{equation}
 s_r=r-\tfrac12\operatorname{sgn}r.
\label{eq:kr}
\end{equation}
The BPS states at $s_r$ are not part of the positive-energy multiplet. Their weight was computed in~\cite{Berkooz:2020uly}, and in terms of the Jacobi theta function $\vartheta_1$ of~\eqref{eq:syk-theta-definition}, whose second argument is the nome, it is given by
\begin{equation}
 \rho_{\rm BPS}(s_r)=\vartheta_1\!\left(\pi|r|,\e^{-\pi^2/\lambda}\right).
\label{eq:BPS-deg}
\end{equation}
If we continue the microscopic density~\eqref{eq:syk-multiplet-density} to $\varphi=i\lambda|r|$ and apply the Jacobi imaginary transformation, we find
\begin{equation}
 \mathop{\rm Res}_{E=0}\rho_r(E)=\frac{i}{2\pi}\rho_{\rm BPS}(s_r),
\label{eq:BPS-residue}
\end{equation}
with the opposite sign on the other sheet. This is the finite-$q$ counterpart of the corresponding relation in $\mathcal N=2$ JT supergravity~\cite{Turiaci:2023jfa,Johnson:2023ofr,JohnsonUsatyuk:2024}. At nonzero integer $|r|$, the residue vanishes and the pole cancels, while at $r=0$ zero energy becomes a hard edge.

After dividing out $\mathcal W_r$ and adopting the continuum normalization of Appendix~\ref{app:SYK-continuum}, we obtain the global coupling
\begin{equation}
 \widetilde\Gamma_q^{\rm glob}(\Ezero)=
 \frac{\sqrt{\Ezero(1+\Ezero)}}{2\pi(1-q)}
 \prod_{n=1}^{\infty}
 \left[1-\frac{16q^{2n}\Ezero(1+\Ezero)}{(1-q^{2n})^2}\right].
\label{eq:Gamma-global}
\end{equation}
The lower local curve replaces the factor $1+\Ezero$, which is inherited from the exact global coordinate, by $1$, and instead has
\begin{equation}
 \widetilde\Gamma_q^{\rm loc}(\Ezero)=
 \frac{\sqrt{\Ezero}}{2\pi(1-q)}
 \prod_{n=1}^{\infty}
 \left[1-\frac{16q^{2n}\Ezero}{(1-q^{2n})^2}\right].
\label{eq:Gamma-top}
\end{equation}
Here the label $\mathrm{loc}$ refers to the lower local curve, as in $x_{\rm loc}$ and $y_{\rm loc}$. Both couplings use the same density convention, and we relate them to the normalization of the string equation in Appendix~\ref{subsec:ODE-continuum-data}, especially in~\eqref{eq:ODE-Gamma-normalization}. Although the two couplings are different at finite $q$, they have the same limit,
\begin{equation}
 \lim_{q\to1^-}
 \widetilde\Gamma_q^{\rm glob}((1-q)^2E_0)
 =\lim_{q\to1^-}
 \widetilde\Gamma_q^{\rm loc}((1-q)^2E_0)
 =\frac{\sin(2\pi\sqrt{E_0})}{4\pi^2}.
\label{eq:Gamma-limit}
\end{equation}
This is the $\mathcal N=2$ JT coupling~\cite{Johnson:2023ofr,Ahmed:2025extended,Johnson:2025power}. The global coupling also fixes the coefficient $\Lambda_{q,c}$ of the logarithmic term in the matrix potential, $\Lambda_{q,c}=2\pi(1-q)\widetilde\Gamma_q^{\rm glob}$, as we show in Appendix~\ref{app:potential-limits}. Both couplings can also be written in terms of theta functions, and we give these forms in Appendix~\ref{app:theta}.

%-------------------------------------------------
\section{An intersection formula for the local volumes}
\label{sec:intersection}
%-------------------------------------------------

We now turn to an intersection-theoretic interpretation of the local volumes. Following the description of the $\mathcal N=2$ super WP volumes in~\cite[Section V.3]{Johnson:2026boundary}, we regard the local theory as an Airy--Bessel spin class with a gravitational dressing of $\kappa$-class insertions. Recall that for ordinary WP volumes this dressing is $\exp(2\pi^2\kappa_1)$, which we will call the Wolpert dressing, while the $\mathcal N=2$ theory requires higher $\kappa$-classes as well. In this section, we determine the finite-$q$ dressing from the lower local curve, and explain its role in the volumes and in boundary addition.

Let us first fix our conventions. We take $n\geq1$ and $d=2g-2+n>0$, and write $D=3g-3+n$ for the complex dimension of $\overline{\mathcal M}_{g,n}$, the moduli space of stable marked curves. The class $\psi_i$ is the first Chern class of the cotangent line at the $i$th marking, and a $\kappa$-class is obtained by inserting a power of $\psi$ at one extra marking and then forgetting that marking,
\begin{equation}
 \kappa_m=\mathsf f_*(\psi_{n+1}^{m+1}),\qquad
 \mathsf f:\overline{\mathcal M}_{g,n+1}\longrightarrow
       \overline{\mathcal M}_{g,n},\qquad
 \kappa_0=d.
 \label{eq:it-kappa-convention}
\end{equation}
We use complex cohomological degree, so that $\deg\psi_i=1$ and $\deg\kappa_m=m$, and all the integrals below select the part of degree $D$.

\subsection{Spectral moments and the intersection class}
\label{subsec:it-class}

Using $\mathcal F_q$ from~\eqref{eq:string-moment-data}, we can write the lower local curve~\eqref{eq:local-curve} as
\begin{equation}
 y_{\rm loc}(w)=\frac{w}{\Ezero}\,\frac{\mathcal F_q(-w^2)}{1-w^2/\Ezero},\qquad
 \mathcal F_q(-v)=\exp\!\left[-\sum_{m\geq1}\frac{16^mZ_m}{m}v^m\right]
       =\sum_{j\geq0}(-1)^j\mathfrak f_jv^j.
 \label{eq:it-spectral-factor}
\end{equation}
For $\Ezero\neq0$, we define the normalized moments $a_m$ and their logarithm, with coefficients $\sigma_m$, by
\begin{align}
 a_m(\Ezero,q)&=(2m+1)!!\sum_{j=0}^m (-1)^j\mathfrak f_j\Ezero^{-(m-j)},\qquad a_0=1,
 \label{eq:it-moments}\\
 \sum_{m\geq0}a_m(\Ezero,q)v^m&=\exp\!\left[-\sum_{m\geq1}\sigma_m(\Ezero,q)v^m\right].
 \label{eq:it-log-moments}
\end{align}
The moments $a_m$ encode the Taylor coefficients of the normalized curve, and their logarithm gives the $\kappa$-couplings $\sigma_m$, as in~\cite{Johnson:2026boundary}. Since only degree $D$ contributes at each topology, we need only finitely many terms of these formal series. The volume formula is then given by
\begin{equation}
 \begin{aligned}
 V^q_{g,n}(\mathbf b;\Ezero)&=
 \int_{\overline{\mathcal M}_{g,n}}
 \Omega^q_{g,n}(\Ezero)
 \exp\!\left(\frac12\sum_{i=1}^n b_i^2\psi_i\right),\\
 \Omega^q_{g,n}(\Ezero)&=\Ezero^d
 \exp\!\left(\sum_{m\geq1}\sigma_m(\Ezero,q)\kappa_m\right).
 \end{aligned}
 \label{eq:it-volume-class}
\end{equation}
This follows from the one-branch intersection theorem~\cite{Eynard:2011intersection,DuninBarkowski:2012givental,ChekhovNorbury:2019}, and we give the conversion from descendants to $\kappa$-classes, together with the normalization, in Appendix~\ref{app:intersection-translation}. The couplings do not depend on the genus, and it is the dimension of moduli space that determines which combinations of them contribute. The prefactor $\Ezero^d$ fixes both $V^q_{0,3}=\Ezero$ and the gap factor in boundary addition.

In practice, the couplings can be computed recursively from the moments,
\begin{equation}
 \sigma_m=-a_m-\frac1m\sum_{j=1}^{m-1}j\,\sigma_j a_{m-j}.
 \label{eq:it-sigma-recurrence}
\end{equation}
For example, using $\mathfrak f_1=16Z_1$ and $\mathfrak f_2=128(Z_1^2-Z_2)$, we find
\begin{align}
 \sigma_1&=48Z_1-\frac{3}{\Ezero},
 \nonumber\\
 \sigma_2&=1920Z_2-768Z_1^2+\frac{96Z_1}{\Ezero}
                         -\frac{21}{2\Ezero^2}.
 \label{eq:it-first-couplings}
\end{align}
As a first check, using $\int_{\overline{\mathcal M}_{1,1}}\psi_1=\int_{\overline{\mathcal M}_{1,1}}\kappa_1=1/24$, we get
\begin{equation}
 V^q_{1,1}(b;\Ezero)=\frac{\Ezero}{24}\left(\sigma_1+\frac{b^2}{2}\right)
              =-\frac18+\frac{\Ezero}{48}(b^2+96Z_1).
 \label{eq:it-check11}
\end{equation}
So the constant $-1/8$ in our first example~\eqref{eq:V11-local} comes from the $-3/\Ezero$ term in $\sigma_1$. To test $\sigma_2$ as well, let us look at one more example. On $\overline{\mathcal M}_{1,2}$, we have
\[
 \int\psi_1^2=\int\psi_1\psi_2=\int\kappa_2=\frac1{24},
 \qquad
 \int\kappa_1\psi_i=\frac1{12},\qquad
 \int\kappa_1^2=\frac18.
\]
Using these, and writing $|\mathbf b|^2=b_1^2+b_2^2$, we find
\begin{align}
 V^q_{1,2}
 &=\Ezero^2\left[
       \frac{|\mathbf b|^4}{192}+\frac{\sigma_1|\mathbf b|^2}{24}
       +\frac{\sigma_1^2}{16}+\frac{\sigma_2}{24}\right]
 \nonumber\\
 &=\frac18-\frac{\Ezero}{8}(|\mathbf b|^2+112Z_1)
   +\frac{\Ezero^2}{192}
      \left[|\mathbf b|^4+384Z_1|\mathbf b|^2+21504Z_1^2+15360Z_2\right],
 \label{eq:it-check12}
\end{align}
in agreement with Appendix~\ref{app:localvolumes}. As in~\cite{Johnson:2026boundary}, the genus-zero volumes with all boundaries shrunk to cusps determine the couplings one by one. Indeed, since $\int_{\overline{\mathcal M}_{0,m+3}}\kappa_m=1$ for $m\geq1$, the volume $V^q_{0,m+3}(\mathbf0)$ is equal to $\Ezero^{m+1}\sigma_m$ plus terms involving the lower couplings. For example, $\int_{\overline{\mathcal M}_{0,5}}\kappa_1^2=5$ gives the independent check
\begin{equation}
 V^q_{0,5}(\mathbf0;\Ezero)
 =\Ezero^3\left(\sigma_2+\frac52\sigma_1^2\right)
 =12\Ezero-624Z_1\Ezero^2+(4992Z_1^2+1920Z_2)\Ezero^3.
 \label{eq:it-check05}
\end{equation}

\subsection{The Airy--Bessel class and its dressing}
\label{subsec:it-spin}

To understand the structure of this class, we start at $q=0$, where $\mathcal F_q=1$ and the local curve becomes the Airy--Bessel curve, which is the starting point of the continuum construction in~\cite[Section V.3]{Johnson:2026boundary}. The corresponding class is the deformed negative $2$-spin class, which is constructed from square roots of the anticanonical bundle with prescribed twists at the markings~\cite{Chidambaram:2022cqc}. In order to write it in terms of $\kappa$-classes, we define $\mathbb K_j=\mathbb K_j(g,n)$ by
\begin{equation}
 \begin{aligned}
 \mathbb K(v)&=\sum_{j\geq0}\mathbb K_jv^j
       =\exp\!\left(\sum_{m\geq1}k_m\kappa_m v^m\right),\\
 \e^{-\sum_{m\geq1}k_mv^m}
       &=\sum_{m\geq0}(-1)^m(2m+1)!!\,v^m.
 \end{aligned}
 \label{eq:it-K-class}
\end{equation}
The negative $2$-spin reconstruction of~\cite{Chidambaram:2022cqc} shows that, for $n\geq1$,
\begin{equation}
 \mathbb K_j=0\quad(j>d),\qquad \mathbb K_d=\Theta_{g,n},
 \label{eq:it-K-relations}
\end{equation}
where $\Theta_{g,n}$ is Norbury's spin-derived class, normalized by $\int_{\overline{\mathcal M}_{1,1}}\Theta_{1,1}=1/8$. The first few classes are $\mathbb K_0=1$, $\mathbb K_1=3\kappa_1$ and $\mathbb K_2=\tfrac92\kappa_1^2-\tfrac{21}{2}\kappa_2$. At $q=0$, the couplings are $\sigma_m=k_m(-1/\Ezero)^m$, and the resulting Airy--Bessel class is the polynomial
\begin{equation}
 \Omega^{\rm AB}_{g,n}(\Ezero)
 =\Ezero^d\mathbb K(-1/\Ezero)
 =\sum_{j=0}^{d}(-1)^j\Ezero^{d-j}\mathbb K_j.
 \label{eq:it-AB-class}
\end{equation}
In the normalization of~\cite{Chidambaram:2022cqc}, this is $(-1)^d\Theta^{2,\sqrt{\Ezero}}_{g,n}$, where $\Theta^{2,\sqrt{\Ezero}}_{g,n}$ is their deformed negative $2$-spin class with deformation parameter $\sqrt{\Ezero}$.

When we restore $q$, the moments change, and the Airy--Bessel class gets dressed,
\begin{equation}
 \Omega^q_{g,n}(\Ezero)
 =\Omega^{\rm AB}_{g,n}(\Ezero)
   \exp\!\left(\sum_{m\geq1}\tilde\sigma_m(\Ezero,q)\kappa_m\right),
 \qquad
 \tilde\sigma_m=\sigma_m-k_m(-1/\Ezero)^m.
 \label{eq:it-dressing}
\end{equation}
In particular, the first two dressing couplings are
\begin{equation}
 \tilde\sigma_1=48Z_1,\qquad
 \tilde\sigma_2=1920Z_2-768Z_1^2+\frac{96Z_1}{\Ezero}.
 \label{eq:it-mixed-dressing}
\end{equation}
The term $96Z_1/\Ezero$ shows that this dressing depends on the gap. In particular, it cannot be obtained by simply multiplying the Airy--Bessel class by the bosonic factor of the upper edge, which would be a natural first guess. To see this explicitly, let $\sigma^{\rm up}_m(q)=\lim_{\Ezero\to\infty}\sigma_m(\Ezero,q)$ be the upper-edge couplings, which are determined by
\[
 \e^{-\sum_{m\geq1}\sigma^{\rm up}_m(q)v^m}
     =\sum_{m\geq0}(2m+1)!!\,(-1)^m\mathfrak f_m v^m.
\]
If we replace $\tilde\sigma_m(\Ezero,q)$ by $\sigma^{\rm up}_m(q)$, we find a mismatch already at $(1,2)$,
\begin{equation}
 V^q_{1,2}-V^{\rm AB\times up}_{1,2}
 =\Ezero^2\frac{96Z_1/\Ezero}{24}=4\Ezero Z_1.
 \label{eq:it-factorization-obstruction}
\end{equation}
This is the finite-$q$ version of the failure of the simple Wolpert dressing found in~\cite[Section V.3]{Johnson:2026boundary}. The $(1,1)$ volume cannot detect the missing $\kappa_2$ coupling, whereas the $(1,2)$ volume can. However, the bosonic factor does give the coefficient of the highest power of the gap, in agreement with~\eqref{eq:top-gap},
\begin{equation}
 [\Ezero^d]\,V^q_{g,n}(\mathbf b;\Ezero)=
 \int_{\overline{\mathcal M}_{g,n}}
 \e^{\sum_{m\geq1}\sigma^{\rm up}_m(q)\kappa_m+\frac12\sum_i b_i^2\psi_i}.
 \label{eq:it-top-coefficient}
\end{equation}

Since $\tilde\sigma_m$ contains negative powers of the gap, the limit $\Ezero=0$ in~\eqref{eq:it-dressing} has to be taken after the terms have been combined and integrated. By~\eqref{eq:gap-polynomial}, the amplitudes are polynomial in the gap, so the apparent poles cancel in this process. In the next subsection, we instead obtain the geometric class at zero gap directly from the irregular local curve.

\subsection{Closing the gap}
\label{subsec:it-zero-gap}

At zero gap, the leading behavior of the lower curve changes, and the curve becomes $y_{\rm loc}(w)=-\mathcal F_q(-w^2)/w$. Because of the $1/w$ pole, we now need the Bessel version of the intersection theorem, which inserts Norbury's $\Theta$-class in place of the ordinary descendant intersections~\cite{ChekhovNorbury:2019,Norbury:2020enumerative}. The remaining $\kappa$-insertions are of the type studied in~\cite{HuangLiuXu:2025}, and their coefficients are determined by
\begin{equation}
 \e^{-\sum_{m\geq1}\tau_m(q)v^m}
 =\sum_{m\geq0}(2m-1)!!\,(-1)^m\mathfrak f_mv^m,\qquad (-1)!!=1.
 \label{eq:it-tau}
\end{equation}
So the hard-edge measure consists of the $\Theta$-class together with its own $q$-dependent $\kappa$-dressing,
\begin{equation}
 V^q_{g,n}(\mathbf b;0)
 =(-1)^n\int_{\overline{\mathcal M}_{g,n}}\Theta_{g,n}
       \e^{\sum_{m\geq1}\tau_m(q)\kappa_m
          +\frac12\sum_i b_i^2\psi_i}.
 \label{eq:it-theta-volume}
\end{equation}
The shift from $(2m+1)!!$ to $(2m-1)!!$ has a geometric origin. Since $\Theta_{g,n+1}=\psi_{n+1}\mathsf f^*\Theta_{g,n}$, we have $\mathsf f_*(\Theta_{g,n+1}\psi_{n+1}^{m}) =\Theta_{g,n}\kappa_m$, so we need one fewer power of the auxiliary cotangent class than in the regular translation. The sign $(-1)^d=(-1)^n$ comes from the leading coefficient $-1$ of the irregular $y$. The first two couplings are
\begin{equation}
 \tau_1=16Z_1,\qquad \tau_2=384Z_2-256Z_1^2.
 \label{eq:it-tau-values}
\end{equation}
Counting degrees also explains why the gapless answers are so simple. The class $\Theta_{g,n}$ already has degree $d$, which leaves only $D-d=g-1$ for all the boundary and $\kappa$-insertions. So the genus-zero volumes vanish, and at genus one there is no room for any dependence on the lengths or on $q$. More generally, this gives the weight $2(g-1)$ in~\eqref{eq:homogeneity}. For example, $\int\Theta_{2,1}\psi_1=3/128$ and $\int\Theta_{2,1}\kappa_1=9/128$ give
\begin{equation}
 V^q_{2,1}(b;0)
 =-\frac{3b^2}{256}-\frac{9\tau_1}{128}
 =-\frac3{256}(b^2+96Z_1).
 \label{eq:it-theta-check}
\end{equation}
Here the second intersection number follows by forgetting a marked point from $\Theta_{2,2}$, where the dilaton relation gives $\int\Theta_{2,2}\psi_2=3\int\Theta_{2,1}\psi_1$.

\subsection{Boundary addition and the continuum limit}
\label{subsec:it-forgetful}

Boundary addition, which is the guiding operation in~\cite{Johnson:2026boundary}, now has a geometric description in terms of the forgetful map. It corresponds to the loop insertion of the string equation, which we discuss in Appendix~\ref{sec:ODE-adding-boundaries}. At nonzero gap, let us define
\begin{equation}
 \mathcal Q_b(\kappa)=\sum_{m\geq0}\left([v^{m+1}]\exp\!\left[
                 \frac{b^2v}{2}+\sum_{\ell\geq1}\sigma_\ell(\Ezero,q)v^\ell
                          \right]\right)\kappa_m .
 \label{eq:it-boundary-insertion}
\end{equation}
The sum truncates under integration. By adding a marking and applying the forgetful map, we find the exact identity
\begin{align}
 V^q_{g,n+1}(\mathbf b,b;\Ezero)
 ={}&\Ezero\sum_{i=1}^n\int_0^{b_i}\ell\,
       V^q_{g,n}(b_1,\ldots,\ell,\ldots,b_n;\Ezero)\,d\ell
 \nonumber\\
 &+\Ezero\int_{\overline{\mathcal M}_{g,n}}
       \Omega^q_{g,n}(\Ezero)
       \e^{\frac12\sum_i b_i^2\psi_i}\mathcal Q_b(\kappa).
 \label{eq:it-boundary-addition}
\end{align}
The second line consists of the $\kappa$-decorated volumes of~\cite{Johnson:2026boundary}, i.e., the original intersection integral with an extra $\kappa_m$ insertion, weighted by $\Ezero$ times its coefficient in~\eqref{eq:it-boundary-insertion}.

The first line comes from contracting a three-pointed sphere that carries the new marking and one of the old ones, and the remaining contribution gives the second line. Since adding a marking raises $d$ by one, both lines come with a factor of $\Ezero$. We derive both terms in Appendix~\ref{app:intersection-boundary}. In this way, boundary addition probes the $\kappa$-couplings through these decorated volumes.

Finally, let us compare with the continuum, which is straightforward at the level of the moments. With $\epsilon=1-q$ and $\Ezero=\epsilon^2E_0$, equations~\eqref{eq:local-scaling} and~\eqref{eq:qzeta-limit} give
\begin{equation}
 \lim_{q\to1^-}\epsilon^{2m}a_m(\epsilon^2E_0,q)
 =a_m^{\rm cl}(E_0)
 =(2m+1)!!\sum_{j=0}^m
       \frac{(-1)^j(2\pi)^{2j}}{(2j+1)!E_0^{m-j}}.
 \label{eq:it-classical-moments}
\end{equation}
Together with $\e^{-\sum_m\sigma_m^{\rm cl}v^m}=\sum_m a_m^{\rm cl}v^m$, this is the continuum prescription of~\cite[Eqs.\ (310)--(313)]{Johnson:2026boundary}. For example, the first two couplings are
\[
 \sigma_1^{\rm cl}=2\pi^2-\frac3{E_0},\qquad
 \sigma_2^{\rm cl}=\frac{4\pi^2}{E_0}-\frac{21}{2E_0^2}.
\]
The discrepancy~\eqref{eq:it-factorization-obstruction} then becomes $\pi^2E_0/6$, in agreement with Eq.\ (296) of that reference. For the zero-gap class, the corresponding moments instead sum to
\[
 \sum_{m\geq0}(2m-1)!!\,
       \frac{(-1)^m(2\pi)^{2m}}{(2m+1)!}\,v^m
 =\int_0^1\e^{-2\pi^2vt^2}\,dt.
\]
This gives $\tau_1^{\rm cl}=2\pi^2/3$ and $\tau_2^{\rm cl}=-8\pi^4/45$, so the continuum measure at zero gap also retains higher $\kappa$-insertions.

We also note that, by~\eqref{eq:colored-identity}, the same intersection formula describes the graded discrete volumes obtained by taking every residue at the lower endpoint.

%-------------------------------------------------
\section{Discussion}
\label{sec:discussion}
%-------------------------------------------------

We have constructed an $\mathcal N=2$ discrete spectral curve from the density of DSSYK, and obtained its quasi-polynomial volumes by topological recursion. The curve has two edge theories, and we focused on the lower one, which gives the local $\mathcal N=2$ volumes. We computed these volumes at finite $q$, showed that they reduce to the Turiaci--Witten volumes~\cite{Turiaci:2023jfa} in the continuum limit, and gave them an intersection-theoretic interpretation. The full amplitudes, on the other hand, contain mixed-endpoint terms that survive the continuum limit, so they are not simply the super WP volumes. We conclude with some open questions that follow from this construction.

The most immediate question is whether the higher amplitudes of our curve agree with the connected correlators of microscopic SYK. The density fixes only the disk, so it cannot settle this question by itself. Moreover, the blocks of different charge depend on the same random couplings, and resolving the density into multiplet contributions does not imply that they are statistically independent. A useful first calculation would be the connected two-point function of resolvents at fixed average multiplet charge, followed by the same calculation between different charges, with the normalization and the double-scaled limit specified carefully. This would test our choice of two-boundary input, and it would also tell us whether conditioning or unfolding is needed to relate the microscopic ensemble to our recursive construction.

In Appendix~\ref{app:potential}, we reconstructed a matrix potential with the required planar density. It is interesting to ask whether this potential admits a nonperturbative definition with the same spectral properties. This requires a convergent integration contour and an analysis of the competing saddles, including the limit in which the zero-energy pole reaches the lower edge. Such a definition would also allow us to examine the BPS interpretation beyond the disk. The residue relation suggests a connection with the difference of matrix dimensions in a supersymmetric ensemble, but to make this identification one has to specify that ensemble and its charge sectors. One could then compare its edge limit with a nonperturbative solution of the local string equation, keeping in mind that perturbative agreement alone would not imply that the two completions coincide.

The intersection formula raises a related geometric question. We have identified classes on the ordinary moduli space of stable marked curves whose integrals give the local $\mathcal N=2$ volumes, and one may ask whether these classes arise from a measure on an independently defined space of finite-$q$ super Riemann surfaces. This would require a geometric meaning for both $q$ and the charge-dependent gap, together with a prescription for gluing boundaries, which should reproduce the two limits that we know, namely the Turiaci--Witten limit and the deformed negative $2$-spin class at $q=0$. Moreover, the mixed term in~\eqref{eq:it-mixed-dressing} shows that the finite-$q$ correction depends on the gap, so a universal gap-independent factor multiplying the undeformed spin class is not enough.

For the global curve, general topological recursion gives an intersection description with an endpoint label on each component of a stable curve~\cite{Eynard:2011intersection}, but the physical meaning of these labels is not yet clear. Recall that the mixed term~\eqref{eq:full-mixed-continuum} survives the continuum scaling, even though the two disk limits are the familiar JT curves. It would be interesting to understand whether the full amplitudes can be described by gluing the two local theories with an additional interaction, which would then be constrained by the mixed terms that we have computed. This may also clarify whether the global theory has a spacetime interpretation that retains the information lost when only one edge is considered.

It would also be useful to have a direct counting interpretation of the discrete volumes. At present, their discreteness refers only to the integer boundary lengths and the associated parity dependence, and we have not given a cell decomposition or a weighted counting problem that produces the quasi-polynomials. The bosonic lattice-point construction suggests looking for a surface or ribbon-graph model, with weights that account for the signs of the supersymmetric answers. A simpler first case may be the lower local theory at $q=0$, whose intersection class is the deformed negative $2$-spin class, and one can then ask how the counting weights should change at finite $q$.

In a companion work with Saraswat~\cite{Ahmed:2026qdeformed}, we study the $\mathcal N=1$ $q$-deformed theory. Unlike the bosonic and $\mathcal{N}=2$ cases, the $\mathcal{N}=1$ theory does not have a \emph{nice} discrete quasi-polynomial. However, it does have a local edge expansion and we use that to define the $q$-deformed super WP volumes. It is unclear whether it is a fundamental property of $\mathcal{N}=1$ theories or simply an artifact of our naive generalization of the bosonic case. We hope to address some of these questions in future work. We also hope to report progress on the $\mathcal{N}=4$ case in a future work, although the double-scaling limit of the microscopic SYK is unclear at the moment.

\Needspace{12\baselineskip}
\section*{Acknowledgments}
We thank Krishan Saraswat and Clifford Johnson for discussions. W.A. is supported by the U.S. Department of Energy under grant DE-SC 0011702. J.N. is supported by the Beijing Institute of Mathematical Sciences and Applications. This research was supported in part by grant NSF PHY-2309135 to the Kavli Institute for Theoretical Physics (KITP). J.N. would like to thank the KITP for hospitality during the Graduate Fellowship program, during which part of this research was conducted. Claude and ChatGPT have been helpful in carrying out several computations\footnote{Astonishingly, with appropriate prompting, the calculations that took us days to do could be independently computed by these AI models in minutes.}.

%-------------------------------------------------
\appendix
\addtocontents{toc}{\protect\setcounter{tocdepth}{1}}

\section{Derivation of the DSSYK spectral data}
\label{app:SYK}
%-------------------------------------------------

In this appendix, we derive the spectral input used in Sections~\ref{sec:global} and~\ref{sec:BPS}, keeping track of the normalization that relates the continuum to the BPS states. We follow the oriented-chord construction of Berkooz, Brukner, Narovlansky and Raz~\cite{Berkooz:2020uly}. We will first perform the charge projection, then resolve the chord transfer matrix spectrally, and finally separate the result into individual supermultiplets. The product representation then determines both the global curve and its zero-energy residue.

Throughout this appendix, $E$ denotes the microscopic SYK energy, $s$ the charge of a state, and $r$ the average charge of a non-BPS supermultiplet. These should be distinguished from the affine energy $\Ehat=q^{1/4}E/2$ and, in the continuum limit, from $E_c=\Ehat/(1-q)^2$. For the Pochhammer symbols, we use the convention $(v;q)_n=\prod_{j=0}^{n-1}(1-vq^j)$, where products over several arguments are understood multiplicatively.

\subsection{The microscopic normalization and charge projection}
\label{app:SYK-model}

We use the Hamiltonian, couplings and charge convention of Section~\ref{subsec:SYK-main}, with $\operatorname{tr}1=1$. Since $\operatorname{tr}(\chi_{\mathbf j}\chi_{\mathbf j}^\dagger)=2^{-p}$, this convention fixes $\langle\operatorname{tr}H\rangle=1$. The alternating-word identity~\eqref{eq:syk-alternating-words} follows from $Q^2=(Q^\dagger)^2=0$, which holds because disjoint odd fermion strings anticommute, while overlapping strings have vanishing product.

Recall that the charge operator is $\gamma=p^{-1}\sum_i(\bar\chi_i\chi_i-\tfrac12)$, so that $[\gamma,Q^\dagger]=Q^\dagger$, and a state with $N_{\rm occ}$ occupied sites has charge $s=(N_{\rm occ}-N/2)/p$. In the limit~\eqref{eq:DSSYK-limit}, Stirling's formula then gives
\begin{equation}
 \frac{1}{2^N}\binom{N}{N/2+sp}
 =ds\,\sqrt{\frac{\lambda}{\pi}}\,q^{s^2}\,[1+o(1)],
 \qquad ds=\frac1p.
\label{eq:syk-charge-measure}
\end{equation}
So $\rho(E;s)$ measures the fraction of the full Hilbert space per unit energy and per unit charge, and we keep this full-space normalization for its moments,
\begin{equation}
 m_k(s)=\int dE\,E^k\rho(E;s),\qquad
 m_0(s)=\sqrt{\frac{\lambda}{\pi}}\,q^{s^2}.
\label{eq:syk-moments-definition}
\end{equation}
In order to project onto a fixed charge, let us introduce the moments $\widetilde m_k(\theta)=\langle\operatorname{tr}(H^k\e^{-\theta\gamma})\rangle$ at chemical potential $\theta$, in terms of which
\begin{equation}
 m_k(s)=\int_{-\infty}^{\infty}\frac{d\theta}{2\pi}\,
 \e^{i\theta s}\widetilde m_k(i\theta).
\label{eq:syk-charge-projection}
\end{equation}

The Gaussian average pairs each $Q$ insertion in~\eqref{eq:syk-alternating-words} with a $Q^\dagger$ insertion. We denote these endpoints by O and X, so that a pairing is an oriented chord diagram $P$. At fixed $k$, the overlap $N_{\rm sh}$ of two random index sets tends to a Poisson variable of mean $p^2/N=\lambda/2$, while triple overlaps vanish in the limit. A pair of chords whose endpoints alternate as OXOX around the circle permits an overlap, and each shared index removes one factor $1/2$ from the normalized trace, giving a relative factor $2^{N_{\rm sh}}$. For the other cyclic order, OOXX, an overlap forces two identical fermion operators together, and the trace vanishes. We call these two kinds of pairs friends and enemies, respectively, and their averaged weights are
\begin{equation}
 \langle 2^{N_{\rm sh}}\rangle=\e^{\lambda/2}=q^{-1/2},
 \qquad \Pr(N_{\rm sh}=0)=\e^{-\lambda/2}=q^{1/2}.
\label{eq:syk-pair-weights}
\end{equation}
Note that crossing chords are enemies, and their disjoint strings acquire a minus sign on interchange because $p$ is odd.

We now cut the trace at the insertion of the chemical potential, and let $n_R$ be the number of chords whose O endpoint precedes X, and $n_L$ the number of chords with the opposite orientation, with $n_{\rm net}=n_R-n_L$. An unused fermion site contributes a factor $\cosh(\theta/2p)$, whereas a right- or left-pointing chord contributes $\e^{\pm\theta/2}$ at leading order, so that
\begin{equation*}
 [\cosh(\theta/2p)]^{N-kp}\,\e^{\theta n_{\rm net}/2}
 \longrightarrow \e^{\theta^2/(4\lambda)+\theta n_{\rm net}/2}.
\end{equation*}
The corrections from shared indices vanish, since their number is of order one at fixed $k$ while $p\to\infty$. If we let $n_{\rm fr}(P)$, $n_{\rm en}(P)$ and $n_{\rm cr}(P)$ be the numbers of friends, enemies and crossings, respectively, the resulting moments are
\begin{equation}
 \widetilde m_k(\theta)=\frac{\e^{\theta^2/(4\lambda)}}{2^k}
 \sum_{P}(-1)^{n_{\rm cr}(P)}
 q^{[n_{\rm en}(P)-n_{\rm fr}(P)]/2}\,\e^{\theta n_{\rm net}(P)/2}.
\label{eq:syk-chord-moments}
\end{equation}
Here the sum runs over the pairings of both alternating words in~\eqref{eq:syk-alternating-words}. The Gaussian integral in~\eqref{eq:syk-charge-projection} is elementary,
\begin{equation*}
 \int\frac{d\theta}{2\pi}\,
 \e^{-\theta^2/(4\lambda)+i\theta(s+n_{\rm net}/2)}
 =\sqrt{\frac{\lambda}{\pi}}\,q^{(s+n_{\rm net}/2)^2}.
\end{equation*}
It follows that
\begin{equation}
 m_k(s)=\frac{1}{2^k}\sqrt{\frac{\lambda}{\pi}}
 \sum_{P}(-1)^{n_{\rm cr}(P)}
 q^{[n_{\rm en}(P)-n_{\rm fr}(P)]/2+(s+n_{\rm net}(P)/2)^2}.
\label{eq:syk-charge-chords}
\end{equation}

\subsection{The transfer matrix and its spectral measure}
\label{app:SYK-transfer}

After the charge projection, the chord weights can be expressed locally. In the linear ordering, friends are either disjoint chords of the same orientation or nested chords of opposite orientations. Using $n_{\rm en}+n_{\rm fr}=\binom{k}{2}$ and $n_R+n_L=k$, we find
\begin{equation}
 \frac{n_{\rm en}-n_{\rm fr}}{2}=\frac{k-n_{\rm net}^2}{4}+N_{\rm loc},\qquad
 N_{\rm loc}=N_{\rm nest}^{\rm same}+N_{\rm cross}^{\rm same}
      -N_{\rm nest}^{\rm opp}.
\label{eq:syk-local-weights}
\end{equation}
To see this, note that the number of disjoint pairs with the same orientation is $\binom{n_R}{2}+\binom{n_L}{2}-N_{\rm nest}^{\rm same}-N_{\rm cross}^{\rm same}$, and substituting this gives~\eqref{eq:syk-local-weights}. The nonlocal square then cancels,
\begin{equation}
 q^{(n_{\rm en}-n_{\rm fr})/2+(s+n_{\rm net}/2)^2}=q^{s^2+k/4}\,q^{N_{\rm loc}+sn_{\rm net}}.
\label{eq:syk-local-cancellation}
\end{equation}
As a result, only chords that are open at the same time interact through $N_{\rm loc}$.

The chord operators $\mathsf Q$ and $\mathsf Q^\dagger$ act on the ordered list of open endpoints. The first one either appends an O with weight $q^s$, or removes an X endpoint with weight
\begin{equation}
 (-1)^{\#\text{ endpoints below it}}\,
 q^{\#\mathrm X-1-\#\mathrm O\text{ above it}}.
\label{eq:syk-closing-weight}
\end{equation}
Here $\#\mathrm X$ is the total number of open X endpoints, including the one being removed, while $\#\mathrm O\text{ above it}$ counts only the O endpoints above it. The endpoints below the removed one were opened later, so they cross the closing chord, while those above it contain that chord, and this accounts for the sign and for the two terms in the exponent. The rules for $\mathsf Q^\dagger$ are obtained by exchanging O and X and replacing $s$ by $-s$. In terms of the transfer matrix $\mathsf T_s=\{\mathsf Q,\mathsf Q^\dagger\}$, equation~\eqref{eq:syk-charge-chords} then becomes
\begin{equation}
 m_k(s)=2^{-k}\sqrt{\frac{\lambda}{\pi}}\,
 q^{s^2+k/4}\langle\emptyset|\mathsf T_s^k|\emptyset\rangle.
\label{eq:syk-transfer-moments}
\end{equation}

Lists with two adjacent identical endpoints are null, since exchanging which of those endpoints closes first preserves every weight except the crossing sign. After quotienting by the null subspace, we are left with alternating lists~\cite{Berkooz:2020uly}. For the odd-length states $|n\rangle=|\mathrm{X}(\mathrm{OX})^n\rangle$ with $n\geq0$, the rules above give $\mathsf Q^\dagger|n\rangle=0$ and
\begin{equation}
 \mathsf T_s|n\rangle
 =|n+1\rangle+(q^{-s}+q^{s-1})|n\rangle
  +q^{-1}(1-q^{2n})|n-1\rangle.
\label{eq:syk-tridiagonal}
\end{equation}
Here the last term is absent for $n=0$. To see where its coefficient comes from, consider $n\geq1$, for which the first action is
\begin{equation*}
 \mathsf Q|n\rangle
 =q^n|(\mathrm{OX})^n\rangle+|(\mathrm{XO})^n\rangle
  +q^s|(\mathrm{XO})^{n+1}\rangle.
\end{equation*}
Acting with $\mathsf Q^\dagger$ on the first two terms produces $-q^{2n-1}|n-1\rangle$ and $q^{-1}|n-1\rangle$, respectively, and their sum gives the last coefficient in~\eqref{eq:syk-tridiagonal}. At the boundary, we have instead $\mathsf Q|0\rangle=|\emptyset\rangle+q^s|\mathrm{XO}\rangle$.

The inner product is fixed by adjointness, and states of different lengths are orthogonal to one another. Writing $\|n\|^2=\langle n|n\rangle$, equation~\eqref{eq:syk-tridiagonal} implies $\|n\|^2=q^{-1}(1-q^{2n})\|n-1\|^2$, while the boundary relation $\mathsf Q^\dagger|\emptyset\rangle=q^{-s}|0\rangle$ fixes $\|0\|^2=q^s$. Therefore,
\begin{equation}
 \|n\|^2=q^{s-n}(q^2;q^2)_n.
\label{eq:syk-chord-norm}
\end{equation}

To diagonalize the transfer matrix, we consider a state $|\Psi_s(\varphi)\rangle=\sum_{n\geq0}\upsilon_n(\varphi)|n\rangle$, for which the eigenvalue equation reads
\begin{equation*}
 [\mathcal E-(q^{-s}+q^{s-1})]\upsilon_n
 =\upsilon_{n-1}+q^{-1}(1-q^{2n+2})\upsilon_{n+1}.
\end{equation*}
If we set $\mathcal E=q^{-s}+q^{s-1}-2q^{-1/2}\cos\varphi$ and $\upsilon_n=(-1)^nq^{n/2}h_n/(q^2;q^2)_n$, this equation reduces to
\begin{equation}
 2\cos\varphi\,h_n=h_{n+1}+(1-q^{2n})h_{n-1},
 \qquad h_{-1}=0,\quad h_0=1.
\label{eq:syk-hermite}
\end{equation}
These are the continuous $q$-Hermite polynomials in $\cos\varphi$ with base $q^2$, in the conventions of Appendix~A of~\cite{Berkooz:2020uly}, and their orthogonality measure is
\begin{equation}
\begin{gathered}
 \mathfrak w_q(\varphi)=(q^2,\e^{2i\varphi},\e^{-2i\varphi};q^2)_\infty,\\
 \int_0^\pi\frac{d\varphi}{2\pi}\,\mathfrak w_q(\varphi)
 h_n(\varphi)h_m(\varphi)
 =(q^2;q^2)_n\delta_{nm}.
\end{gathered}
\label{eq:syk-hermite-measure}
\end{equation}
In particular, $\int_0^\pi \mathfrak w_q\,d\varphi=2\pi$. Using completeness together with~\eqref{eq:syk-chord-norm}, we then find the spectrum and the normalization of the eigenstates,
\begin{equation}
\begin{aligned}
 \mathcal E_s(\varphi)
 &=2q^{-1/2}\big[\cosh\lambda(s-\tfrac12)-\cos\varphi\big],\\
 \langle\Psi_s(\varphi)|\Psi_s(\varphi')\rangle
 &=q^s\frac{2\pi\delta(\varphi-\varphi')}{\mathfrak w_q(\varphi)}.
\end{aligned}
\label{eq:syk-spectrum}
\end{equation}

The even-length eigenstates that the empty diagram sees are $|\widetilde\Psi_s\rangle=\mathsf Q|\Psi_s\rangle$, together with the analogous $\mathsf Q^\dagger$-exact states obtained by $s\to-s$. By adjointness, their squared norm acquires an extra factor of the eigenvalue,
\begin{equation}
 \langle\widetilde\Psi_s(\varphi)|\widetilde\Psi_s(\varphi')\rangle
 =q^s\mathcal E_s(\varphi)
   \frac{2\pi\delta(\varphi-\varphi')}{\mathfrak w_q(\varphi)}.
\label{eq:syk-exact-norm}
\end{equation}
The overlap with the empty diagram is one, because $\upsilon_0=1$ and $q^{-s}\|0\|^2=1$, and the two exact subspaces are orthogonal by nilpotency. Inserting their spectral resolutions, we therefore find, for $k\geq1$,
\begin{equation}
 \langle\emptyset|\mathsf T_s^k|\emptyset\rangle
 =\int_0^\pi\frac{d\varphi}{2\pi}\,\mathfrak w_q(\varphi)
 \big[q^{-s}\mathcal E_s^{k-1}+q^s\mathcal E_{-s}^{k-1}\big].
\label{eq:syk-spectral-resolution}
\end{equation}
The contribution of the zero eigenvalue is determined separately by $m_0(s)$. The power in the moments has been reduced from $k$ to $k-1$, and this is what produces the supersymmetric $1/E$ factor.

\subsection{The density of a single supermultiplet}
\label{app:SYK-density}

If we combine~\eqref{eq:syk-transfer-moments} and~\eqref{eq:syk-spectral-resolution}, the transfer eigenvalue is converted into the physical energy,
\begin{equation}
 E_r(\varphi)=\frac{q^{1/4}}{2}\mathcal E_s(\varphi)
 =q^{-1/4}\big[\cosh(\lambda r)-\cos\varphi\big],
 \qquad r=s-\tfrac12.
\label{eq:syk-energy}
\end{equation}
The other branch has $r=s+\tfrac12$, and the prefactors combine according to
\begin{equation*}
 2^{-k}q^{s^2+k/4}q^{-s}\mathcal E_s^{k-1}
 =\frac12q^{(s-1/2)^2}E_{s-1/2}^{k-1},
\end{equation*}
which gives the fixed-charge moments~\eqref{eq:syk-energy-moments}. So the compactly supported positive-energy measure is~\eqref{eq:syk-sector-density}, and its mass at zero is fixed by $m_0(s)$. Integrating the delta function with $dE_r/d\varphi=q^{-1/4}\sin\varphi$ then gives~\eqref{eq:syk-multiplet-density}. The density vanishes outside the band $q^{-1/4}[\cosh(\lambda r)-1]\leq E\leq q^{-1/4}[\cosh(\lambda r)+1]$, and both of its edges are soft for $r\ne0$. At $r=0$, on the other hand, $E\propto\varphi^2$ and $\mathfrak w_q\propto\varphi^2$, so the lower edge behaves as $E^{-1/2}$.

The multiplet interpretation follows from the supersymmetry algebra. At $E>0$, let us choose a lower component $|-\rangle\in\ker Q$ and define $|+\rangle=(2E)^{-1/2}Q^\dagger|-\rangle$. Then $Q|+\rangle=\sqrt{2E}|-\rangle$, and the two states have charges $r-\tfrac12$ and $r+\tfrac12$, so the supersymmetric pairing acts on these complementary subspaces within the adjacent charge sectors. Equation~\eqref{eq:syk-sector-density} counts one component in each sector, and summing the two components of a fixed multiplet would give $2\rho_r$.

Note that in~\eqref{eq:syk-energy-moments} the Gaussian weight and the threshold necessarily carry the same shift. This corrects the crossed pairing in equation~(5.18) of~\cite{Berkooz:2020uly}, whose prefactor should also be $\sqrt\lambda/(4\pi^{3/2})$ with the angular measure written there. Both points can be checked before using any BPS information. Since the first two angular moments in the measure $\mathfrak w_q\,d\varphi/(2\pi)$ are $1$ and $0$, we have
\begin{equation}
\begin{aligned}
 m_1(s)&=\frac12\sqrt{\frac\lambda\pi}
 \big[q^{(s-1/2)^2}+q^{(s+1/2)^2}\big],\\
 m_2(s)&=\frac{q^{-1/2}}4\sqrt{\frac\lambda\pi}
 \big[q^{(s-1)^2}+2q^{s^2}+q^{(s+1)^2}\big].
\end{aligned}
\label{eq:syk-low-moments}
\end{equation}
These agree with a direct enumeration of the one- and two-chord diagrams in~\eqref{eq:syk-charge-chords}, and they give $\int m_1(s)\,ds=1$ and $\int m_2(s)\,ds=q^{-1/2}$. For generic $s$, the second moment distinguishes our pairing from the crossed one in the published formula.

\subsection{The product, the global curve and the lower edge}
\label{app:SYK-curve}

With $z=\e^{i\varphi}$ and $A_q=(q^2;q^2)_\infty$, removing the zero-index factors from the two Pochhammer symbols gives
\begin{equation*}
 \mathfrak w_q(\varphi)=4\sin^2\varphi\,A_q
 \prod_{n\geq1}(1-q^{2n}z^2)(1-q^{2n}z^{-2}).
\end{equation*}
Each pair of factors is equal to $(1-q^{2n})^2-q^{2n}(z-z^{-1})^2$, and factoring out the squares, we get
\begin{equation}
 \frac{\mathfrak w_q(\varphi)}{\sin\varphi}
 =4A_q^3\sin\varphi\,\Pq(z),\qquad
 \log\Pq(z)=-\sum_{m\geq1}\frac{Z_m}{m}(z-z^{-1})^{2m}.
\label{eq:syk-product}
\end{equation}
The logarithmic series follows by expanding each $\log(1-X)$ and summing over the product index, and its coefficients are the $q$-zeta values $Z_m=\zh{2m}$. The series converges for $|z-z^{-1}|<2\sinh\lambda$, and outside this neighborhood the infinite product~\eqref{eq:Pq-global} supplies the continuation.

With $c=\cosh(\lambda r)$, the affine energy coordinate and the gap satisfy
\begin{equation}
\begin{gathered}
 \Ehat=\frac{q^{1/4}E}{2}=\frac{2c-x}{4},\qquad
 \Ezero=\frac{c-1}{2}=\sinh^2\!\frac{\lambda r}{2},\\
 \Ehat-\Ezero=\sin^2\!\frac\varphi2,\qquad
 1+\Ezero-\Ehat=\cos^2\!\frac\varphi2.
\end{gathered}
\label{eq:syk-energy-coordinates}
\end{equation}
Using the energy Jacobian and~\eqref{eq:syk-product}, we then obtain the normalized global density~\eqref{eq:syk-global-density} used in Section~\ref{subsec:global-curve}.

We can also check the same normalization directly at the lower edge. With $w$ as in~\eqref{eq:lower-w}, we have $z-z^{-1}=4w\sqrt{1+w^2}$ on the local sheet, and the exact relations
\begin{equation}
\begin{aligned}
 x&=2+4w^2,\qquad \Ehat=\Ezero-w^2,\\
 y\,dx&=\frac{w^2\sqrt{1+w^2}}{\Ezero-w^2}
 \prod_{n\geq1}\left[1-\frac{16q^{2n}w^2(1+w^2)}{(1-q^{2n})^2}\right]dw.
\end{aligned}
\label{eq:syk-exact-edge}
\end{equation}
If we keep only the terms of highest weight, $1+w^2$ is replaced by $1$, while $\Ezero-w^2$ is retained in full, so that
\begin{equation}
 y\,dx\longrightarrow
 \frac{w^2}{\Ezero-w^2}
 \exp\!\left[-\sum_{m\geq1}\frac{Z_m}{m}(4w)^{2m}\right]dw
 =y_{\rm loc}\,dx_{\rm loc}.
\label{eq:syk-edge-normalization}
\end{equation}
This derives the lower local curve~\eqref{eq:local-curve} with $x_{\rm loc}=w^2/2$, including its normalization. On the cut, we have $w=i\sqrt{\Ehat-\Ezero}$, and the same one-form convention then gives the local density~\eqref{eq:local-density}, or equivalently
\begin{equation}
 \rho_0^q(\Ehat)=\frac{\sqrt{\Ehat-\Ezero}}{2\pi\hbar\Ehat}
 \prod_{n\geq1}\left[1+\frac{16q^{2n}(\Ehat-\Ezero)}{(1-q^{2n})^2}\right].
\label{eq:syk-local-product}
\end{equation}

\subsection{The BPS contribution and the zero-energy residue}
\label{app:SYK-BPS}

Let us first recall the Jacobi triple product and the Fourier series of $\vartheta_1$, with $q$ as the theta nome,
\begin{equation}
\begin{aligned}
 \vartheta_1(v,q)
 &=2q^{1/4}\sin v\,(q^2,q^2\e^{2iv},q^2\e^{-2iv};q^2)_\infty\\
 &=2\sum_{n\geq0}(-1)^n q^{(n+1/2)^2}\sin((2n+1)v).
\end{aligned}
\label{eq:syk-theta-definition}
\end{equation}
Including the zero-index factors in $\mathfrak w_q$, it follows that
\begin{equation}
\begin{aligned}
 \mathfrak w_q(\varphi)&=2q^{-1/4}\sin\varphi\,\vartheta_1(\varphi,q),\\
 \rho_r(E)&=\frac{\sqrt\lambda}{2\pi^{3/2}}
 q^{r^2}\frac{\vartheta_1(\varphi,q)}{E}.
\end{aligned}
\label{eq:syk-theta-density}
\end{equation}
This representation allows us to determine the mass at $E=0$, and hence the normalized BPS weight, directly from the normalization.

Let $m_{\rm br}(r)=\int dE\,\rho_r(E)$ be the mass of the branch $r$. Changing variables from $E$ to $\varphi$ in~\eqref{eq:syk-multiplet-density}, we have
\begin{equation}
 m_{\rm br}(r)=\frac{\sqrt\lambda\, q^{1/4}}{4\pi^{3/2}}\,q^{r^2}
 \int_0^\pi\frac{\mathfrak w_q(\varphi)}{\cosh(\lambda r)-\cos\varphi}\,d\varphi.
\label{eq:syk-mass-integral}
\end{equation}
To evaluate this integral, we expand $\sinh v/(\cosh v-\cos\varphi)=1+2\sum_{j\geq1}\e^{-jv}\cos(j\varphi)$ for $v>0$. Together with $2\sin\varphi\sin((2n+1)\varphi)=\cos(2n\varphi)-\cos((2n+2)\varphi)$, this gives
\begin{equation*}
 \int_0^\pi\frac{\sin\varphi\sin((2n+1)\varphi)}{\cosh v-\cos\varphi}\,d\varphi
 =\frac{\pi(\e^{-2nv}-\e^{-(2n+2)v})}{2\sinh v}
 =\pi\e^{-(2n+1)v}.
\end{equation*}
Setting $v=\lambda|r|$ in this integral and inserting the Fourier series~\eqref{eq:syk-theta-definition} into~\eqref{eq:syk-mass-integral}, we complete the square in the exponent and find
\begin{equation}
 m_{\rm br}(r)=\sqrt{\frac\lambda\pi}
 \sum_{n\geq0}(-1)^nq^{(|r|+n+1/2)^2}.
\label{eq:syk-branch-mass}
\end{equation}
The value at $r=0$ follows by continuity, since the hard-edge singularity is integrable.

The BPS weight is then fixed by the normalization condition $\rho_{\rm BPS}(s)=m_0(s)-m_{\rm br}(s-\tfrac12)-m_{\rm br}(s+\tfrac12)$. For $s\geq\tfrac12$, the two series telescope to $m_0(s)$, and by reflection the same holds for $s\leq-\tfrac12$, so there are no BPS states in these sectors. For $|s|<\tfrac12$, on the other hand, shifting the summation indices gives
\begin{equation}
\begin{aligned}
 \rho_{\rm BPS}(s)&=\sqrt{\frac\lambda\pi}
 \left[q^{s^2}-\sum_{n\geq0}(-1)^n
 \big(q^{(n+1-s)^2}+q^{(n+1+s)^2}\big)\right]\\
 &=\sqrt{\frac\lambda\pi}\sum_{j\in\mathbb Z}(-1)^j\e^{-\lambda(j+s)^2}.
\end{aligned}
\label{eq:syk-bps-sum}
\end{equation}
Poisson summation converts the second line into $\sum_{j\in\mathbb Z}\e^{-\pi^2(j+1/2)^2/\lambda}\e^{2\pi i(j+1/2)s}$. So, if we write $\tilde q=\e^{-\pi^2/\lambda}$ for the dual nome and $\vartheta_2(v,\tilde q)=\sum_{j\in\mathbb Z}\tilde q^{(j+1/2)^2}\e^{(2j+1)iv}$, we get
\begin{equation}
 \rho_{\rm BPS}(s)=
 \begin{cases}
 \vartheta_2(\pi s,\tilde q),&|s|<\tfrac12,\\
 0,&|s|\geq\tfrac12.
 \end{cases}
\label{eq:syk-bps-density}
\end{equation}
The values at the endpoints vanish as well. This is the BPS weight per unit charge, in the trace normalization~\eqref{eq:syk-charge-measure}.

For the residue, we use the Jacobi imaginary transformation
\begin{equation}
 \vartheta_1(v,\e^{-\lambda})
 =-i\sqrt{\frac\pi\lambda}\,\e^{-v^2/\lambda}
 \vartheta_1\!\left(\frac{i\pi v}{\lambda},\tilde q\right).
\label{eq:syk-modular-transform}
\end{equation}
Substituting this into~\eqref{eq:syk-theta-density}, we find
\begin{equation}
 \rho_r(E)=\frac{q^{r^2}\e^{-\varphi^2/\lambda}}{2\pi E}
 \left[-i\vartheta_1\!\left(\frac{i\pi\varphi}{\lambda},\tilde q\right)\right].
\label{eq:syk-modular-density}
\end{equation}
For $r\neq0$, the point $E=0$ is reached on the sheet $\varphi=i\lambda|r|$. There the Gaussian factors cancel exactly, $q^{r^2}\e^{-\varphi^2/\lambda}=1$, and the oddness of $\vartheta_1$ fixes the sign, so that
\begin{equation}
 \mathop{\rm Res}_{E=0}\rho_r(E)
 =\frac{i}{2\pi}\vartheta_1(\pi|r|,\tilde q).
\label{eq:syk-residue}
\end{equation}
On the other sheet, the sign is reversed. When $0<|r|<1$, the unique component in the BPS window is $s_r=r-\tfrac12\operatorname{sgn}r$, and using the evenness of $\vartheta_2$ together with $\vartheta_2(v-\pi/2,\tilde q)=\vartheta_1(v,\tilde q)$, we find
\begin{equation}
 \rho_{\rm BPS}(s_r)=\vartheta_1(\pi|r|,\tilde q),\qquad
 \mathop{\rm Res}_{E=0}\rho_r(E)=\frac{i}{2\pi}\rho_{\rm BPS}(s_r).
\label{eq:syk-bps-residue}
\end{equation}
This proves~\eqref{eq:BPS-residue} in the stated range of charges. Equation~\eqref{eq:syk-residue} extends analytically beyond that range, but it can be interpreted as a BPS degeneracy only within the window just described. At nonzero integer $|r|$, the residue vanishes because the factor with $n=|r|$ in the product below cancels the pole. At $r=0$, the zero-energy point is the hard edge and has a branch singularity, so we treat it directly through the hard-edge density.

\subsection{Pole couplings and the continuum limit}
\label{app:SYK-continuum}

In order to obtain the coupling in the normalization of the volume problem, we continue~\eqref{eq:syk-global-density} to the sheet with $\sqrt{\Ehat-\Ezero}=i\sqrt{\Ezero}$ at zero energy, where equations~\eqref{eq:syk-energy-coordinates} imply
\begin{equation*}
 (z-z^{-1})^2=-16(\Ehat-\Ezero)(1+\Ezero-\Ehat)
 \;\longrightarrow\;16\Ezero(1+\Ezero).
\end{equation*}
The residue of the global density is therefore
\begin{equation}
 \mathop{\rm Res}_{\Ehat=0}\rho_0^{\rm glob}(\Ehat)
 =\frac{i\sqrt{\Ezero(1+\Ezero)}}{2\pi\hbar}
 \prod_{n\geq1}\left[1-\frac{16q^{2n}\Ezero(1+\Ezero)}{(1-q^{2n})^2}\right].
\label{eq:syk-global-residue}
\end{equation}

With $\epsilon=1-q$, the continuum disk convention used in~\eqref{eq:density-limit} is $\rho^{\rm cont}_q(E_c)=\epsilon\rho_0(\epsilon^2E_c)$, where $\rho_0$ stands for either $\rho_0^{\rm glob}$ or $\rho^q_0$. This convention includes a change of normalization, since the pushforward of a measure alone would come with the Jacobian $\epsilon^2$. In particular,
\begin{equation}
 \mathop{\rm Res}_{E_c=0}\rho^{\rm cont}_q(E_c)
 =\frac1\epsilon\mathop{\rm Res}_{\Ehat=0}\rho_0(\Ehat)
 =\frac{i}{\hbar}\widetilde\Gamma_q.
\label{eq:syk-rescaled-residue}
\end{equation}
In~\eqref{eq:syk-rescaled-residue}, $\widetilde\Gamma_q$ denotes the rescaled pole coupling, and for the global density this derives~\eqref{eq:Gamma-global}. Repeating the calculation with~\eqref{eq:syk-local-product} removes the factor $1+\Ezero$ and gives~\eqref{eq:Gamma-top}. For the further density normalization $\rho^{\rm J}_{0,q}=\rho^{\rm cont}_q/(2\pi)$, the residue becomes $i\widetilde\Gamma_q/(2\pi\hbar)$, which is the WKB residue of a string equation with coefficient $\widetilde\Gamma_q$, and we relate it to the unscaled string-equation coefficient in Appendix~\ref{subsec:ODE-continuum-data}. Restoring the microscopic weight also provides a normalization check,
\begin{equation}
 \mathop{\rm Res}_{E=0}\rho_r(E)
 =\mathcal W_r\frac{i\epsilon}{\hbar}\,
 \widetilde\Gamma_q^{\rm glob}(\Ezero).
\label{eq:syk-normalization-check}
\end{equation}
These one-form residues are invariant under the energy transformation in~\eqref{eq:syk-global-density}.

The theta-function forms of both couplings follow from a single product identity. Let us define, for $v\geq0$,
\begin{equation}
 \mathfrak g_q(v)=\frac{\sinh(\lambda v)}{4\pi\epsilon}
 \frac{(q^{2+2v},q^{2-2v};q^2)_\infty}{A_q^2}.
\label{eq:syk-coupling-product}
\end{equation}
Since
\begin{equation*}
 (1-q^{2n})^2-4q^{2n}\sinh^2(\lambda v)
 =(1-q^{2n+2v})(1-q^{2n-2v}),
\end{equation*}
the two couplings are given by
\begin{equation}
 \widetilde\Gamma_q^{\rm glob}=\mathfrak g_q(|r|),\qquad
 \widetilde\Gamma_q^{\rm loc}=\mathfrak g_q\big(\lambda^{-1}\operatorname{arsinh}(2\sqrt{\Ezero})\big).
\label{eq:syk-two-couplings}
\end{equation}
For the first identity, we use $\Ezero=\sinh^2(\lambda r/2)$, so that $2\sqrt{\Ezero(1+\Ezero)}=\sinh(\lambda|r|)$. Evaluating~\eqref{eq:syk-theta-definition} at the argument $i\lambda v$ and then applying~\eqref{eq:syk-modular-transform}, we find successively
\begin{equation}
\begin{aligned}
 \mathfrak g_q(v)&=\frac{\vartheta_1(i\lambda v,q)}{8\pi i\epsilon q^{1/4}A_q^3}\\
 &=\frac{\sqrt{\pi/\lambda}\,q^{-v^2-1/4}}{8\pi\epsilon A_q^3}
   \vartheta_1(\pi v,\tilde q).
\end{aligned}
\label{eq:syk-coupling-theta}
\end{equation}
This proves the formula in Appendix~\ref{app:theta} together with its global counterpart. As a check, multiplying by $\mathcal W_r i\epsilon/\hbar$ leaves $i\vartheta_1(\pi|r|,\tilde q)/(2\pi)$, as required by~\eqref{eq:syk-residue}.

Finally, let us hold $E_0$ fixed and set $\Ezero=\epsilon^2E_0$. The elementary bound
\begin{equation}
 0\leq\frac{\epsilon^2q^{2n}}{(1-q^{2n})^2}
 =\left[\frac{1-\e^{-\lambda}}{2\sinh(n\lambda)}\right]^2
 \leq\frac{1}{4n^2}
\label{eq:syk-product-bound}
\end{equation}
controls the tails of the products uniformly, and each term tends to $1/(4n^2)$. Dominated convergence then proves~\eqref{eq:qzeta-limit}, and it justifies taking the limit directly in the products over their full domain, with the logarithmic series used only within its disk of convergence. Using Euler's sine product, we then find, for either coupling,
\begin{equation}
 \widetilde\Gamma_q(\epsilon^2E_0)\longrightarrow
 \frac{\sqrt{E_0}}{2\pi}\prod_{n\geq1}\left(1-\frac{4E_0}{n^2}\right)
 =\frac{\sin(2\pi\sqrt{E_0})}{4\pi^2}.
\label{eq:syk-continuum-coupling}
\end{equation}
This is the $\mathcal N=2$ JT pole coupling of Turiaci and Witten~\cite{Turiaci:2023jfa}, in our normalization.

In order to derive~\eqref{eq:full-edge-JT} directly from the global curve, we use the coordinates~\eqref{eq:full-edge-coordinates}, in terms of which, at either edge,
\begin{equation}
 \Pq(z_\epsilon^\pm(\bar w))=\prod_{n\geq1}
 \left[1-\frac{16\epsilon^2q^{2n}\bar w^2(1+\epsilon^2\bar w^2)}
 {(1-q^{2n})^2}\right].
 \label{eq:syk-full-edge-product}
\end{equation}
Substituting these into~\eqref{eq:global-curve}, we get the exact pullbacks
\begin{equation}
 \begin{aligned}
 (z_\epsilon^+)^*(y\,dx)
 &=\epsilon\,\frac{\bar w^2\sqrt{1+\epsilon^2\bar w^2}}{E_0-\bar w^2}
       \Pq(z_\epsilon^+(\bar w))\,d\bar w,\\
 (z_\epsilon^-)^*(y\,dx)
 &=\epsilon^3\,\frac{\bar w^2\sqrt{1+\epsilon^2\bar w^2}}
       {1+\epsilon^2(E_0+\bar w^2)}\Pq(z_\epsilon^-(\bar w))\,d\bar w.
 \end{aligned}
 \label{eq:syk-full-edge-pullbacks}
\end{equation}
The bound~\eqref{eq:syk-product-bound} shows that this product tends to $\prod_{n\geq1}(1-4\bar w^2/n^2)=\sin(2\pi\bar w)/(2\pi\bar w)$ uniformly on compact sets. The limits of the one-forms hold away from their poles, and the value at $\bar w=0$ is obtained by continuation. Since $z_\epsilon^\pm(\bar w)=\pm(1+2\epsilon\bar w+O(\epsilon^2))$, the pullback of $B$ within each edge tends to $d\bar w_1\,d\bar w_2/(\bar w_1-\bar w_2)^2$. If we instead hold $\Ezero>0$ fixed, the normalizing factors for the two one-forms are $\Ezero/\epsilon^3$ and $(1+\Ezero)/\epsilon^3$, respectively, and both limits are then the bosonic one-form in~\eqref{eq:full-edge-JT}.

In the correlated scaling of the gap, the global factor $1+\Ezero$ tends to one. Similarly, for $E_c>E_0$, the local product~\eqref{eq:syk-local-product} gives
\begin{equation}
 \epsilon\rho_0^q(\epsilon^2E_c)\longrightarrow
 \frac{\sinh(2\pi\sqrt{E_c-E_0})}{4\pi^2\hbar E_c}.
\label{eq:syk-continuum-disk}
\end{equation}
This is the $\mathcal N=2$ JT disk density of Turiaci and Witten~\cite{Turiaci:2023jfa}, in our normalization. The exact global density has the same limit, because $1+\Ezero-\Ehat\to1$, and setting $E_0=0$ recovers~\eqref{eq:density-limit}.

The microscopic normalization has to be recovered separately. At fixed $r$, we have $\Ezero/\epsilon^2\to r^2/4$ and $\varphi/\lambda\to2\sqrt{E_c-r^2/4}$, and for this edge scaling the leading Fourier mode in~\eqref{eq:syk-modular-density} gives
\begin{equation}
 \rho_r(E)\,dE\sim\e^{-\pi^2/(4\lambda)}
 \frac{\sinh(2\pi\sqrt{E_c-r^2/4})}{\pi E_c}\,dE_c.
\label{eq:syk-microscopic-limit}
\end{equation}
Correspondingly, for $0<|r|<1$, the BPS weight satisfies $\rho_{\rm BPS}(s_r)\sim2\e^{-\pi^2/(4\lambda)}\sin(\pi|r|)$. The corrections from higher theta modes are of relative order $\e^{-2\pi^2/\lambda}$ at fixed scaled energy, while the changes of variables and the Gaussian prefactors also produce ordinary corrections in $\lambda$. These differences account for the different prefactors in~\eqref{eq:syk-continuum-disk} and~\eqref{eq:syk-microscopic-limit}, and the pole residues agree once the normalization~\eqref{eq:syk-normalization-check} is restored.

\section{The matrix potential and Joukowsky uniformization}
\label{app:potential}
%-------------------------------------------------

In Section~\ref{subsec:global-curve}, we obtained the global curve from its disk density. In this appendix, we reconstruct a matrix potential that produces the same one-cut planar saddle, and then recover the curve from its resolvent. The calculation follows the density-to-potential construction reviewed in~\cite[Section 2]{Okuyama:2023byh}, which originates in the matrix-model construction of~\cite{Jafferis:2022ETH}, together with the Joukowsky method of~\cite[Section 3]{Okuyama:2023byh}. Compared with the bosonic case, the additional factor $(2c-x)^{-1}$ introduces a subtraction, fixed by regularity of the physical resolvent, and, at generic charge, a logarithmic term in the potential. We will derive both by dividing the density by $2c-x$ before reconstructing the potential.

Throughout this appendix, $x$ is the global spectral coordinate, and the matrix potential $\mathcal V_{q,c}(x)$ should not be confused with the auxiliary Schr\"odinger potential $u(\xi)$ of Section~\ref{sec:ODE}. We will first take $0<q<1$ and $c>1$, so that both endpoints of the cut are soft, and treat the hard-edge limit $c\downarrow1$ separately at the end of this appendix.

\subsection{The density and its normalization}
\label{app:potential-density}

We will use the notation
\begin{equation}
 x_\ast=2c,\qquad A_q=(q^2;q^2)_\infty,\qquad
 \alpha_q=\frac{1}{4A_q^3},\qquad
 \Pi_q(x)=\prod_{n\geq1}
 \left[1-\frac{q^{2n}(x^2-4)}{(1-q^{2n})^2}\right].
 \label{eq:pot-notation}
\end{equation}
In particular, $\Pi_q(x(z))=\Pq(z)$. The product defining $\Pi_q$ converges locally uniformly in the complex $x$-plane, so it defines an entire function. The global curve has a pole at $x=x_\ast$ unless $\Pi_q(x_\ast)=0$, and for $c>1$ this cancellation happens exactly when $|r|$ is a nonzero integer.

The bosonic density of unit mass with deformation parameter $q^2$, after rescaling its support to $[-2,2]$, is given by
\begin{equation}
 \varrho_{\rm b}(x)=\frac{A_q^3}{2\pi}
 \sqrt{4-x^2}\,\Pi_q(x)
 =\frac{\mathfrak w_q(\varphi)}{2\pi\sqrt{4-x^2}},
 \qquad x=2\cos\varphi,
 \qquad \int_{-2}^2\varrho_{\rm b}(x)\,dx=1.
 \label{eq:pot-bosonic-density}
\end{equation}
Here $\mathfrak w_q$ is the angular measure in~\eqref{eq:syk-hermite-measure}, and the normalization also follows from its constant Fourier coefficient. If we transform the supersymmetric density~\eqref{eq:syk-global-density} using $\Ehat=(x_\ast-x)/4$, we get
\begin{equation}
 \varrho_{q,c}(x):=\frac{\hbar}{4}
 \rho_0^{\rm glob}\!\left(\frac{x_\ast-x}{4}\right)
 =\frac{\sqrt{4-x^2}\,\Pi_q(x)}{8\pi(x_\ast-x)}
 =\alpha_q\frac{\varrho_{\rm b}(x)}{x_\ast-x}.
 \label{eq:pot-super-density}
\end{equation}
Here the factor $1/4$ is the absolute value of the energy Jacobian, and removing the disk factor $\hbar^{-1}$ preserves the normalization of $y$ in~\eqref{eq:global-curve}. The resulting continuum measure has the filling fraction defined below, which we will keep explicit.

We define the filling fraction and the resolvent by
\begin{equation}
 \nu_{q,c}=\int_{-2}^2\varrho_{q,c}(t)\,dt,\qquad
 R_{q,c}(x)=\int_{-2}^2\frac{\varrho_{q,c}(t)}{x-t}\,dt
 =\frac{\nu_{q,c}}{x}+O(x^{-2}).
 \label{eq:pot-resolvent-convention}
\end{equation}
These are the conventions of a formal $L\times L$ Hermitian matrix integral
\begin{equation}
 \mathcal Z_{q,c}=\int dM\,
  \exp[-\hbar^{-1}\operatorname{Tr}\mathcal V_{q,c}(M)],
 \qquad L\hbar=\nu_{q,c}.
 \label{eq:pot-matrix-integral}
\end{equation}
Varying the eigenvalue action at fixed filling, we obtain the saddle-point equation
\begin{equation}
 \frac12\mathcal V'_{q,c}(x)
 =\operatorname{PV}\int_{-2}^2
       \frac{\varrho_{q,c}(t)}{x-t}\,dt,
 \qquad -2<x<2.
 \label{eq:pot-saddle}
\end{equation}
Equivalently, one could use a unit-mass convention with $\mathcal V_{\rm unit}=\mathcal V_{q,c}/\nu_{q,c}$, $\varrho_{\rm unit}=\varrho_{q,c}/\nu_{q,c}$ and weight $\exp[-L\operatorname{Tr}\mathcal V_{\rm unit}(M)]$, in which case $y$ is also divided by $\nu_{q,c}$. We will keep the original normalization, in order to match the global curve used in the main text.

\subsection{The bosonic potential as a Chebyshev series}
\label{app:potential-bosonic}

The Jacobi triple product gives a rapidly convergent Chebyshev expansion of the bosonic density. To write it, let us define
\begin{equation}
 \mathfrak d_n(q)=(-1)^{n-1}
       \big[q^{n(n-1)}+q^{n(n+1)}\big],\qquad n\geq1.
 \label{eq:pot-Chebyshev-coefficients}
\end{equation}
Then, with $T_n$ and $U_n$ the Chebyshev polynomials of the first and second kinds, we have
\begin{equation}
 \varrho_{\rm b}(x)
 =\frac{1}{\pi\sqrt{4-x^2}}
  \left[1-\sum_{n\geq1}\mathfrak d_n(q)T_{2n}(x/2)\right].
 \label{eq:pot-density-Chebyshev}
\end{equation}
Here we have combined the two terms with Fourier indices $n$ and $-n$, but kept the constant term separate, since this makes the principal-value calculation straightforward. For $-2<x<2$, we have
\begin{equation}
 \begin{aligned}
 \operatorname{PV}\int_{-2}^2
     \frac{dt}{(x-t)\sqrt{4-t^2}}&=0,\\
 \operatorname{PV}\int_{-2}^2
     \frac{T_m(t/2)\,dt}{(x-t)\sqrt{4-t^2}}
     &=-\frac{\pi}{2}U_{m-1}(x/2),\qquad m\geq1.
 \end{aligned}
 \label{eq:pot-Hilbert-Chebyshev}
\end{equation}
These identities follow by writing $t$ as twice a cosine and using the Chebyshev recurrence. The constant integral vanishes in the principal-value prescription, the $m=1$ integral is $-\pi/2$, and the recurrence generates the remaining $U_{m-1}$.

The identity $dT_{2n}(x/2)/dx=nU_{2n-1}(x/2)$ supplies the antiderivative, so substituting into the bosonic saddle equation and integrating, we find
\begin{equation}
 \mathcal V'_{\rm b}(x)=\sum_{n\geq1}\mathfrak d_n(q)U_{2n-1}(x/2),\qquad
 \mathcal V_{\rm b}(x)=\sum_{n\geq1}\frac{\mathfrak d_n(q)}{n}T_{2n}(x/2)
              +\mathrm{const}.
 \label{eq:pot-bosonic-potential}
\end{equation}
This is Okuyama's potential~\cite{Okuyama:2023byh}, with $q$ replaced by $q^2$ and the support rescaled to $[-2,2]$. The coefficients decay as $q^{n^2}$ up to factors exponential in $n$, so the potential and all its derivatives converge locally uniformly. In the algebraic limit $q\to0$, we have $\mathfrak d_1\to1$ while all the other $\mathfrak d_n\to0$, which gives the Gaussian potential $\mathcal V_{\rm b}=x^2/2+\mathrm{const}$.

\subsection{The Joukowsky map and physical-sheet regularity}
\label{app:potential-Joukowsky}

The square root on the cut plane is uniformized by
\begin{equation}
 x=z+z^{-1},\qquad \sqrt{x^2-4}=z-z^{-1},\qquad
 z=\frac{x+\sqrt{x^2-4}}2,
 \label{eq:pot-Joukowsky}
\end{equation}
where $\sqrt{x^2-4}\sim x$ at infinity selects the physical sheet $|z|>1$. The other sheet is obtained by $z\mapsto z^{-1}$, and the two sides of the cut lift to the unit circle. The expansion at $z=0$ which defines the discrete volumes in Section~\ref{sec:globalTR} lies on this other sheet, and the two conventions are related by inversion. Since $dx=(1-z^{-2})dz$, the ramification points are at $z=\pm1$. Near $z=1$, for example, we have $x-2=(z-1)^2/z$, so the apparent square-root branching in $x$ becomes a regular dependence on the local coordinate $z$. In this way, the coordinate $z$ resolves the branch cut while preserving the genuine poles of $y$.

The Chebyshev identity
\begin{equation}
 U_{2n-1}\!\left(\frac{z+z^{-1}}2\right)
 =\sum_{j=0}^{n-1}\big(z^{2j+1}+z^{-2j-1}\big)
 \label{eq:pot-Chebyshev-Laurent}
\end{equation}
turns the potential derivative into a Laurent expansion, whose coefficients telescope,
\begin{equation}
 \sum_{n=j+1}^N \mathfrak d_n(q)
 =(-1)^j q^{j(j+1)}-(-1)^Nq^{N(N+1)}.
 \label{eq:pot-telescoping}
\end{equation}
Consequently, in terms of the entire odd function
\begin{equation}
 S_q(v)=\sum_{j\geq0}(-1)^j q^{j(j+1)}v^{2j+1},
 \label{eq:pot-partial-theta}
\end{equation}
we obtain
\begin{equation}
 \mathcal V'_{\rm b}(x(z))=S_q(z)+S_q(z^{-1}),\qquad
 R_{\rm b}(x(z))=S_q(z^{-1}).
 \label{eq:pot-bosonic-resolvent}
\end{equation}
The second expression is the part that is regular on the physical sheet and vanishes at infinity, and its leading term is $1/x$, as required by unit mass. We can also check this separation directly from the Cauchy transform of~\eqref{eq:pot-density-Chebyshev},
\begin{equation}
 R_{\rm b}(x)=
 \frac{1-\sum_{n\geq1}\mathfrak d_n(q)z^{-2n}}{\sqrt{x^2-4}}
 =S_q(z^{-1}),\qquad |z|>1.
 \label{eq:pot-bosonic-Cauchy}
\end{equation}
So the Cauchy transform and the regularity on the physical sheet fix the resolvent away from the cut.

It follows that
\begin{equation}
 y_{\rm b}(x(z))=\frac12\mathcal V'_{\rm b}(x(z))-R_{\rm b}(x(z))
 =\frac12\big[S_q(z)-S_q(z^{-1})\big].
 \label{eq:pot-bosonic-y}
\end{equation}
Applying the Jacobi triple product to this antisymmetric Laurent series, we find
\begin{equation}
 \begin{aligned}
 S_q(z)-S_q(z^{-1})
 &=(z-z^{-1})\prod_{n\geq1}
       (1-q^{2n})(1-q^{2n}z^2)(1-q^{2n}z^{-2})\\
 &=A_q^3(z-z^{-1})\Pi_q(x(z)).
 \end{aligned}
 \label{eq:pot-triple-product}
\end{equation}
In particular, $y_{\rm b}(x)=A_q^3\sqrt{x^2-4}\,\Pi_q(x)/2$ on the physical sheet, with the opposite sign on the other sheet.

\subsection{The supersymmetric potential}
\label{app:potential-super}

Let us now turn to the supersymmetric density. The rational reweighting in~\eqref{eq:pot-super-density} can be handled exactly, and to do so we define
\begin{equation}
 \eta=\frac{x_\ast-\sqrt{x_\ast^2-4}}2=c-\sqrt{c^2-1},\qquad
 0<\eta<1,\qquad R_{\rm b}(x_\ast)=S_q(\eta).
 \label{eq:pot-pole-data}
\end{equation}
The partial fraction identity
\begin{equation}
 \frac{1}{(x_\ast-t)(x-t)}
 =\frac{1}{x_\ast-x}\left(\frac1{x-t}-\frac1{x_\ast-t}\right)
 \label{eq:pot-partial-fractions}
\end{equation}
then gives both the new resolvent and the filling,
\begin{equation}
 R_{q,c}(x)=\alpha_q\frac{R_{\rm b}(x)-R_{\rm b}(x_\ast)}{x_\ast-x},\qquad
 \nu_{q,c}=\alpha_qR_{\rm b}(x_\ast).
 \label{eq:pot-super-resolvent}
\end{equation}
This gives a useful check of the microscopic normalization. For $c=\cosh(\lambda r)$, we have $\eta=q^{|r|}$, and restoring the multiplet weight in~\eqref{eq:syk-global-density}, we find
\begin{equation}
 \frac{\mathcal W_r}{\hbar}\nu_{q,c}
 =\sqrt{\frac{\lambda}{\pi}}
   \sum_{j\geq0}(-1)^j q^{(|r|+j+1/2)^2},
 \label{eq:pot-microscopic-mass}
\end{equation}
which is the non-BPS continuum mass in~\eqref{eq:syk-branch-mass}. So the eigenvalue density that we use to reconstruct the potential accounts for the non-BPS continuum, while the BPS delta function is treated separately.

Since $x_\ast$ lies outside the cut, the apparent pole of this resolvent is removable,
\begin{equation}
 R_{q,c}(x_\ast)=-\alpha_q R'_{\rm b}(x_\ast)
 =\alpha_q\int_{-2}^2\frac{\varrho_{\rm b}(t)}{(x_\ast-t)^2}\,dt.
 \label{eq:pot-removable-pole}
\end{equation}
Taking the sum of its upper and lower boundary values in~\eqref{eq:pot-saddle} then determines the derivative of the potential,
\begin{equation}
 \mathcal V'_{q,c}(x)
 =\frac{\alpha_q}{x_\ast-x}\big[\mathcal V'_{\rm b}(x)-2R_{\rm b}(x_\ast)\big]
 =\frac{\alpha_q}{x_\ast-x}
    \left[\sum_{n\geq1}\mathfrak d_n(q)U_{2n-1}(x/2)-2S_q(\eta)\right].
 \label{eq:pot-super-force}
\end{equation}
Here the subtraction $2R_{\rm b}(x_\ast)$ is fixed by the physical Cauchy transform, which completes the expression for the derivative of the potential.

In order to integrate this expression, we separate the removable divided difference from the actual pole, and define
\begin{equation}
 \Lambda_{q,c}:=\alpha_q\big[\mathcal V'_{\rm b}(x_\ast)-2R_{\rm b}(x_\ast)\big]
 =2\alpha_q y_{\rm b}(x_\ast)
 =\frac{\sqrt{x_\ast^2-4}}4\Pi_q(x_\ast).
 \label{eq:pot-log-coefficient}
\end{equation}
Then
\begin{equation}
 \mathcal V'_{q,c}(x)
 =\alpha_q\frac{\mathcal V'_{\rm b}(x)-\mathcal V'_{\rm b}(x_\ast)}{x_\ast-x}
       +\frac{\Lambda_{q,c}}{x_\ast-x}.
 \label{eq:pot-force-split}
\end{equation}
The first term is entire, with value $-\alpha_q\mathcal V''_{\rm b}(x_\ast)$ at $x=x_\ast$. Choosing $\mathcal V_{q,c}(0)=0$, we then find
\begin{equation}
 \mathcal V_{q,c}(x)
 =\alpha_q\int_0^x
       \frac{\mathcal V'_{\rm b}(v)-\mathcal V'_{\rm b}(x_\ast)}{x_\ast-v}\,dv
       -\Lambda_{q,c}\log\!\left(\frac{x_\ast-x}{x_\ast}\right).
 \label{eq:pot-integrated}
\end{equation}
For real $x<x_\ast$, the logarithm is real. Its analytic continuation has a branch cut starting at $x_\ast$, and a change of branch of the logarithm changes the potential by an additive constant.

To evaluate the integral as an explicit polynomial series, we use
\begin{equation}
 U_m(x/2)=\sum_{j=0}^{\lfloor m/2\rfloor}
                 (-1)^j\binom{m-j}{j}x^{m-2j}
 \label{eq:pot-U-polynomial}
\end{equation}
and define, for $m\geq1$,
\begin{equation}
 \mathcal J_m(x;x_\ast)=
 -\sum_{j=0}^{\lfloor(m-1)/2\rfloor}(-1)^j\binom{m-j}{j}
   \sum_{\ell=0}^{m-2j-1}
       \frac{x_\ast^\ell x^{m-2j-\ell}}{m-2j-\ell}.
 \label{eq:pot-polynomial-primitive}
\end{equation}
These are the required primitives, since $\mathcal J_m(0;x_\ast)=0$ and, by the elementary difference-of-powers identity, $\partial_x\mathcal J_m=[U_m(x/2)-U_m(x_\ast/2)]/(x_\ast-x)$. Therefore,
\begin{equation}
 \mathcal V_{q,c}(x)
 =\alpha_q\sum_{n\geq1}\mathfrak d_n(q)\mathcal J_{2n-1}(x;x_\ast)
       -\Lambda_{q,c}\log\!\left(\frac{x_\ast-x}{x_\ast}\right).
 \label{eq:pot-explicit-potential}
\end{equation}
For example,
\begin{equation}
 \mathcal J_1=-x,\qquad
 \mathcal J_3=-\frac{x^3}{3}-\frac{x_\ast x^2}{2}+(2-x_\ast^2)x.
 \label{eq:pot-first-primitives}
\end{equation}
The first two terms of the entire part are therefore $\alpha_q[(1+q^2)\mathcal J_1-(q^2+q^6)\mathcal J_3]$. The termwise integration of the Chebyshev series is justified by its local uniform convergence, including at the removable point $x=x_\ast$.

\subsection{Recovering the global curve from the potential}
\label{app:potential-recovery}

We can now recover the curve from the potential~\eqref{eq:pot-super-force}. Its physical resolvent must have the discontinuity~\eqref{eq:pot-super-density} and the asymptotic behavior $\nu_{q,c}/x$, and it must be regular at $x=x_\ast$. The resolvent~\eqref{eq:pot-super-resolvent} satisfies these conditions, and it is the unique such solution. Indeed, the difference of two solutions extends across the cut and is regular elsewhere, including at $x=\pm2$ thanks to the soft-edge behavior, so the decay at infinity forces it to vanish.

Consequently, the spectral function that follows from this potential is
\begin{equation}
 \begin{aligned}
 y_{q,c}(x)
 &=\frac12\mathcal V'_{q,c}(x)-R_{q,c}(x)\\
 &=\frac{\alpha_q}{x_\ast-x}
       \left[\frac12\mathcal V'_{\rm b}(x)-R_{\rm b}(x)\right]
 =\frac{\sqrt{x^2-4}\,\Pi_q(x)}{8(x_\ast-x)}.
 \end{aligned}
 \label{eq:pot-y-from-potential}
\end{equation}
Applying the Joukowsky map in the same normalization, we then get
\begin{equation}
 x(z)=z+z^{-1},\qquad
 y_{q,c}(z)=\frac18\frac{z-z^{-1}}{2c-z-z^{-1}}\Pq(z).
 \label{eq:pot-recovered-curve}
\end{equation}
This is the global curve~\eqref{eq:global-curve}. On the upper side of the physical cut, we have $\sqrt{x^2-4}=i\sqrt{4-x^2}$ and $y_{q,c,+}(x)=i\pi\varrho_{q,c}(x)$, which checks both the sign and the normalization against the original density.

The regularity subtraction takes a simple form on the uniformizing plane,
\begin{equation}
 R_{q,c}(x(z))=
 \alpha_q\frac{S_q(z^{-1})-S_q(\eta)}{x_\ast-z-z^{-1}}.
 \label{eq:pot-uniformized-resolvent}
\end{equation}
The denominator vanishes at $z=\eta$ and at $z=\eta^{-1}$. At the physical preimage $z=\eta^{-1}$, the numerator vanishes as well, so the resolvent is regular there, while its continuation to the other sheet generally has a pole at $z=\eta$. So, at generic charge, $y$ has an off-cut pole on both sheets. If $|r|$ is a positive integer, however, we have $\Pi_q(x_\ast)=0$, so both poles of $y$ are removable and the logarithmic coefficient $\Lambda_{q,c}$ vanishes, which is the exception described in Appendix~\ref{app:SYK-BPS}. For this meromorphic potential derivative, the regularity on the physical sheet supplements the selection of negative Laurent powers near the unit circle, and it fixes the required subtraction.

For completeness, let us also evaluate the planar loop equation explicitly. With the loop source
\begin{equation}
 \mathcal S_{q,c}(x)=\int_{-2}^2\varrho_{q,c}(t)
       \frac{\mathcal V'_{q,c}(x)-\mathcal V'_{q,c}(t)}{x-t}\,dt,
 \label{eq:pot-loop-source}
\end{equation}
symmetrizing the two density integrals using the saddle equation gives $R_{q,c}^2-\mathcal V'_{q,c}R_{q,c}+\mathcal S_{q,c}=0$. Writing $\mathcal S_{\rm b}=\mathcal V'_{\rm b}R_{\rm b}-R_{\rm b}^2$ for the bosonic loop source, a direct substitution gives
\begin{equation}
 \mathcal S_{q,c}(x)
 =\frac{\alpha_q^2}{(x_\ast-x)^2}
       \big[\mathcal S_{\rm b}(x)-R_{\rm b}(x_\ast)\mathcal V'_{\rm b}(x)+R_{\rm b}(x_\ast)^2\big].
 \label{eq:pot-loop-source-explicit}
\end{equation}
It follows that
\begin{equation}
 y_{q,c}^2=\frac14(\mathcal V'_{q,c})^2-\mathcal S_{q,c}
 =\frac{(x^2-4)\Pi_q(x)^2}{64(x_\ast-x)^2}.
 \label{eq:pot-loop-curve}
\end{equation}
Here the meromorphic potential leads to a correspondingly general loop source $\mathcal S_{q,c}$, which extends the polynomial case. The density construction fixes the one-cut saddle, and the genus-zero Bergman kernel is the one used in~\eqref{eq:global-curve}, so this construction supplies the planar matrix-model input to the recursion of Section~\ref{sec:globalTR}.

\subsection{The logarithmic coefficient, the BPS pole and limiting checks}
\label{app:potential-limits}

The coefficient of the logarithm can be written in the variables of Section~\ref{sec:BPS}. Since $x_\ast=2+4\Ezero$, we have
\begin{equation}
 \begin{aligned}
 \Lambda_{q,c}
 &=\sqrt{\Ezero(1+\Ezero)}
   \prod_{n\geq1}\left[
    1-\frac{16q^{2n}\Ezero(1+\Ezero)}{(1-q^{2n})^2}\right]\\
 &=2\pi(1-q)\widetilde\Gamma_q^{\rm glob}(\Ezero).
 \end{aligned}
 \label{eq:pot-log-BPS}
\end{equation}
In particular, $\operatorname{Res}_{x=x_\ast}y_{q,c}=-\Lambda_{q,c}/2$ and $\operatorname{Res}_{x=x_\ast}\mathcal V'_{q,c}=-\Lambda_{q,c}$, while the physical resolvent is regular there. The microscopic BPS interpretation uses the charge range and the multiplet weight derived in Appendices~\ref{app:SYK-BPS} and~\ref{app:SYK-continuum}, and~\eqref{eq:pot-log-BPS} is the signed analytic continuation of this relation, which can be interpreted as a positive count of states only in the corresponding physical range of $c$. The theta-function representation is collected in Appendix~\ref{app:theta}.

As a first check, let us take $q\to0$ at fixed $c>1$ in this family of spectral curves. Then $A_q=1$ and $S_q(v)=v$, and we find
\begin{equation}
 \mathcal V_{0,c}(x)
 =-\frac{x}{4}-\frac{\sqrt{x_\ast^2-4}}4
                       \log\!\left(\frac{x_\ast-x}{x_\ast}\right),\qquad
 \nu_{0,c}=\frac{\eta}{4}.
 \label{eq:pot-q-zero}
\end{equation}
In terms of the positive energy coordinate $\Ehat=(x_\ast-x)/4$, this becomes
\begin{equation}
 \mathcal V_{0,c}
 =\Ehat-\sqrt{\Ezero(1+\Ezero)}\log\Ehat+\mathrm{const}.
 \label{eq:pot-Laguerre}
\end{equation}
So the supersymmetric rational factor turns the Gaussian potential into a linear-plus-logarithmic potential of Laguerre type. Note that in this check we hold $c$ fixed as a parameter of the curve, whereas the microscopic $q\to0$ limit at fixed nonzero multiplet charge $r$ would use a different scaling of the parameters.

Next, let us take $c\downarrow1$ at fixed $q$. Although $\Lambda_{q,c}\to0$ in this limit, the density develops a hard edge at $x=2$, which corresponds to $\Ehat=0$, and the force becomes the removable divided difference
\begin{equation}
 \mathcal V'_{q,1}(x)
 =\alpha_q\frac{\mathcal V'_{\rm b}(x)-\mathcal V'_{\rm b}(2)}{2-x},
 \qquad
 \varrho_{q,1}(x)
 =\frac{\Pi_q(x)}{8\pi}\sqrt{\frac{2+x}{2-x}}.
 \label{eq:pot-hard-edge}
\end{equation}
The real domain now has a wall at $x=2$ (equivalently, $\Ehat\geq0$), and varying the matrix integral with this constraint gives a wall term in the hard-edge loop equation. This term survives even though the logarithmic coefficient tends to zero, as we can see directly from the limit. Indeed, the logarithmic part of~\eqref{eq:pot-loop-source} is $\Lambda_{q,c}R_{q,c}(x_\ast)/(x_\ast-x)$, and
\begin{equation}
 \lim_{x_\ast\downarrow2}\Lambda_{q,c}R_{q,c}(x_\ast)=\frac1{16}.
 \label{eq:pot-wall-coefficient}
\end{equation}
This follows from~\eqref{eq:pot-removable-pole}, since $y_{\rm b}(x_\ast)\sim A_q^3\sqrt{x_\ast^2-4}/2$ and $R'_{\rm b}(x_\ast)\sim-A_q^3/\sqrt{x_\ast^2-4}$. If we let $\mathcal S^{\rm reg}_{q,1}$ denote the loop source~\eqref{eq:pot-loop-source} formed using only the limiting regular force, the limiting loop equation is
\begin{equation}
 R_{q,1}^2-\mathcal V'_{q,1}R_{q,1}
       +\mathcal S^{\rm reg}_{q,1}(x)+\frac{1}{16(2-x)}=0.
 \label{eq:pot-hard-wall-loop}
\end{equation}
For $q=0$, this reduces to the linear potential in $\Ehat$, with its hard-edge density on $[0,1]$.

To summarize, we have obtained an explicit potential, filling fraction and one-cut planar solution for the global curve, which supply the planar data for the matrix-model construction. A nonperturbative definition would also require a convergent integration contour and an analysis of the global stability of the saddle across the parameter range. As in the bosonic construction~\cite{Okuyama:2023byh}, relating the higher-genus expansion to microscopic DSSYK correlators is a separate question. Note that these questions concern the matrix potential studied here, and not the auxiliary string-equation solution of Appendix~\ref{app:ODEdata}, whose nonperturbative completion would be a separate construction.

%-------------------------------------------------
\section{Proofs of the edge identities and the discrete degree bound}
\label{app:coloring-proof}

%-------------------------------------------------
\begingroup
\newcommand{\Res}{\mathop{\rm Res}} \newcommand{\dd}{\,d} \newcommand{\loc}{\mathrm{loc}} \newcommand{\Nplus}{N^{(+\cdots+)}} \newcommand{\wplus}{\omega^{+}}

In this appendix, we prove the lower-edge identity~\eqref{eq:colored-identity} for every stable marked topology, compute the mixed-color contributions retained by the external parity projection, and derive the upper-edge counterpart. In the final subsection, we prove the total-degree bound for the full discrete volumes. The proofs use the colored recursion graphs obtained by iterating~\eqref{eq:global-TR}. We work with \(n\geq1\) and \(2g-2+n>0\), so the unstable disk and cylinder, as well as the unmarked free energies, are treated separately from the stable marked amplitudes.

\subsection{Colored recursion and the all-topology statement}

As in the main text, \(\Ezero\) is the unscaled gap, so that \(c=1+2\Ezero\). To begin with, we regard the \(Z_m\) of~\eqref{eq:qzeta-intro} as independent formal variables, and write the global data as
\begin{align}
 x(z)&=z+z^{-1},&
 y(z)&=\frac{z-z^{-1}}{8(2+4\Ezero-x(z))}
       \exp\!\left[-\sum_{m\geq1}\frac{Z_m}{m}(z-z^{-1})^{2m}\right],
 \label{eq:app-color-global}\\
 B(z_1,z_2)&=\frac{dz_1\,dz_2}{(z_1-z_2)^2},&
 \iota(z)&=z^{-1}.
 \nonumber
\end{align}
Let us first expand the two-residue recursion~\eqref{eq:global-TR} and label each recursive residue by its ramification point. In the resulting colored recursion graphs, a vertex represents an individual kernel insertion, and the external marked points are treated separately. A term of type \((g,n)\) has \(d=2g-2+n\) such vertices, since the complexities of the children sum to \(d-1\) and \(B\) has complexity zero. If we assign the label \(+\) to every vertex, we obtain the monochromatic sector directly at the level of the recursion, without using an intersection-theoretic representation in terms of stable graphs.

This restriction only uses the local expansion at \(z=1\). We set \(\wplus_{0,2}=B\) and, for \(d=2g-2+n>0\), define
\begin{equation}
\begin{split}
\wplus_{g,n}(z_1,J)=\Res_{z=1}(-K(z_1,z))
\bigg[&\wplus_{g-1,n+1}(z,\iota z,J)\\
&+\sum_{\substack{g_1+g_2=g\\I\sqcup I^c=J}}'
\wplus_{g_1,|I|+1}(z,I)
\wplus_{g_2,|I^c|+1}(\iota z,I^c)\bigg].
\end{split}
\label{eq:app-color-recursion}
\end{equation}
Here the first term is omitted for \(g=0\), and the prime excludes the \((0,1)\) factors but retains \((0,2)\). Every child is again computed with the single residue at \(+1\), and the differential pullbacks, including \(d(\iota z)=-dz/z^2\), are part of this formula. The kernel is given by
\begin{equation}
 K(z_1,z)=
 \frac{2(2+4\Ezero-z-z^{-1})z^3
       \exp[\sum_{m\geq1}Z_m(z-z^{-1})^{2m}/m]}
 {(z^2-1)(z_1-z)(z_1z-1)}\,\frac{dz_1}{dz}.
 \label{eq:app-color-kernel}
\end{equation}
This is the same sign convention as in~\eqref{eq:global-TR}.

The corresponding local data are
\begin{equation}
 \begin{gathered}
 x_\loc(w)=\frac{w^2}{2},\qquad
 y_\loc(w)=\frac{w}{\Ezero-w^2}
 \exp\!\left[-\sum_{m\geq1}\frac{Z_m}{m}(4w)^{2m}\right],\\
 B_\loc(w_1,w_2)=\frac{dw_1\,dw_2}{(w_1-w_2)^2}.
 \end{gathered}
 \label{eq:app-color-local}
\end{equation}
For the local recursion, we use the same sign, the involution \(w\mapsto-w\) and the residue at \(w=0\). We then define the monochromatic discrete volumes and the local volumes by
\begin{align}
\wplus_{g,n}(\mathbf z)
 &=\sum_{b_1,\ldots,b_n\geq1}\Nplus_{g,n}(\mathbf b)
       \prod_i b_i z_i^{b_i-1}dz_i,
 \label{eq:app-color-discrete}\\
\frac{\omega^\loc_{g,n}(\mathbf w)}{\prod_i dw_i}
 &=\int_{(0,\infty)^n} V^q_{g,n}(\mathbf b)
       \prod_i b_i \e^{-w_i b_i}\dd b_i .
 \label{eq:app-color-laplace}
\end{align}
The second identity is understood term by term for polynomials.

\begin{colorprop}\label{prop:app-color-main}
For \(g\geq0\), \(n\geq1\) and \(2g-2+n>0\), the function \(\Nplus_{g,n}\) is a polynomial in the \(b_i^2\), \(\Ezero\) and finitely many \(Z_m\). With the weights
\begin{equation}
 \wt(b_i)=1,\qquad \wt(Z_m)=2m,\qquad \wt(\Ezero)=-2,
 \label{eq:app-color-weight}
\end{equation}
all its monomials have weight at most \(2g-2\), and
\begin{equation}
 \displaystyle
 V^q_{g,n}(b_1,\ldots,b_n)=
 2^{-n}\,[\Nplus_{g,n}]_{2g-2}
       \left(\frac{b_1}{2},\ldots,\frac{b_n}{2}\right).
 \label{eq:app-color-identity}
\end{equation}
This is an identity over \(\mathbb Q[\Ezero,Z_1,Z_2,\ldots]\), so in particular it holds after specializing \(Z_m=\zeta_{q^2}(2m)\), and also at \(\Ezero=0\).
\end{colorprop}

The proposition concerns the fully monochromatic recursion on stable marked topologies defined above. The external parity projection of the full two-residue answer is a different operation, and we examine its relation to the monochromatic sector in Appendix~\ref{app:coloring-parity}.

\subsection{A polynomial basis for the monochromatic sector}
\begin{colorlemma}\label{lem:app-color-basis}
For every stable pair, \(\wplus_{g,n}\) has a unique finite expansion
\begin{equation}
 \wplus_{g,n}(\mathbf z)
   =\sum_{\mathbf k\geq0}\mathcal C_{g,n;\mathbf k}(\Ezero,\mathbf Z)
                        \prod_{i=1}^n \Xi_{k_i}(z_i),
 \qquad
 \Xi_k(z)=d_z\!\left[(z\partial_z)^{2k}\frac{z}{1-z}\right],
 \label{eq:app-color-basis}
\end{equation}
where the \(\mathcal C_{g,n;\mathbf k}\) are polynomials and \(\deg_{\Ezero} \mathcal C_{g,n;\mathbf k}\leq d\). Consequently
\begin{equation}
 \Nplus_{g,n}(\mathbf b)
 =\sum_{\mathbf k}\mathcal C_{g,n;\mathbf k}\prod_i b_i^{2k_i}.
 \label{eq:app-color-poly}
\end{equation}
\end{colorlemma}

\begin{proof}
At a fixed recursion step, only finitely many Taylor coefficients of the exponential in~\eqref{eq:app-color-kernel} can contribute to the residue, so by induction we get rational external differentials with coefficients in \(\mathbb Q[\Ezero,\mathbf Z]\). In each external variable, their only possible pole is at \(1\), and their scalar coefficient is \(O(z_i^{-2})\) at infinity and regular at zero. These properties follow from the external denominators in \(K\) and \(B\), and they are preserved by taking the residue at \(1\).

We also need the anti-invariance in \emph{each} external variable,
\begin{equation}
 \iota_i^*\wplus_{g,n}=-\wplus_{g,n}.
 \label{eq:app-color-anti}
\end{equation}
For the first variable, this follows from \(\iota_1^*K=-K\), which can be checked directly in~\eqref{eq:app-color-kernel}. For another variable in a stable child, we use induction. If that variable lies in a \(B\) factor instead, the two ordered splittings containing that factor combine as
\begin{align*}
 B(z,z_i)\widetilde\omega(\iota z,J')+B(\iota z,z_i)\widetilde\omega(z,J')
  =\big[B(\iota z,z_i)-B(z,z_i)\big]\widetilde\omega(z,J').
\end{align*}
Here \(\widetilde\omega\) is the stable child, and we have used its anti-invariance in the first variable. Since \((\iota\times\iota)^*B=B\), the inversion of \(z_i\) exchanges the two Bergman factors, which gives the required minus sign. The only case in which both factors are unstable is \((0,3)\), whose direct residue is
\begin{equation}
 \wplus_{0,3}=8\Ezero\prod_{i=1}^3\frac{dz_i}{(1-z_i)^2},
 \label{eq:app-color-03}
\end{equation}
which has the same property, while the other initial case, \((1,1)\), has only the first external variable. This proves~\eqref{eq:app-color-anti} directly in each variable, without using the permutation symmetry of the restricted recursion.

The differentials \(\Xi_k\) are anti-invariant and satisfy
\begin{equation}
 \Xi_k(z)=\sum_{b\geq1}b^{2k+1}z^{b-1}dz,\qquad
 \Xi_k(z)\sim(2k+1)!\frac{dz}{(z-1)^{2k+2}}\quad(z\to1).
 \label{eq:app-color-xi}
\end{equation}
Conversely, for a rational one-form with the properties above, anti-invariance forces the highest pole order at \(1\) to be even, since the local derivative of \(\iota\) there is \(-1\). We can then subtract the appropriate multiple of \(\Xi_k\) to reduce that order, and repeat, until we are left with a holomorphic one-form on the sphere, which must vanish. This proves both spanning and independence. Applying the argument in each variable gives~\eqref{eq:app-color-basis}, and~\eqref{eq:app-color-xi} then gives~\eqref{eq:app-color-poly}.

Finally, every fully expanded recursion term contains exactly \(d\) kernels, since a quadratic bracket has total child complexity \(d-1\), with \(B\) assigned complexity zero. Each kernel is linear in \(\Ezero\), which proves the claimed bound on the degree in the gap.
\end{proof}

\subsection{The exact coordinate and the scaling of residues}
Let us now introduce a formal parameter \(\beta\) and make the simultaneous substitutions
\begin{equation}
 \Ezero\longmapsto\beta^2\Ezero,\qquad
 Z_m\longmapsto\beta^{-2m}Z_m,\qquad
 z=z_\beta(w):=\big(\sqrt{1+\beta^2w^2}+\beta w\big)^2.
 \label{eq:app-color-scaling}
\end{equation}
Here the square root is expanded formally with constant term \(1\). These substitutions extract the weights algebraically, and the physical \(q\) is kept separate from the independent formal rescalings of the \(Z_m\). For the calculation, it is convenient to write
\begin{equation}
 \mathfrak r_\beta(w)=\sqrt{1+\beta^2w^2},\qquad
 \Phi_\beta(w)=\sum_{m\geq1}\frac{Z_m}{m}
                   (4w)^{2m}(1+\beta^2w^2)^m.
 \label{eq:app-color-scaled-functions}
\end{equation}
At every fixed order in \(w\), the sum \(\Phi_\beta\) involves only finitely many \(Z_m\). The coordinate satisfies the exact identities
\begin{equation}
 z_\beta(-w)=z_\beta(w)^{-1},\qquad
 x(z_\beta(w))=2+4\beta^2w^2,\qquad
 z_\beta(w)-z_\beta(w)^{-1}=4\beta w\,\mathfrak r_\beta(w).
 \label{eq:app-color-coordinate}
\end{equation}
In particular, the involution becomes \(w\mapsto-w\).

Let \(Y_\beta\) be the pullback of \(y\,dx\) after the parameter substitutions in~\eqref{eq:app-color-scaling}. By direct substitution, we find
\begin{equation}
 Y_\beta(w)
 =\beta\,\frac{w^2\mathfrak r_\beta(w)}{\Ezero-w^2}\e^{-\Phi_\beta(w)}\,dw
 =\beta\big(y_\loc\,dx_\loc+O(\beta^2)\big).
 \label{eq:app-color-omega-scaling}
\end{equation}
It is useful to keep the Bergman kernel exact as well, and its pullback is
\begin{equation}
 B_\beta(w_1,w_2)
 =\frac{z_\beta'(w_1)z_\beta'(w_2)}{(z_\beta(w_1)-z_\beta(w_2))^2}\,dw_1\,dw_2
 =\frac{dw_1\,dw_2}
 {\mathfrak r_\beta(w_1)\mathfrak r_\beta(w_2)\,[w_1\mathfrak r_\beta(w_2)-w_2\mathfrak r_\beta(w_1)]^2}.
 \label{eq:app-color-B-exact}
\end{equation}
For example, expanding the exact expression gives
\begin{equation}
 B_\beta(w_1,w_2)=\left[
 \frac1{(w_1-w_2)^2}-\frac{\beta^2}{2}
 +\frac{3\beta^4}{8}(w_1^2+w_2^2)+O(\beta^6)\right]dw_1\,dw_2.
 \label{eq:app-color-B-scaling}
\end{equation}
The correction to the standard double pole is holomorphic near the diagonal, and only even powers of \(\beta\) appear, since \(z_{-\beta}=1/z_\beta\) and \(B\) is invariant under simultaneous inversion.

The coordinate-free form of the kernel~\eqref{eq:app-color-kernel} is
\begin{equation}
 K(z_1,z)=\frac{\displaystyle\int_{\iota(z)}^z B(z_1,\cdot)}
                          {4\,y(z)\,dx(z)}.
 \label{eq:app-color-intrinsic}
\end{equation}
Its numerator also has a simple exact pullback,
\begin{equation}
 \begin{aligned}
 \int_{-w}^{w}B_\beta(w_1,\cdot)
 &=\left[\frac1{z_\beta(w_1)-z_\beta(w)}
       -\frac1{z_\beta(w_1)-z_\beta(-w)}\right]z_\beta'(w_1)\,dw_1\\
 &=\frac{2w\,\mathfrak r_\beta(w)}{\mathfrak r_\beta(w_1)(w_1^2-w^2)}\,dw_1.
 \end{aligned}
 \label{eq:app-color-integrated-B}
\end{equation}
Here we have evaluated the rational primitive at its endpoints, which gives the formal integral without any choice of contour. Dividing by \(4Y_\beta\) then cancels the square root that involves the residue variable, and we get
\begin{equation}
 K_\beta^+(w_1,w)
 =\beta^{-1}
 \frac{(\Ezero-w^2)\e^{\Phi_\beta(w)}}
 {2w(w_1^2-w^2)\mathfrak r_\beta(w_1)}\,\frac{dw_1}{dw}.
 \label{eq:app-color-K-exact}
\end{equation}
We will write \(K_\beta=K_\beta^+\) for the lower-edge calculation. In terms of the local kernel, this is
\begin{align}
 K_\beta(w_1,w)&=\beta^{-1}\big(K_\loc(w_1,w)+O(\beta^2)\big),
 \label{eq:app-color-K-scaling}\\
 K_\loc(w_1,w)&=
 \frac{\Ezero-w^2}{2w(w_1^2-w^2)}
 \e^{\Phi_0(w)}\,\frac{dw_1}{dw}.
 \nonumber
\end{align}
More precisely, the factor multiplying \(\beta^{-1}K_\loc\) is the invertible power series
\begin{equation}
 \mathcal U_\beta(w_1,w)=
 \frac{\exp[\Phi_\beta(w)-\Phi_0(w)]}{\mathfrak r_\beta(w_1)}
 =1+O(\beta^2).
 \label{eq:app-color-kernel-unit}
\end{equation}
This recovers the factor \(1/2\) in the local kernel and makes its formal scaling explicit. Note that the factor in \(y_\loc\) is \((\Ezero-w^2)^{-1}\), so in the recursion kernel \(\Ezero-w^2\) appears in the \emph{numerator}. The kernel therefore remains polynomial in \(\Ezero\), including at \(\Ezero=0\).

\begin{colorlemma}\label{lem:app-color-scaling}
Let \(\widehat\omega^\beta_{g,n}\) denote the pullback of \(\wplus_{g,n}\) under~\eqref{eq:app-color-scaling}, in all external variables. Then
\begin{equation}
 \beta^{\,d}\widehat\omega^\beta_{g,n}
      =\omega^\loc_{g,n}+O(\beta^2),\qquad d=2g-2+n.
 \label{eq:app-color-correlator-scaling}
\end{equation}
In fact, the left-hand side is a formal power series in \(\beta^2\).
\end{colorlemma}

\begin{proof}
Residues are invariant under \(z=z_\beta(w)\), viewed as an invertible local coordinate over the field of Laurent series in \(\beta\), since \(z_\beta(0)=1\) and \(z_\beta'(0)=2\beta\neq0\). So~\eqref{eq:app-color-recursion} becomes the single-residue recursion in \(w\), with the exact pulled-back data \(B_\beta,K_\beta\) and the involution \(w\mapsto-w\), and the original recursive order of the residues is preserved.

Next, we justify extracting the coefficients in \(\beta\) before evaluating each residue. For this, we need a bound on the principal part (i.e., the negative powers of the local coordinate) that is uniform in the coefficient being extracted. The coordinate satisfies
\begin{equation}
 z_\beta(w)-1=2\beta w\,[\mathfrak r_\beta(w)+\beta w],
 \label{eq:app-color-coordinate-unit}
\end{equation}
where the last factor has constant term one, and is therefore invertible as a formal power series. Consequently, a fixed pole of order \(\ell\) at \(z=1\) pulls back to a Laurent series with pole order at most \(\ell\) in \(w\), independently of the order in \(\beta\). Lemma~\ref{lem:app-color-basis} supplies finite pole orders for every stable child, and when two child variables are set to \(w\) and \(-w\), their orders add but remain finite. The unstable child \(B_\beta(w,-w)\) has a fixed double pole, including the pullback of \(d(-w)\), while for the other Bergman factors the external variables are held as parameters. In addition,~\eqref{eq:app-color-K-exact} bounds the pole of the kernel in the residue variable by a fixed simple pole, and the factor \(\mathcal U_\beta\) contains only nonnegative powers of \(w\), so it preserves this bound.

We now induct on \(d\), with \(B_\beta=B_\loc+O(\beta^2)\) as the complexity-zero datum. Assuming the claimed normalized expansions for the children, the root integrand, multiplied by \(\beta^d\) and with the external differentials suppressed, belongs to \(w^{-\ell_{\max}}\mathbb A[[w,\beta^2]]\) for some finite \(\ell_{\max}\). Here \(\mathbb A\) is a coefficient algebra containing the formal parameters and the external variables, localized as necessary in those external variables, and its dependence on \(\Ezero\) remains polynomial. Since the bound on \(\ell_{\max}\) that we have just established is independent of the power of \(\beta\), taking the residue and extracting a parameter coefficient are commuting algebraic operations for such an integrand \(I_\beta(w)\,dw\),
\begin{equation}
 [\beta^{2j}]\Res_{w=0}I_\beta(w)\,dw
 = [\beta^{2j}w^{-1}]I_\beta(w)
 =\Res_{w=0}[\beta^{2j}]I_\beta(w)\,dw,
 \qquad j\geq0.
 \label{eq:app-color-residue-exchange}
\end{equation}
So the exchange is justified coefficient by coefficient in the formal series.

The genus-reducing child has complexity \(d-1\), and the two children in a product have total complexity \(d-1\). Since the root kernel contributes \(\beta^{-1}\), we get the singularity bound \(\beta^{-d}\). By~\eqref{eq:app-color-residue-exchange}, the coefficient of \(\beta^{-d}\) is the residue of the leading child coefficients and \(K_\loc\), which obey the local recursion, with the same explicit minus sign, differential pullbacks and ordered splittings. The recursive initial data agree as well, so uniqueness proves~\eqref{eq:app-color-correlator-scaling}. All corrections occur in even powers of \(\beta\), and the formal identity also covers specializations of the parameters for which the leading coefficient vanishes.
\end{proof}

\subsection{Extracting the homogeneous polynomial}
The basis forms of Lemma~\ref{lem:app-color-basis} have the leading pullback
\begin{equation}
 z_\beta^*\Xi_k=
 \beta^{-(2k+1)}
 \left[\frac{(2k+1)!}{2^{2k+1}}\frac{dw}{w^{2k+2}}
       +O(\beta^2)\right].
 \label{eq:app-color-xi-scaling}
\end{equation}
Indeed, the leading term is determined by~\eqref{eq:app-color-xi}, \(z_\beta-1=2\beta w+O(\beta^2)\) and \(dz_\beta=(2\beta+O(\beta^2))dw\). Moreover, the anti-invariance of \(\Xi_k\) and \(z_{-\beta}=1/z_\beta\) show that its pullback is odd in \(\beta\), so the bracket in~\eqref{eq:app-color-xi-scaling} contains only even powers.

Now consider a nonzero monomial in~\eqref{eq:app-color-poly}, which is a nonzero rational multiple of
\[
 \Ezero^j\prod_m Z_m^{j_m}\prod_i b_i^{2k_i},
 \qquad
 \wt=2\sum_i k_i+2\sum_m m j_m-2j.
\]
Up to that rational factor, its contribution to the scaled differential begins at
\begin{equation}
 \beta^{-(\wt+n)}\,\Ezero^j\prod_m Z_m^{j_m}
       \prod_i\frac{(2k_i+1)!}{2^{2k_i+1}}
                 \frac{dw_i}{w_i^{2k_i+2}} .
 \label{eq:app-color-monomial}
\end{equation}
For distinct tuples \((j,(j_m),\mathbf k)\), these leading expressions are linearly independent, since the parameter monomials are independent over \(\mathbb Q\) and the Laurent powers of the \(w_i\) distinguish \(\mathbf k\).

If \(\Nplus_{g,n}\) vanishes, the identity follows immediately from Lemma~\ref{lem:app-color-scaling}. Otherwise, let \(\wt_{\max}\) be its largest weight. By linear independence, the resulting term of order \(\beta^{-(\wt_{\max}+n)}\) is nonzero, while smaller weights contribute only less singular powers. Since the scaling lemma bounds the singularity by \(\beta^{-d}\), we get
\begin{equation}
 \wt_{\max}\leq d-n=2g-2.
 \label{eq:app-color-degree-bound}
\end{equation}
Consequently, the coefficient of \(\beta^{-d}\) comes only from the leading terms~\eqref{eq:app-color-monomial} of weight \(2g-2\), and the coordinate corrections, which begin at higher orders, leave this coefficient unchanged.

Writing \(\mathcal C^{\mathrm{top}}_{\mathbf k}\) for the part of \(\mathcal C_{g,n;\mathbf k}\) of parameter weight \(2g-2-2\sum_i k_i\), we have proved
\begin{equation}
 \omega^\loc_{g,n}=
 \sum_{\mathbf k}\mathcal C^{\mathrm{top}}_{\mathbf k}
       \prod_i \frac{(2k_i+1)!}{2^{2k_i+1}}
                    \frac{dw_i}{w_i^{2k_i+2}}.
 \label{eq:app-color-leading-form}
\end{equation}
The elementary Laplace identity
\begin{equation}
 \int_0^\infty b^{2k+1}\e^{-wb}\dd b
             =\frac{(2k+1)!}{w^{2k+2}},\qquad \Re w>0,
 \label{eq:app-color-laplace-monomial}
\end{equation}
turns~\eqref{eq:app-color-leading-form} into
\[
 V^q_{g,n}(\mathbf b)
   =\sum_{\mathbf k}\mathcal C^{\mathrm{top}}_{\mathbf k}
                  \prod_i\frac{b_i^{2k_i}}{2^{2k_i+1}}
   =2^{-n}[\Nplus_{g,n}]_{2g-2}(\mathbf b/2).
\]
This proves Proposition~\ref{prop:app-color-main}. \hfill\(\square\)

\paragraph{Two useful consequences.}
Since \(j\leq d\) and the \(Z_m\) have nonnegative weights, the bound~\eqref{eq:app-color-degree-bound} also gives
\[
 \deg_{\mathbf b}\Nplus_{g,n}\leq 2g-2+2d=6g-6+2n.
\]
Moreover, all statements are polynomial in \(\Ezero\), so setting \(\Ezero=0\) gives the hard-edge identity directly, without any expansion in \(1/\Ezero\), even when the leading coefficient vanishes at that specialization.

\subsection{External characters and internal colors}
\label{app:coloring-parity}

Let us now allow poles at both endpoints. We put \(\Xi_k^+=\Xi_k\) and let \(\Xi_k^-\) be its pullback under \(z\mapsto-z\), so that their expansions at the origin are
\begin{equation}
 \Xi_k^\pm(z)=\sum_{b\geq1}(\pm1)^b b^{2k+1}z^{b-1}dz.
 \label{eq:app-color-signed-basis}
\end{equation}
The only pole of \(\Xi_k^+\) is at \(+1\), and the only pole of \(\Xi_k^-\) is at \(-1\). The rationality and anti-invariance argument of Lemma~\ref{lem:app-color-basis}, now allowing both poles, gives a unique expansion of the full correlator in products of these signed basis forms. So \(N_{\mathcal I}\) in~\eqref{eq:parity-character} is the coefficient polynomial for poles at \(-1\) in the variables indexed by \(\mathcal I\), and at \(+1\) in the remaining variables, or equivalently,
\begin{equation}
 N_{\mathcal I}=2^{-n}\sum_{\boldsymbol\varsigma\in\{\pm1\}^n}
 \left(\prod_{i\in\mathcal I}\varsigma_i\right)
 N^q_{g,n}\big|_{\boldsymbol\varsigma},
 \label{eq:app-color-fourier}
\end{equation}
where each parity sector is regarded as a polynomial in the same formal boundary variables. In particular, \(N_\varnothing\) is the average of all the parity-sector polynomials.

This projection acts on the external variables after the residues have been evaluated, so the internal kernel labels are still summed over. The following mixed graph shows how this distinction survives the homogeneous projection.

\paragraph{A mixed graph at \((1,2)\).}
We write \(\omega=W\prod_i dz_i\), and let \(W^-_{1,1}\) be the one-loop contribution from the residue at \(-1\). Using~\eqref{eq:N11minus} and~\eqref{eq:app-color-signed-basis}, we find
\begin{align}
 W^-_{1,1}(z)
 &=-\frac{\mathcal C^-_{1,1;0}}{(1+z)^2}
   -\mathcal C^-_{1,1;1}\frac{z^2-4z+1}{(1+z)^4},
 \label{eq:app-color-W11minus}\\
 \mathcal C^-_{1,1;0}&=4(1+\Ezero)Z_1-\frac{2\Ezero}{3}-\frac{5}{12},
 \qquad \mathcal C^-_{1,1;1}=\frac{1+\Ezero}{6},\nonumber\\
 W^-_{1,1}(1)&=-(1+\Ezero)Z_1+\frac{3\Ezero}{16}+\frac18.
 \label{eq:app-color-W11value}
\end{align}
This function is regular at the lower branch point. Let \(\widehat K_{\rm glob}\) be the scalar kernel, defined by \(K_{\rm glob}(z_1,z)=\widehat K_{\rm glob}(z_1,z)\,dz_1/dz\). Then the contribution with a lower root, an upper one-loop child and both external poles at \(+1\) is
\begin{equation}
 W^{\rm mix}_{1,2}(z_1,z_2)=
 \Res_{z=1}\frac{\widehat K_{\rm glob}(z_1,z)}{z^2}
 \left[\frac{W^-_{1,1}(z^{-1})}{(z-z_2)^2}
       +\frac{W^-_{1,1}(z)}{(z^{-1}-z_2)^2}\right].
 \label{eq:app-color-mixed-residue}
\end{equation}
Here the factor \(z^{-2}\) incorporates \(d(z^{-1})=-dz/z^2\) and the explicit minus sign of the recursion. Near the residue, we have
\[
 \widehat K_{\rm glob}(z_1,z)
   =\frac{4\Ezero}{(z-1)(z_1-1)^2}+O(1),
\]
while the square bracket is analytic, and consequently
\begin{equation}
 W^{\rm mix}_{1,2}(z_1,z_2)=
 \frac{8\Ezero\,W^-_{1,1}(1)}{(1-z_1)^2(1-z_2)^2}.
 \label{eq:app-color-mixed-answer}
\end{equation}
At this topology there are exactly two recursive vertices. The other mixed term with a lower root contains \(\omega^-_{0,3}=-8(1+\Ezero)\prod_i dz_i/(1+z_i)^2\), whose remaining external pole is at \(-1\), so it is removed by the all-lower external projection. The terms with an upper root have their first external pole at \(-1\), and they are removed as well. Therefore,~\eqref{eq:app-color-mixed-answer} is the full difference,
\begin{equation}
 N_{\varnothing;1,2}-\Nplus_{1,2}
   =\Ezero+\frac32\Ezero^2-8\Ezero(1+\Ezero)Z_1.
 \label{eq:app-color-exact-difference}
\end{equation}
Taking the part of weight zero, rescaling the boundaries and using Proposition~\ref{prop:app-color-main}, we get
\begin{equation}
 \displaystyle
 \frac14[N_{\varnothing;1,2}]_0(b_1/2,b_2/2)
       -V^q_{1,2}(b_1,b_2)=-2\Ezero Z_1.
 \label{eq:app-color-failure12}
\end{equation}
Since \(\wt(\Ezero Z_1)=0\), the mixed graph contributes at the same weight as \(V^q_{1,2}\). Its correction is generically nonzero, and it vanishes at special values of the parameters, such as \(\Ezero=0\) or \(q=0\).

For comparison, the Fourier transformation of~\eqref{eq:app-N12} gives
\begin{equation}
 N_{\{1\};1,2}=N_{\{2\};1,2}
       =\frac{c^2-1}{2}=2\Ezero(1+\Ezero).
 \label{eq:app-color-mixed-character}
\end{equation}
These mixed \emph{external} characters describe a separate part of the answer, while the contribution of internal mixed colors to the all-lower external character is given by~\eqref{eq:app-color-exact-difference}.

\paragraph{The next one-boundary example.}
We define
\[
 \Delta_{g,n}(\mathbf b)=
 2^{-n}[N_{\varnothing;g,n}]_{2g-2}(\mathbf b/2)
      -V^q_{g,n}(\mathbf b).
\]
At \((2,1)\), we have \(N_\varnothing=(N_{\rm even}+N_{\rm odd})/2\), where \(N_{\rm even}\) and \(N_{\rm odd}\) are the two parity polynomials of~\eqref{eq:app-N21}. Substituting~\eqref{eq:app-N21}--\eqref{eq:app-N21-even-coeffs} and~\eqref{eq:app-V21}, and extracting the indicated weight, we obtain the exact identity
\begin{equation}
\begin{split}
 \Delta_{2,1}(b)=-\frac1{96}\Big[
 &21504\Ezero^2Z_1^3+384\Ezero^2Z_1^2b^2+15360\Ezero^2Z_1Z_2+\Ezero^2Z_1b^4\\
 &-8256\Ezero Z_1^2-24\Ezero Z_1b^2-3840\Ezero Z_2+24Z_1\Big].
\end{split}
 \label{eq:app-color-failure21}
\end{equation}
Every monomial here has weight two. In particular, \(\Delta_{2,1}|_{\Ezero=0}=-Z_1/4\), so the discrepancy persists even at the hard edge. The same extraction in the earlier examples gives
\begin{equation}
 \Delta_{0,3}=\Delta_{1,1}=\Delta_{0,4}=0.
 \label{eq:app-color-earlier}
\end{equation}
So, ordered by the recursion complexity \(2g-2+n\), equations~\eqref{eq:app-color-failure12} and~\eqref{eq:app-color-earlier} locate the first mixed-color correction at \((1,2)\), and the distinction between external parity and internal color persists after the homogeneous projection.

\subsection{The upper-edge counterpart}
\label{app:coloring-upper}

The same argument also proves the upper-edge assertion of Section~\ref{sec:edges}, with its different grading. Let us define \(\omega^-_{g,n}\) by taking the residue at \(-1\) at every recursive vertex, so that its coefficients have the form
\begin{equation}
 N^{(-\cdots-)}_{g,n}(\mathbf b)
       =(-1)^{b_1+\cdots+b_n}\widetilde N^{(-\cdots-)}_{g,n}(\mathbf b),
 \label{eq:app-color-upper-polynomial}
\end{equation}
where \(\widetilde N^{(-\cdots-)}_{g,n}\) is a polynomial in the boundary squares. We use the same parameter substitutions \(\Ezero\mapsto\beta^2\Ezero\) and \(Z_m\mapsto\beta^{-2m}Z_m\), but now set
\[
 z^-_\beta(\tilde w)=-\big(\sqrt{1+\beta^2\tilde w^2}+\beta \tilde w\big)^2.
\]
Then \(x(z^-_\beta(\tilde w))=-2-4\beta^2\tilde w^2\), and the exact pullback is
\begin{equation}
 (z^-_\beta)^*(y\,dx)=
 \beta^3\frac{\tilde w^2\sqrt{1+\beta^2\tilde w^2}}{1+\beta^2(\Ezero+\tilde w^2)}
 \exp\!\left[-\sum_{m\geq1}\frac{Z_m}{m}
             (4\tilde w)^{2m}(1+\beta^2\tilde w^2)^m\right]d\tilde w.
 \label{eq:app-color-upper-scaling}
\end{equation}
The Bergman kernel and its integrated primitive~\eqref{eq:app-color-integrated-B} are unchanged by the extra simultaneous minus sign, so dividing by the exact one-form in~\eqref{eq:app-color-upper-scaling} gives
\begin{equation}
 K^-_\beta(\tilde w_1,\tilde w)=\beta^{-3}
 \frac{[1+\beta^2(\Ezero+\tilde w^2)]\e^{\Phi_\beta(\tilde w)}}
 {2\tilde w(\tilde w_1^2-\tilde w^2)\mathfrak r_\beta(\tilde w_1)}\,\frac{d\tilde w_1}{d\tilde w}.
 \label{eq:app-color-upper-K-exact}
\end{equation}
In particular, \(K^-_\beta=\beta^{-3}(K_{\rm up}+O(\beta^2))\), where \(K_{\rm up}\) is the kernel of
\begin{equation}
 x_{\rm up}(\tilde w)=\frac{\tilde w^2}{2},\qquad
 y_{\rm up}(\tilde w)=\tilde w
 \exp\!\left[-\sum_{m\geq1}\frac{Z_m}{m}(4\tilde w)^{2m}\right],
 \qquad B_{\rm up}=\frac{d\tilde w_1\,d\tilde w_2}{(\tilde w_1-\tilde w_2)^2}.
 \label{eq:app-color-bosonic}
\end{equation}
The extra factor in~\eqref{eq:app-color-upper-K-exact} is an invertible series with nonnegative powers of \(\tilde w\) and \(\beta^2\), just as at the lower edge, so the same bounded-principal-part argument and the coefficient identity~\eqref{eq:app-color-residue-exchange} apply. Counting three powers of \(\beta^{-1}\) per kernel, the residue induction gives \(\beta^{3d}(z^-_\beta)^*\omega^-_{g,n} =\omega^{\rm up}_{g,n}+O(\beta^2)\), with the parameter substitutions understood. Moreover, \((z^-_\beta)^*\Xi_k^-=(z^+_\beta)^*\Xi_k^+\), where \(z^+_\beta=z_\beta\), so the boundary normalization is the same as before. The linear independence of the leading basis terms then implies \(\wt(\widetilde N^{(-\cdots-)}_{g,n})\leq3d-n=6g-6+2n\), and hence
\begin{equation}
 \displaystyle
 V^{q,\mathrm{up}}_{g,n}(\mathbf b)=
 2^{-n}[\widetilde N^{(-\cdots-)}_{g,n}]_{6g-6+2n}(\mathbf b/2).
 \label{eq:app-color-upper-identity}
\end{equation}
Here $V^{q,\mathrm{up}}_{g,n}$ is the upper-edge specialization defined in Section~\ref{sec:edges}. The formula uses the polynomial factor $\widetilde N^{(-\cdots-)}_{g,n}$, with the alternating character factored out as above. In the bosonic convention where \(\zeta_q(2m)\) multiplies \((2\tilde w)^{2m}\) in the exponent,~\eqref{eq:app-color-bosonic} corresponds to the substitution \(\zeta_q(2m)\mapsto4^mZ_m\) in~\eqref{eq:upper-substitution}. So the upper and lower identities follow from the same colored recursion, with kernel weights three and one, respectively.

\subsection{The total-degree bound for the full discrete volumes}
\label{app:global-degree}

We now include both endpoints in every step of the recursion. Let \(D=3g-3+n\), with \(n\geq1\) and \(2g-2+n>0\). The two-pole basis argument of Appendix~\ref{app:coloring-parity} then gives
\begin{equation}
 \omega_{g,n}
 =\sum_{\mathcal I\subseteq\{1,\ldots,n\}}\sum_{\mathbf k\geq0}
 \mathcal C^{\mathcal I}_{g,n;\mathbf k}
 \prod_{i\in\mathcal I}\Xi_{k_i}^-(z_i)\prod_{i\notin\mathcal I}\Xi_{k_i}^+(z_i),
 \label{eq:global-degree-basis}
\end{equation}
which is a finite sum at fixed topology. The factor \(b_i\) in each basis coefficient explains the boundary factors in~\eqref{eq:discrete-def}, and after dividing by them, the remaining powers are \(b_i^{2k_i}\). What we must bound is the \emph{sum} of the \(k_i\), not just each one separately.

For \(c>1\), both ramification points are regular. At fixed topology, only finitely many Taylor coefficients of the analytic product enter the recursion, so the regular local expansion applies even though the global curve is not algebraic, and its external descendant indices \(j_i\) satisfy \(\sum_i j_i\leq D\)~\cite[Proposition 4.1]{Eynard:2011intersection}. It is useful to recall the dimension count behind this bound. Consider a connected stable graph with \(e\) internal edges, and with vertex genera and valences \(g_v,n_v\). Then
\begin{equation}
 \sum_v(3g_v-3+n_v)=D-e.
 \label{eq:global-degree-dimension}
\end{equation}
Here the vertices are stable components in the intersection expansion, rather than the individual kernel insertions counted earlier. If we give each internal half-edge \(h\) a descendant index \(j_h\geq0\), and represent the background translation by extra markings with indices \(\ell\geq2\), then matching the degree of the insertions to the dimension at every vertex gives
\begin{equation}
 \sum_i j_i+\sum_{\text{internal half-edges }h}j_h
       +\sum_{\text{extra markings}}(\ell-1)=D-e.
 \label{eq:global-degree-budget}
\end{equation}
All the terms other than the external sum are nonnegative.

The standard differential with descendant index \(j\) has a pole of order \(2j+2\) at its assigned endpoint, and the basis reduction used above is triangular in the pole order, i.e., it replaces this differential by \(\Xi_k^\pm\) with \(k\leq j\). Hence every nonzero term in~\eqref{eq:global-degree-basis} obeys
\begin{equation}
 \sum_i k_i\leq D.
 \label{eq:global-degree-bound}
\end{equation}
Each fixed-parity polynomial therefore has total degree at most \(D\) in the boundary squares, as claimed in Section~\ref{sec:globalTR}.

The same count also bounds the deformation data. Near either endpoint, \(z-z^{-1}\) has a simple zero, so \(Z_m\) first changes the coefficient of degree \(2m\) beyond the leading term of \(y\), and it first enters an intersection coupling of degree \(m\). Since the leading term and \(B\) are independent of the \(Z_m\), no \(Z_m\) with \(m>D\) can occur.

Finally, the explicit recursion kernel~\eqref{eq:global-kernel} is linear in \(c\), and each stable amplitude uses \(2g-2+n\) kernels, so its coefficients are polynomials in \(c\) and finitely many \(Z_m\). The identities established for \(c>1\) therefore extend to \(c=1\), including the bounds on the degree and on the deformation data, even though the lower endpoint then becomes a hard edge. We also note that the permutation symmetry of \(\omega_{g,n}\) exchanges the lengths \emph{and} their parities, so within one fixed parity pattern it gives exactly the permutations that preserve that pattern.

\endgroup

%-------------------------------------------------
\section{More global discrete data}
\label{app:globaldata}
%-------------------------------------------------

In this appendix, we collect the global discrete volumes at $(1,2)$ and $(2,1)$. At $(g,n)=(1,2)$, let us put $|\mathbf b|^2=b_1^2+b_2^2$ and
\begin{equation}
 \mathcal X_{1,2}=|\mathbf b|^4+96Z_1|\mathbf b|^2+1344Z_1^2+960Z_2.
\end{equation}
Then, indicating the parities of $(b_1,b_2)$ in the arguments, we have
\begin{equation}
 \begin{aligned}
 N^q_{1,2}({\rm odd},{\rm odd})&=\frac16[(c^2+1)\mathcal X_{1,2}-12c^2|\mathbf b|^2-192c^2Z_1+192Z_1+20c^2-4],\\
 N^q_{1,2}({\rm odd},{\rm even})&=N^q_{1,2}({\rm even},{\rm odd})=-\frac{c}{3}[\mathcal X_{1,2}-6|\mathbf b|^2+5],\\
 N^q_{1,2}({\rm even},{\rm even})&=\frac16[(c^2+1)\mathcal X_{1,2}-12c^2|\mathbf b|^2-192c^2Z_1+192Z_1+32c^2-16].
 \end{aligned}
\label{eq:app-N12}
\end{equation}
These formulae are exact in $q$.

At $(g,n)=(2,1)$, the recursion takes the particularly simple one-point form
\begin{equation}
 W_{2,1}(z_1)=\sum_{\pm}\mathop{\rm Res}_{z=\pm1}
 \frac{\widehat K_{\rm glob}(z_1,z)}{z^2}
 \left[W_{1,2}(z,z^{-1})+W_{1,1}(z)W_{1,1}(z^{-1})\right].
\label{eq:app-W21-global-recursion}
\end{equation}
Expanding the resulting rational function around $z_1=0$ gives the two parity polynomials, which we write as
\begin{equation}
 N^q_{2,1}(b)=
 \begin{cases}
 -\dfrac{1}{17280}\displaystyle\sum_{j=0}^{4}{\cal A}_j b^{2j},
 & b\ \text{odd},\\[8pt]
 \dfrac{c}{17280}\displaystyle\sum_{j=0}^{4}{\cal B}_j b^{2j},
 & b\ \text{even}.
 \end{cases}
\label{eq:app-N21}
\end{equation}
The coefficients in the odd sector are
\begin{align}
 {\cal A}_4={}&5(1+3c^2),\\
 {\cal A}_3={}&12\left[(232+696c^2)Z_1+3-49c^2\right],\\
 {\cal A}_2={}&6\left[(68928+206784c^2)Z_1^2
 +(6704-15312c^2)Z_1\right.\nonumber\\
 &\left.\hspace{2.0cm} +(27840+83520c^2)Z_2-247+1015c^2\right],\\
 {\cal A}_1={}&4\Big[(4400640+13201920c^2)Z_1^3
 +(1126080-60480c^2)Z_1^2\nonumber\\
 &\hspace{0.8cm} +(5598720+16796160c^2)Z_1Z_2
 +(-25656+39672c^2)Z_1\nonumber\\
 &\hspace{0.8cm} +(463680-8640c^2)Z_2
 +(1612800+4838400c^2)Z_3+1861-4923c^2\Big],\\
 {\cal A}_0={}&9\Big[17305600Z_1^4+9123840Z_1^3+49029120Z_1^2Z_2
 +436608Z_1^2\nonumber\\
 &\hspace{0.8cm}+13424640Z_1Z_2+33259520Z_1Z_3+6624Z_1
 +11366400Z_2^2\nonumber\\
 &\hspace{0.8cm}+266880Z_2+4730880Z_3+12902400Z_4-667\nonumber\\
 &\hspace{0.8cm}+c^2\big(51916800Z_1^4+9205760Z_1^3
 +147087360Z_1^2Z_2-110976Z_1^2\nonumber\\
 &\hspace{2.1cm}+15144960Z_1Z_2+99778560Z_1Z_3-8352Z_1
 +34099200Z_2^2\nonumber\\
 &\hspace{2.1cm}-51840Z_2+5877760Z_3+38707200Z_4+1575\big)\Big].
\label{eq:app-N21-odd-coeffs}
\end{align}
For the even sector, the coefficients are
\begin{align}
 {\cal B}_4={}&5(c^2+3),\\
 {\cal B}_3={}&24\left[(348+116c^2)Z_1-10-13c^2\right],\\
 {\cal B}_2={}&48\left[(25848+8616c^2)Z_1^2
 +(300-1376c^2)Z_1\right.\nonumber\\
 &\left.\hspace{2.0cm} +(10440+3480c^2)Z_2-23+119c^2\right],\\
 {\cal B}_1={}&128\Big[(412560+137520c^2)Z_1^3
 +(51840-18540c^2)Z_1^2\nonumber\\
 &\hspace{0.8cm} +(524880+174960c^2)Z_1Z_2
 +(-2574+3012c^2)Z_1\nonumber\\
 &\hspace{0.8cm} +(21600-7380c^2)Z_2
 +(151200+50400c^2)Z_3+165-286c^2\Big],\\
 {\cal B}_0={}&18432\Big[25350Z_1^4+8930Z_1^3+71820Z_1^2Z_2+3Z_1^2
 +13530Z_1Z_2\nonumber\\
 &\hspace{0.8cm}+48720Z_1Z_3+30Z_1+16650Z_2^2+45Z_2+4900Z_3
 +18900Z_4-3\nonumber\\
 &\hspace{0.8cm}+c^2\big(8450Z_1^4+20Z_1^3+23940Z_1^2Z_2+156Z_1^2
 +420Z_1Z_2\nonumber\\
 &\hspace{2.1cm}+16240Z_1Z_3-36Z_1+5550Z_2^2+60Z_2+280Z_3
 +6300Z_4+4\big)\Big].
\label{eq:app-N21-even-coeffs}
\end{align}
As expected from the pole structure of the global curve, each sector is a quartic in $b^2$, with $q$-zeta values up to $Z_4=\zeta_{q^2}(8)$. As a check, if we set $c=1+2\Ezero$, resolve the constant parity character and take its homogeneous lower-edge part, then~\eqref{eq:app-N21} reproduces the previously obtained $(2,1)$ edge projection, including the mixed-color correction to the local volume.

%-------------------------------------------------
\section{Local volume data}
\label{app:localvolumes}
%-------------------------------------------------

For reference, we collect here the local volumes used in the checks of Sections~\ref{sec:localvol} and~\ref{sec:ODE}. With $|\mathbf b|^2=b_1^2+b_2^2$, we have
\begin{equation}
 V^q_{0,3}=\Ezero,
\end{equation}
\begin{equation}
 V^q_{1,1}(b)=-\frac18+\frac{\Ezero}{48}(b^2+96Z_1),
\end{equation}
\begin{equation}
 V^q_{0,4}(b_1,b_2,b_3,b_4)
 =-3\Ezero+\frac{\Ezero^2}{2}
 \left(\sum_{i=1}^{4}b_i^2+96Z_1\right),
\end{equation}
\begin{equation}
 \begin{aligned}
 V^q_{1,2}(b_1,b_2)={}&\frac18-\frac{\Ezero}{8}(|\mathbf b|^2+112Z_1)\\
 &+\frac{\Ezero^2}{192}
 \left(|\mathbf b|^4+384Z_1|\mathbf b|^2+21504Z_1^2+15360Z_2\right).
 \end{aligned}
\end{equation}
At $(2,1)$,
\begin{equation}
\begin{aligned}
V^q_{2,1}(b)=\frac{1}{2211840}\Big\{&-
\left(25920b^2+2488320Z_1\right)\\
&+\Ezero\left(15408b^4+7603200Z_1b^2+493240320Z_1^2+236666880Z_2\right)\\
&-\Ezero^2\Big(696b^6+827136Z_1b^4
 +(211230720Z_1^2+89579520Z_2)b^2\\
&\hspace{2.05cm}+9968025600Z_1^3+14120386560Z_1Z_2+4789370880Z_3\Big)\\
&+\Ezero^3\Big(5b^8+11136Z_1b^6+(6617088Z_1^2+2672640Z_2)b^4\\
&\hspace{2.05cm}+(1126563840Z_1^3+1433272320Z_1Z_2+412876800Z_3)b^2\\
&\hspace{2.05cm}+39872102400Z_1^4+112963092480Z_1^2Z_2
 +26188185600Z_2^2\\
&\hspace{2.05cm}+76629934080Z_1Z_3+29727129600Z_4\Big)\Big\}.
\end{aligned}
\label{eq:app-V21}
\end{equation}

In the gapless theory, we have
\begin{equation}
 V^q_{0,n}(\mathbf b;0)=0,
 \qquad
 V^q_{1,n}(\mathbf b;0)=\frac{(-1)^n(n-1)!}{8},
\end{equation}
\begin{equation}
 V^q_{2,1}(b;0)=-\frac{3}{256}(b^2+96Z_1),
\end{equation}
\begin{equation}
 V^q_{2,2}(b_1,b_2;0)=\frac{9}{256}(b_1^2+b_2^2+128Z_1),
 \qquad
 V^q_{2,3}(\mathbf b;0)=-\frac{9}{64}\left(\sum_{i=1}^{3}b_i^2+160Z_1\right),
\end{equation}
\begin{equation}
 V^q_{3,1}(b;0)=-\frac{3}{8192}
 \left(5b^4+2688Z_1b^2+181248Z_1^2+76800Z_2\right).
\end{equation}
In the gapless sector, these expressions make the homogeneity and the delayed onset of the $q$-dependence manifest.

%-------------------------------------------------
%-------------------------------------------------
\section{String-equation data and checks}
\label{app:ODEdata}
%-------------------------------------------------

In this appendix, we derive the data used in Section~\ref{sec:ODE} from the lower local density, following~\cite{Johnson:2023ofr,Ahmed:2025extended,AhmedNaskar:2026resolvent}. As in Section~\ref{sec:ODE}, $\Ezero$ and $\Ehat$ are the unscaled threshold and energy, $\xi$ is the auxiliary coordinate, and $\widetilde\Gamma_q^{\rm str}$ is the signed coefficient of the string equation, while the continuum threshold is $E_0$, with $\Ezero=\epsilon^2E_0$ and $\epsilon=1-q$. When we differentiate $u_0(\xi)$ or the analytic part $G_0(u)=\sum_{k\geq1}t_ku^k$ of the string function, we hold $\Ezero$ and $q$ fixed.

\subsection{The density coefficients}
\label{subsec:ODE-density-coefficients}

For $0<q<1$, and initially for $\Ezero>0$, the local density~\eqref{eq:local-density} reads
\begin{equation}
 \rho^q_0(\Ehat)=\frac{\sqrt{\Ehat-\Ezero}}{2\pi\hbar\Ehat}\,
 \mathcal F_q(\Ehat-\Ezero),\qquad
 \mathcal F_q(v)=\prod_{n\geq1}
 \left(1+\frac{16q^{2n}v}{(1-q^{2n})^2}\right)
 =\sum_{n\geq0}\mathfrak f_n v^n.
 \label{eq:app-ODE-density}
\end{equation}
Since the sum of the absolute values of the coefficients of $v$ in the product factors converges, the product defines an entire function of $v$. Near the origin, it has the logarithmic representation
\begin{equation}
 \log\mathcal F_q(v)
 =-\sum_{m\geq1}\frac{(-16)^mZ_m}{m}v^m
 =16Z_1v-128Z_2v^2+\frac{4096}{3}Z_3v^3-16384Z_4v^4+\cdots.
 \label{eq:app-ODE-log-density}
\end{equation}
When we evaluate at $v=-\Ezero$ below, we use the product or its entire Taylor series, and the logarithmic representation only within its disk of convergence. The same coefficients enter the intersection formula through $\mathcal F_q(-w^2)$ in~\eqref{eq:it-spectral-factor}.

For compactness, let us define
\begin{equation}
 \begin{aligned}
 \mathfrak p_2&=Z_1^2-Z_2,\\
 \mathfrak p_3&=Z_1^3-3Z_1Z_2+2Z_3,\\
 \mathfrak p_4&=Z_1^4-6Z_1^2Z_2+3Z_2^2+8Z_1Z_3-6Z_4.
 \end{aligned}
 \label{eq:app-ODE-P-polynomials}
\end{equation}
In terms of these, exponentiating~\eqref{eq:app-ODE-log-density} gives
\begin{equation}
 \mathfrak f_0=1,\qquad \mathfrak f_1=16Z_1,\qquad
 \mathfrak f_2=128\mathfrak p_2,\qquad
 \mathfrak f_3=\frac{2048}{3}\mathfrak p_3,\qquad
 \mathfrak f_4=\frac{8192}{3}\mathfrak p_4.
 \label{eq:app-ODE-f-list}
\end{equation}
More generally, differentiating $\mathcal F_q$ gives the triangular recurrence
\begin{equation}
 n\mathfrak f_n=\sum_{m=1}^n(-1)^{m+1}16^mZ_m\mathfrak f_{n-m},
 \qquad n\geq1.
 \label{eq:app-ODE-f-recurrence}
\end{equation}

If we expand $\Ehat^{-1}$ in powers of $\Ehat-\Ezero$, for $|\Ehat-\Ezero|<\Ezero$, we get the threshold series
\begin{equation}
 \rho^q_0(\Ehat)=\frac1{\pi\hbar}\sum_{\ell\geq1}\varpi_\ell(\Ehat-\Ezero)^{\ell-1/2},
 \qquad
 \varpi_\ell=\frac12\sum_{j=0}^{\ell-1}
       \frac{(-1)^{\ell-1-j}\mathfrak f_j}{\Ezero^{\ell-j}}.
 \label{eq:app-ODE-canonical}
\end{equation}
In particular,
\begin{equation}
 \varpi_1=\frac1{2\Ezero},\qquad
 \varpi_2=\frac{8Z_1}{\Ezero}-\frac1{2\Ezero^2},\qquad
 \varpi_3=\frac{64\mathfrak p_2}{\Ezero}-\frac{8Z_1}{\Ezero^2}+\frac1{2\Ezero^3}.
 \label{eq:app-ODE-C123}
\end{equation}
Here the factor $1/2$ in~\eqref{eq:app-ODE-canonical} is fixed by the prefactor of the original density.

\subsection{Abel inversion and the complete inverse-string function}
\label{subsec:ODE-Abel}

For a decreasing tree-level potential, the WKB density is given by
\begin{equation}
 \rho^q_0(\Ehat)=\frac{1}{2\pi\hbar}
 \int_{\xi_\ast}^{\mu}\frac{d\xi}{\sqrt{\Ehat-u_0(\xi)}},
 \qquad u_0(\xi_\ast)=\Ehat,\qquad u_0(\mu)=\Ezero.
\label{eq:WKB-density}
\end{equation}
Here $\xi_\ast$ is the turning point. The signed square root of~\eqref{eq:string-equation} gives the tree-level equation, and we fix the origin of $\xi$ by
\begin{equation}
 G_0(u_0(\xi))+\xi=\frac{\widetilde\Gamma^{\rm str}_q}{\sqrt{u_0(\xi)}},
 \qquad G_0(u)=\sum_{k\geq1}t_ku^k,\qquad G_0(0)=0.
 \label{eq:app-ODE-tree}
\end{equation}
Let us also introduce
\begin{equation}
 \mathcal G(u)=G_0(u)-\widetilde\Gamma^{\rm str}_q u^{-1/2},
 \qquad \xi(u)=-\mathcal G(u).
 \label{eq:app-ODE-complete-G}
\end{equation}
On a decreasing branch, changing the integration variable in~\eqref{eq:WKB-density}, we find
\begin{equation}
 \rho^q_0(\Ehat)=\frac1{2\pi\hbar}
 \int_{\Ezero}^{\Ehat}\frac{\mathcal G'(u)}{\sqrt{\Ehat-u}}\,du.
 \label{eq:app-ODE-Abel}
\end{equation}
So the full density determines $\mathcal G$ directly, and its analytic component $G_0$ is obtained after separating the pole term.

In order to obtain the dictionary between the coefficients, we write $u=\Ezero+v$ and expand
\begin{equation*}
 \mathcal G'(\Ezero+v)=\sum_{\ell\geq1}
 \frac{\mathcal G^{(\ell)}(\Ezero)}{(\ell-1)!}\,v^{\ell-1}.
\end{equation*}
The required integral follows from the substitution $v=(\Ehat-\Ezero)t$,
\begin{equation}
 \begin{aligned}
 \int_0^{\Ehat-\Ezero}\frac{v^{\ell-1}}{\sqrt{\Ehat-\Ezero-v}}\,dv
 &=(\Ehat-\Ezero)^{\ell-1/2}\int_0^1t^{\ell-1}(1-t)^{-1/2}\,dt\\
 &=\frac{2^\ell(\ell-1)!}{(2\ell-1)!!}\,(\Ehat-\Ezero)^{\ell-1/2}.
 \end{aligned}
 \label{eq:app-ODE-beta-integral}
\end{equation}
Comparing with~\eqref{eq:app-ODE-canonical}, we find
\begin{equation}
 \mathcal G_\ell\equiv\mathcal G^{(\ell)}(\Ezero)
 =\frac{(2\ell-1)!!}{2^{\ell-1}}\varpi_\ell
 =\frac{(2\ell-1)!!}{2^\ell}
    \sum_{j=0}^{\ell-1}\frac{(-1)^{\ell-1-j}\mathfrak f_j}{\Ezero^{\ell-j}}.
 \label{eq:app-ODE-Ap}
\end{equation}
Only $\mathfrak f_0,\ldots,\mathfrak f_{\ell-1}$ enter a given $\mathcal G_\ell$, and the first three are
\begin{equation}
 \mathcal G_1=\frac1{2\Ezero},\qquad
 \mathcal G_2=\frac{12Z_1}{\Ezero}-\frac3{4\Ezero^2},\qquad
 \mathcal G_3=\frac{240\mathfrak p_2}{\Ezero}-\frac{30Z_1}{\Ezero^2}+\frac{15}{8\Ezero^3}.
 \label{eq:app-ODE-A123}
\end{equation}
These coefficients fix the derivatives in~\eqref{eq:app-ODE-u12-explicit}--\eqref{eq:app-ODE-u3-explicit}. In order to reconstruct the times themselves, we next separate the nonanalytic pole term.

\subsection{The pole coefficient and the analytic times}
\label{subsec:ODE-times}

The pole part of $\mathcal G'(u)$ is $\widetilde\Gamma^{\rm str}_q/(2u^{3/2})$, and its contribution to~\eqref{eq:app-ODE-Abel} is elementary, since
\begin{equation}
 \frac{d}{du}\sqrt{\frac{\Ehat-u}{u}}
 =-\frac{\Ehat}{2u^{3/2}\sqrt{\Ehat-u}},\qquad
 \int_{\Ezero}^{\Ehat}\frac{du}{u^{3/2}\sqrt{\Ehat-u}}
 =\frac{2\sqrt{\Ehat-\Ezero}}{\Ehat\sqrt{\Ezero}}.
 \label{eq:app-ODE-pole-integral}
\end{equation}
The pole term therefore contributes
\begin{equation}
 \rho_\Gamma(\Ehat)=\frac{\widetilde\Gamma^{\rm str}_q}{2\pi\hbar\sqrt{\Ezero}}
                \frac{\sqrt{\Ehat-\Ezero}}{\Ehat}.
 \label{eq:app-ODE-pole-density}
\end{equation}
We isolate the analytic part $G_0$ by choosing the pole coefficient such that the numerator of $\rho^q_0-\rho_\Gamma$ vanishes at $\Ehat=0$, which gives
\begin{equation}
 \widetilde\Gamma^{\rm str}_q=\sqrt{\Ezero}\,\mathcal F_q(-\Ezero)
 =\sqrt{\Ezero}\prod_{n\geq1}
       \left(1-\frac{16q^{2n}\Ezero}{(1-q^{2n})^2}\right).
 \label{eq:app-ODE-Gamma}
\end{equation}
Equivalently, on the sheet with $\sqrt{\Ehat-\Ezero}=i\sqrt{\Ezero}$ at zero energy,
\begin{equation}
 \mathop{\rm Res}_{\Ehat=0}\rho^q_0(\Ehat)=\frac{i\widetilde\Gamma^{\rm str}_q}{2\pi\hbar}.
 \label{eq:app-ODE-pole-residue}
\end{equation}
Equation~\eqref{eq:app-ODE-pole-residue} fixes the signed coupling of the lower local string equation. The exact finite-$q$ microscopic BPS relation~\eqref{eq:BPS-residue} uses the global coupling~\eqref{eq:Gamma-global}, with the multiplet weight $\mathcal W_r$ restored, in the charge range specified in Section~\ref{sec:BPS}. The continuum-normalized local and global couplings share the correlated limit~\eqref{eq:Gamma-limit}; their relation to the unscaled string-equation convention is given in Appendix~\ref{subsec:ODE-continuum-data}.

The remaining density is
\begin{equation}
 \rho_G(\Ehat)=\frac{\sqrt{\Ehat-\Ezero}}{2\pi\hbar}\Upsilon_q(\Ehat-\Ezero),
 \qquad
 \Upsilon_q(v)=\frac{\mathcal F_q(v)-\mathcal F_q(-\Ezero)}{v+\Ezero}.
 \label{eq:app-ODE-subtracted-density}
\end{equation}
The quotient $\Upsilon_q$ is entire, since its apparent pole is removable. Using
\begin{equation*}
 \frac{v^n-(-\Ezero)^n}{v+\Ezero}
 =\sum_{j=0}^{n-1}v^{n-1-j}(-\Ezero)^j
\end{equation*}
we find
\begin{equation}
 \Upsilon_q(v)=\sum_{\ell\geq1}v^{\ell-1}\sum_{n\geq \ell}\mathfrak f_n(-\Ezero)^{n-\ell}.
 \label{eq:app-ODE-divided-difference}
\end{equation}
Applying~\eqref{eq:app-ODE-beta-integral} to the subtracted density now determines the analytic derivatives,
\begin{equation}
 G_0^{(\ell)}(\Ezero)=\frac{(2\ell-1)!!}{2^\ell}
               \sum_{n\geq \ell}\mathfrak f_n(-\Ezero)^{n-\ell},
 \qquad \ell\geq1.
 \label{eq:app-ODE-G0-derivatives}
\end{equation}
Once the nonanalytic term in $\mathcal G$ has been separated, the function $G_0$ is analytic at the origin, so we can Taylor expand its derivatives from $u=\Ezero$ back to $u=0$,
\begin{equation}
 \begin{aligned}
 t_k=\frac{G_0^{(k)}(0)}{k!}
 &=\frac1{k!}\sum_{\ell\geq k}
      \frac{(-\Ezero)^{\ell-k}}{(\ell-k)!}\,G_0^{(\ell)}(\Ezero)\\
 &=\frac1{k!}\sum_{n\geq k}\mathfrak f_n(-\Ezero)^{n-k}
      \sum_{\ell=k}^n\frac{(2\ell-1)!!}{2^\ell(\ell-k)!}.
 \end{aligned}
 \label{eq:app-ODE-times-double-sum}
\end{equation}
The inner sum is
\begin{equation}
 \sum_{\ell=k}^n\frac{(2\ell-1)!!}{2^\ell(\ell-k)!}
 =\frac{(2n+1)!!}{2^n(2k+1)(n-k)!}.
 \label{eq:app-ODE-finite-sum}
\end{equation}
This identity can be proven by induction, since at $n=k$ both sides equal $(2k-1)!!/2^k$, and subtracting the right-hand side at $n-1$ from that at $n$ gives the last summand. The exact multicritical times are therefore
\begin{equation}
 t_k(\Ezero,q)=\frac1{(2k+1)k!}
 \sum_{n\geq k}\frac{(2n+1)!!}{2^n(n-k)!}
                  \mathfrak f_n(-\Ezero)^{n-k},\qquad k\geq1.
 \label{eq:app-ODE-times}
\end{equation}
For fixed $k$, the factorial ratio in this series grows only polynomially in $n$, so the entire Taylor coefficients of $\mathcal F_q$ give locally uniform convergence. This justifies the preceding rearrangements, and the resulting $G_0$ is entire in its argument, as is also explicit in~\eqref{eq:app-ODE-G0-integral} below.

For small gap, the first few times are
\begin{equation}
 \begin{aligned}
 t_1&=8Z_1-160\Ezero \mathfrak p_2+\frac{4480}{3}\Ezero^2\mathfrak p_3
                      -8960\Ezero^3\mathfrak p_4+O(\Ezero^4),\\
 t_2&=48\mathfrak p_2-896\Ezero \mathfrak p_3+8064\Ezero^2\mathfrak p_4+O(\Ezero^3),\\
 t_3&=\frac{640}{3}\mathfrak p_3-3840\Ezero \mathfrak p_4+O(\Ezero^2),\\
 t_4&=\frac{2240}{3}\mathfrak p_4+O(\Ezero).
 \end{aligned}
 \label{eq:app-ODE-times-expansion}
\end{equation}
At nonzero $\Ezero$, all times $t_k$ with $k\geq\ell$ can contribute to $G_0^{(\ell)}(\Ezero)$. The endpoint calculation therefore uses the full tail of the series, and~\eqref{eq:app-ODE-G0-derivatives} resums these contributions exactly.

\subsection{The Fermi level and the implicit solution}
\label{subsec:ODE-mu-solution}

At the endpoint, where $u_0(\mu)=\Ezero$, the tree equation gives the Fermi level
\begin{equation}
 \mu=\frac{\widetilde\Gamma^{\rm str}_q}{\sqrt{\Ezero}}-G_0(\Ezero)
    =\mathcal F_q(-\Ezero)-\sum_{k\geq1}t_k\Ezero^k.
 \label{eq:app-ODE-mu-endpoint}
\end{equation}
At fixed $n$, substituting~\eqref{eq:app-ODE-times} into $G_0(\Ezero)$ gives the coefficient
\begin{equation*}
 \frac{(2n+1)!!}{2^nn!}\,\mathfrak f_n(-\Ezero)^n
 \sum_{k=1}^n\frac{(-1)^k}{2k+1}\binom nk.
\end{equation*}
The finite sum can be evaluated by integrating the binomial theorem,
\begin{equation}
 \sum_{k=0}^n\frac{(-1)^k}{2k+1}\binom nk
 =\int_0^1(1-t^2)^n\,dt
 =\frac{2^n n!}{(2n+1)!!}.
 \label{eq:app-ODE-binomial-integral}
\end{equation}
It follows that
\begin{equation}
 G_0(\Ezero)=\sum_{n\geq1}\left(1-\frac{(2n+1)!!}{2^nn!}\right)\mathfrak f_n(-\Ezero)^n,
 \qquad
 \mu=\sum_{n\geq0}\frac{(2n+1)!!}{2^nn!}\,\mathfrak f_n(-\Ezero)^n.
 \label{eq:app-ODE-mu-series}
\end{equation}
This is the extension of~\eqref{eq:app-ODE-times} to $k=0$, so that $\mu=t_0$. Note that the origin in this calculation is fixed by the condition $G_0(0)=0$, since the density only determines $\xi(u)-\mu$, and the origin of $\xi$ is set by this additional convention.

The sum in~\eqref{eq:app-ODE-mu-series} is $\mathcal T_q(\Ezero)$ from~\eqref{eq:string-moment-data}, and together with the times and the pole coefficient, it gives the specialization~\eqref{eq:string-specialization}. The derivatives in that specialization act on the background parameter at fixed $q$, whereas the derivatives with respect to $u$ in $G_0(u)$ hold the background fixed. Here $\mu=t_0$ denotes the endpoint value, and the string equation keeps the form $\R=\sum_{k\geq1}t_kR_k[u]+\xi$.

There is also an integral representation, which follows from the arcsine moments,
\begin{equation}
 \mathcal T_q(\Ezero)=\frac1\pi\int_0^1
 \frac{\mathcal F_q(-\Ezero t)-2\Ezero t\mathcal F_q'(-\Ezero t)}{\sqrt{t(1-t)}}\,dt.
 \label{eq:app-ODE-M-integral}
\end{equation}
Indeed, $\pi^{-1}\int_0^1t^n/\sqrt{t(1-t)}\,dt=4^{-n}\binom{2n}{n}$, and the numerator contributes the additional factor $2n+1$. In particular,
\begin{equation}
 \begin{aligned}
 \mu&=1-24Z_1\Ezero+240\mathfrak p_2\Ezero^2-\frac{4480}{3}\mathfrak p_3\Ezero^3
                                      +6720\mathfrak p_4\Ezero^4+O(\Ezero^5),\\
 \frac{\widetilde\Gamma^{\rm str}_q}{\sqrt{\Ezero}}
 &=1-16Z_1\Ezero+128\mathfrak p_2\Ezero^2-\frac{2048}{3}\mathfrak p_3\Ezero^3
                               +\frac{8192}{3}\mathfrak p_4\Ezero^4+O(\Ezero^5).
 \end{aligned}
 \label{eq:app-ODE-mu-Gamma-expansions}
\end{equation}
Their difference gives the analytic contribution $G_0(\Ezero)$.

The times can also be resummed before inverting the tree equation,
\begin{equation}
 G_0(u)=\int_0^1\big[\mathcal T_q(\Ezero-ut^2)-\mathcal T_q(\Ezero)\big]\,dt.
 \label{eq:app-ODE-G0-integral}
\end{equation}
Expanding the integrand in $u$ and using $\int_0^1t^{2k}\,dt=(2k+1)^{-1}$ reproduces~\eqref{eq:string-specialization}, so an exact implicit specification of $u_0(\xi)$ is
\begin{equation}
 \xi=\frac{\widetilde\Gamma^{\rm str}_q}{\sqrt{u_0}}
   -\int_0^1\big[\mathcal T_q(\Ezero-u_0t^2)-\mathcal T_q(\Ezero)\big]\,dt,
 \qquad u_0(\mu)=\Ezero.
 \label{eq:app-ODE-implicit-solution}
\end{equation}
Moreover, the identity~\eqref{eq:app-ODE-binomial-integral} implies
\begin{equation*}
 \int_0^{\sqrt{\Ezero}}\mathcal T_q(\Ezero-v^2)\,dv
       =\sqrt{\Ezero}\,\mathcal F_q(-\Ezero)=\widetilde\Gamma^{\rm str}_q.
\end{equation*}
So the same inverse relation also has the useful endpoint-centered form
\begin{equation}
 \xi=\mu-\frac1{\sqrt{u_0}}
            \int_{\sqrt{\Ezero}}^{\sqrt{u_0}}
                  \mathcal T_q(\Ezero-v^2)\,dv.
 \label{eq:app-ODE-centered-solution}
\end{equation}
At $u_0=\Ezero$ the integral vanishes, which makes the endpoint manifest.

We can also obtain the inverse relation directly from the Abel transform. To do so, we apply a second Abel integral to~\eqref{eq:app-ODE-Abel} and interchange the two integrals, using
\begin{equation*}
 \int_v^u\frac{d\Ehat}{\sqrt{(\Ehat-v)(u-\Ehat)}}=\pi.
\end{equation*}
This gives
\begin{equation}
 \xi(u)=\mu-\frac1\pi\int_{\Ezero}^u
 \frac{\sqrt{\Ehat-\Ezero}\,\mathcal F_q(\Ehat-\Ezero)}{\Ehat\sqrt{u-\Ehat}}\,d\Ehat,
 \qquad u\geq \Ezero.
 \label{eq:app-ODE-direct-inversion}
\end{equation}
The two representations specify the same local branch, and the nonzero derivative $\mathcal G'(\Ezero)=1/(2\Ezero)$ ensures that the analytic inverse near $\xi=\mu$ is unique.

\subsection{Explicit endpoint derivatives}
\label{subsec:ODE-endpoint-derivation}

The resummed times make the cancellation of the pole explicit. From~\eqref{eq:app-ODE-G0-derivatives}, we have
\begin{equation}
 \begin{aligned}
 G_0'(\Ezero)&=\frac{1-\mathcal F_q(-\Ezero)}{2\Ezero},\\
 G_0''(\Ezero)&=\frac{12Z_1}{\Ezero}+\frac{3(\mathcal F_q(-\Ezero)-1)}{4\Ezero^2},\\
 G_0'''(\Ezero)&=\frac{240\mathfrak p_2}{\Ezero}-\frac{30Z_1}{\Ezero^2}
                                      +\frac{15(1-\mathcal F_q(-\Ezero))}{8\Ezero^3}.
 \end{aligned}
 \label{eq:app-ODE-G0-resummed}
\end{equation}
Since $\widetilde\Gamma^{\rm str}_q=\sqrt{\Ezero}\,\mathcal F_q(-\Ezero)$, adding the derivatives of $-\widetilde\Gamma^{\rm str}_q u^{-1/2}$, we get
\begin{equation}
 \begin{aligned}
 \mathcal G'(\Ezero)&=G_0'(\Ezero)+\frac{\mathcal F_q(-\Ezero)}{2\Ezero},\\
 \mathcal G''(\Ezero)&=G_0''(\Ezero)-\frac{3\mathcal F_q(-\Ezero)}{4\Ezero^2},\\
 \mathcal G'''(\Ezero)&=G_0'''(\Ezero)+\frac{15\mathcal F_q(-\Ezero)}{8\Ezero^3}.
 \end{aligned}
 \label{eq:app-ODE-pole-cancellation}
\end{equation}
All the explicit $\mathcal F_q(-\Ezero)$ terms cancel, and we recover~\eqref{eq:app-ODE-A123}. So, although the separate times and the pole coefficient contain infinite series, they combine in such a way that each endpoint derivative depends on only finitely many density coefficients.

We write $u_\ell=u_0^{(\ell)}(\mu)$ and $\mathcal G_\ell=\mathcal G^{(\ell)}(\Ezero)$. Differentiating $\mathcal G(u_0(\xi))+\xi=0$ once, twice and three times, we get
\begin{equation}
 \begin{aligned}
 \mathcal G_1u_1+1&=0,\\
 \mathcal G_2u_1^2+\mathcal G_1u_2&=0,\\
 \mathcal G_3u_1^3+3\mathcal G_2u_1u_2+\mathcal G_1u_3&=0.
 \end{aligned}
 \label{eq:app-ODE-chain-rule}
\end{equation}
Solving these in order, we find
\begin{equation}
 u_1=-\frac1{\mathcal G_1},\qquad
 u_2=-\frac{\mathcal G_2}{\mathcal G_1^3},\qquad
 u_3=\frac{\mathcal G_1\mathcal G_3-3\mathcal G_2^2}{\mathcal G_1^5}.
 \label{eq:app-ODE-inverse123}
\end{equation}
The first two substitutions are immediate,
\begin{equation}
 \begin{aligned}
 u_0'(\mu)&=-\left(\frac1{2\Ezero}\right)^{-1}=-2\Ezero,\\
 u_0''(\mu)&=-8\Ezero^3
          \left(\frac{12Z_1}{\Ezero}-\frac3{4\Ezero^2}\right)
             =6\Ezero-96Z_1\Ezero^2.
 \end{aligned}
 \label{eq:app-ODE-u12-explicit}
\end{equation}
For the third derivative, the two contributions to the numerator are
\begin{equation}
 \begin{aligned}
 \mathcal G_1\mathcal G_3&=\frac{120(Z_1^2-Z_2)}{\Ezero^2}
                    -\frac{15Z_1}{\Ezero^3}+\frac{15}{16\Ezero^4},\\
 3\mathcal G_2^2&=\frac{432Z_1^2}{\Ezero^2}
                    -\frac{54Z_1}{\Ezero^3}+\frac{27}{16\Ezero^4}.
 \end{aligned}
 \label{eq:app-ODE-third-numerator}
\end{equation}
Their difference is $-(312Z_1^2+120Z_2)/\Ezero^2+39Z_1/\Ezero^3-3/(4\Ezero^4)$, and multiplying by $\mathcal G_1^{-5}=32\Ezero^5$ gives
\begin{equation}
 \begin{aligned}
 u_0'''(\mu)&=-24\Ezero+1248Z_1\Ezero^2-(9984Z_1^2+3840Z_2)\Ezero^3\\
 &=-24\Ezero\big[1-52Z_1\Ezero+(416Z_1^2+160Z_2)\Ezero^2\big].
 \end{aligned}
 \label{eq:app-ODE-u3-explicit}
\end{equation}

Every higher derivative is obtained by applying $d/d\xi=-[\mathcal G'(u)]^{-1}d/du$ \emph{before} evaluating at $u=\Ezero$. At the next two orders, this gives
\begin{equation}
 \begin{aligned}
 u_0^{(4)}(\mu)&=\frac{-15\mathcal G_2^3+10\mathcal G_1\mathcal G_2\mathcal G_3-\mathcal G_1^2\mathcal G_4}{\mathcal G_1^7},\\
 u_0^{(5)}(\mu)&=\frac{-105\mathcal G_2^4+105\mathcal G_1\mathcal G_2^2\mathcal G_3
                     -10\mathcal G_1^2\mathcal G_3^2-15\mathcal G_1^2\mathcal G_2\mathcal G_4+\mathcal G_1^3\mathcal G_5}{\mathcal G_1^9}.
 \end{aligned}
 \label{eq:app-ODE-inverse45}
\end{equation}
The extra Abel coefficients are
\begin{equation}
 \begin{aligned}
 \mathcal G_4&=\frac{105}{16}\left(
       \frac{\mathfrak f_3}{\Ezero}-\frac{\mathfrak f_2}{\Ezero^2}
       +\frac{\mathfrak f_1}{\Ezero^3}-\frac1{\Ezero^4}\right),\\
 \mathcal G_5&=\frac{945}{32}\left(
       \frac{\mathfrak f_4}{\Ezero}-\frac{\mathfrak f_3}{\Ezero^2}
       +\frac{\mathfrak f_2}{\Ezero^3}-\frac{\mathfrak f_1}{\Ezero^4}
       +\frac1{\Ezero^5}\right).
 \end{aligned}
 \label{eq:app-ODE-A45}
\end{equation}
Substituting~\eqref{eq:app-ODE-f-list} into~\eqref{eq:app-ODE-inverse45}, we obtain the endpoint data used in the genus-two check,
\begin{equation}
 \begin{aligned}
 u_0^{(4)}(\mu)=-40\Ezero\Big[&-3+336Z_1\Ezero
 -(7584Z_1^2+2208Z_2)\Ezero^2\\
 &+(40448Z_1^3+35328Z_1Z_2+7168Z_3)\Ezero^3\Big].
 \end{aligned}
 \label{eq:app-u4}
\end{equation}
\begin{equation}
 \begin{aligned}
 u_0^{(5)}(\mu)=-720\Ezero\Big[&1-200Z_1\Ezero
 +(8912Z_1^2+2096Z_2)\Ezero^2\\
 &-(124160Z_1^3+87296Z_1Z_2+14336Z_3)\Ezero^3\\
 &+(496640Z_1^4+698368Z_1^2Z_2+229376Z_1Z_3
       +80896Z_2^2+43008Z_4)\Ezero^4\Big].
 \end{aligned}
 \label{eq:app-u5}
\end{equation}
Here $u_0^{(4)}$ and $u_0^{(5)}$ denote the fourth and fifth derivatives of the tree-level potential, so the superscripts are derivative orders, while powers of $\hbar$ organize the quantum expansion. Inserting these results into the universal genus-two resolvent primitive $\widehat Q_2$ (defined by $\widehat R_2=d\widehat Q_2/d\xi$ in Appendix~\ref{subsec:ODE-resolvent}, and given explicitly in~\cite{Ahmed:2025extended}), we obtain~\eqref{eq:app-V21}.

For completeness, we also give the first terms of the local solution,
\begin{equation}
 \begin{aligned}
 u_0(\xi)={}&\Ezero-2\Ezero(\xi-\mu)+(3\Ezero-48Z_1\Ezero^2)(\xi-\mu)^2\\
 &-4\Ezero\big[1-52Z_1\Ezero+(416Z_1^2+160Z_2)\Ezero^2\big](\xi-\mu)^3\\
 &+\frac{u_0^{(4)}(\mu)}{24}(\xi-\mu)^4
   +\frac{u_0^{(5)}(\mu)}{120}(\xi-\mu)^5+O\big((\xi-\mu)^6\big).
 \end{aligned}
 \label{eq:app-ODE-u-series}
\end{equation}
An all-order expression is given by Lagrange inversion,
\begin{equation}
 u_0(\xi)=\Ezero+\sum_{n\geq1}\frac{(\mu-\xi)^n}{n!}
 \left.\frac{d^{n-1}}{dv^{n-1}}
 \left[\frac{v}{\mathcal G(\Ezero+v)-\mathcal G(\Ezero)}\right]^n
 \right|_{v=0}.
 \label{eq:app-ODE-Lagrange}
\end{equation}
The quotient is regular at $v=0$ because $\mathcal G_1\neq0$, so the density specifies both the implicit solution and every coefficient of its endpoint expansion.

\subsection{Checks and the gapless branch}
\label{subsec:ODE-checks-gapless}

The moment representation gives a useful check that is independent of the threshold expansion. Indeed, the specialization~\eqref{eq:string-specialization} implies
\begin{equation}
 (2u\partial_u+1)G_0(u)+\mu=\mathcal T_q(\Ezero-u).
 \label{eq:app-ODE-M-identity}
\end{equation}
The same differential operator annihilates the $u^{-1/2}$ term, so with $\xi(u)=\widetilde\Gamma^{\rm str}_q/\sqrt u-G_0(u)$, it follows that
\begin{equation}
 2u\xi'(u)+\xi(u)=\mu-\mathcal T_q(\Ezero-u).
 \label{eq:app-ODE-x-identity}
\end{equation}
At $u=\Ezero$, using $\xi(\Ezero)=\mu$ and $\mathcal T_q(0)=1$, this gives $\xi'(\Ezero)=-1/(2\Ezero)$, and the next two differentiated identities are
\begin{equation*}
 2u\xi''+3\xi'=\mathcal T_q'(\Ezero-u),\qquad
 2u\xi'''+5\xi''=-\mathcal T_q''(\Ezero-u).
\end{equation*}
Since $\mathcal T_q'(0)=-24Z_1$ and $\mathcal T_q''(0)=480\mathfrak p_2$, they reproduce~\eqref{eq:app-ODE-A123}, with $\xi=-\mathcal G$.

At $q=0$, all the $Z_m$ vanish and $\mathcal F_q=\mathcal T_q=1$, so the exact data become $t_{k\geq1}=0$, $\mu=1$ and $\widetilde\Gamma^{\rm str}_q=\sqrt{\Ezero}$. On the branch with positive $\xi$, we then have
\begin{equation}
 u_0(\xi)=\frac{\Ezero}{\xi^2},\qquad
 u_0^{(\ell)}(1)=(-1)^\ell(\ell+1)!\Ezero.
 \label{eq:app-ODE-qzero}
\end{equation}
This directly checks the terms linear in the gap in all five derivatives.

At zero gap and fixed $q$, the parameter formulas instead give
\begin{equation}
 \widetilde\Gamma^{\rm str}_q=0,\qquad \mu=1,\qquad
 t_k(0,q)=\frac{(2k-1)!!}{2^k k!}\mathfrak f_k
         =\frac1{4^k}\binom{2k}{k}\mathfrak f_k.
 \label{eq:app-ODE-gapless-data}
\end{equation}
The classical string equation then becomes $u_0(G_0(u_0)+\xi)^2=0$, and the branch relevant at the Fermi surface is $u_0(\xi)=0$ for $0<\xi\leq\mu$. Its constant segment must be included in the density,
\begin{equation}
 \rho^q_0(\Ehat)=\frac1{2\pi\hbar}
 \left[\frac\mu{\sqrt{\Ehat}}+\int_0^{\Ehat}\frac{G_0'(u)}{\sqrt{\Ehat-u}}\,du\right]
 =\frac{\mathcal F_q(\Ehat)}{2\pi\hbar\sqrt{\Ehat}}.
 \label{eq:app-ODE-gapless-density}
\end{equation}
The beta integral verifies this equality term by term with~\eqref{eq:app-ODE-gapless-data}, and this direct calculation replaces the singular intermediate quantities $\mathcal G_\ell$ when $\Ezero=0$. The vanishing of the leading classical coupling is compatible with $\widetilde\Gamma_q^{\rm str}=\hbar\Gamma_q$ for any fixed $\Gamma_q$. The gapless volume sector considered here uses $\Gamma_q=0$.

\subsection{Normalization and the undeformed specialization}
\label{subsec:ODE-continuum-data}

We first relate the normalizations. The unscaled coefficient of the string equation and the continuum-normalized local pole coupling obey
\begin{equation}
 \widetilde\Gamma_q^{\rm str}(\Ezero)
 =2\pi(1-q)\,\widetilde\Gamma_q^{\rm loc}(\Ezero).
 \label{eq:ODE-Gamma-normalization}
\end{equation}
The normalization of~\eqref{eq:app-ODE-density} is the one that we used for the endpoint derivatives, and its continuum disk density is $2\pi$ times the convention underlying the undeformed times in~\cite{Johnson:2023ofr,Ahmed:2025extended,Johnson:2026boundary}. When passing between these conventions, all the string data must be changed together.

As before, let $\epsilon=1-q$ and let $E_0$ be a fixed continuum threshold, so that the unscaled threshold is $\Ezero=\epsilon^2E_0$, and define
\begin{equation}
 \overline{\mathcal F}_q(v)=\mathcal F_q(\epsilon^2v),\qquad
 \overline{\mathcal T}_q(E_0)
       =\mathcal T_q(\epsilon^2E_0).
 \label{eq:app-ODE-scaled-generators}
\end{equation}
The density in the continuum convention of~\cite{Johnson:2026boundary} is
\begin{equation}
 \rho^{\rm J}_{0,q}(E_c;E_0)
 =\frac\epsilon{2\pi}
     \rho^q_0(\epsilon^2E_c;\epsilon^2E_0,q)
 =\frac1{4\pi^2\hbar}
   \frac{\sqrt{E_c-E_0}}{E_c}
   \overline{\mathcal F}_q(E_c-E_0).
 \label{eq:app-ODE-J-density}
\end{equation}
The rescaling of the energy is accompanied by a change in the normalization of the density. At tree level, $u=\epsilon^2u^{\rm J}$ and $\xi=2\pi \xi^{\rm J}$ give
\begin{equation}
 \begin{aligned}
 t_k^{\rm J}(E_0,q)
       &=\frac{\epsilon^{2k}}{2\pi}t_k(\epsilon^2E_0,q),\\
 \mu^{\rm J}(E_0,q)
       &=\frac1{2\pi}\mu(\epsilon^2E_0,q),\\
 \widetilde\Gamma^{\rm J}(E_0,q)
       &=\frac1{2\pi\epsilon}\widetilde\Gamma^{\rm str}_q(\epsilon^2E_0).
 \end{aligned}
 \label{eq:app-ODE-normalization-map}
\end{equation}
These are relations between tree-level specializations, and a change of variables in the full quantum differential equation also rescales its derivative-counting parameter. In particular, the coupling defined in~\eqref{eq:Gamma-top} obeys
\begin{equation}
 \widetilde\Gamma_q^{\rm loc}(\Ezero)
       =\frac{\widetilde\Gamma^{\rm str}_q(\Ezero)}{2\pi(1-q)},\qquad
 \widetilde\Gamma^{\rm J}(E_0,q)
       =\widetilde\Gamma_q^{\rm loc}(\epsilon^2E_0).
 \label{eq:app-ODE-pole-conventions}
\end{equation}
This is the difference in normalization between Sections~\ref{sec:ODE} and~\ref{sec:BPS}.

The product bound~\eqref{eq:syk-product-bound} gives the compact-uniform convergence of the rescaled entire factor,
\begin{equation}
 \overline{\mathcal F}_q(v)\longrightarrow
 \prod_{n\geq1}\left(1+\frac{4v}{n^2}\right)
 =\frac{\sinh(2\pi\sqrt v)}{2\pi\sqrt v},\qquad
 \epsilon^{2n}\mathfrak f_n\longrightarrow\frac{(4\pi^2)^n}{(2n+1)!}.
 \label{eq:app-ODE-classical-F}
\end{equation}
Here the value at $v=0$ is defined by continuity. Inserting these coefficients into the moment transform, we get the Bessel function $J_0$ of the first kind,
\begin{equation}
 \overline{\mathcal T}_q(E_0)
 \longrightarrow\sum_{n\geq0}\frac{(-\pi^2E_0)^n}{(n!)^2}
 =J_0(2\pi\sqrt{E_0}).
 \label{eq:app-ODE-classical-M}
\end{equation}
The interchange of the limit and the moment transform, including its derivatives, follows from~\eqref{eq:app-ODE-M-integral} and the compact-uniform convergence of the product. Differentiating the Bessel series gives
\begin{equation*}
 \frac{d^k}{dE_0^k}J_0(2\pi\sqrt{E_0})
 =(-1)^k\pi^kE_0^{-k/2}
    J_k(2\pi\sqrt{E_0}).
\end{equation*}
So the rescaled string-equation data and their limits are
\begin{equation}
 \begin{aligned}
 t_k^{\rm J}&=\frac{(-1)^k}{2\pi(2k+1)k!}
        \frac{d^k\overline{\mathcal T}_q(E_0)}{dE_0^k}
 \ \longrightarrow\
 \frac{\pi^{k-1}J_k(2\pi\sqrt{E_0})}
      {2(2k+1)k!E_0^{k/2}},\\[3pt]
 \mu^{\rm J}=t_0^{\rm J}&=\frac{\overline{\mathcal T}_q(E_0)}{2\pi}
 \ \longrightarrow\ \frac{J_0(2\pi\sqrt{E_0})}{2\pi},\\[3pt]
 \widetilde\Gamma^{\rm J}&=\frac{\sqrt{E_0}}{2\pi}
                 \overline{\mathcal F}_q(-E_0)
 \ \longrightarrow\ \frac{\sin(2\pi\sqrt{E_0})}{4\pi^2}.
 \end{aligned}
 \label{eq:app-ODE-J-specialization}
\end{equation}
All the limits in this display are $q\to1^-$. The limiting pole coupling is the Turiaci--Witten result~\cite{Turiaci:2023jfa}, and the complete specialization of the string equation agrees with~\cite[Section V.1]{Johnson:2026boundary}. The Bessel quotients have regular limits at $E_0=0$.

Finally, the endpoint derivatives transform consistently,
\begin{equation}
 (u_0^{\rm J})^{(\ell)}(\mu^{\rm J})
 =\frac{(2\pi)^\ell}{\epsilon^2}
       u_0^{(\ell)}(\mu;\epsilon^2E_0,q).
 \label{eq:app-ODE-J-derivatives}
\end{equation}
For example, $\epsilon^2Z_1\to\pi^2/24$ and $\epsilon^4Z_2\to\pi^4/1440$ give
\begin{equation}
 \begin{aligned}
 (u_0^{\rm J})'(\mu^{\rm J})&=-4\pi E_0,\\
 (u_0^{\rm J})''(\mu^{\rm J})&\longrightarrow
       24\pi^2E_0-16\pi^4E_0^2,\\
 (u_0^{\rm J})'''(\mu^{\rm J})&\longrightarrow
       -192\pi^3E_0+416\pi^5E_0^2
                              -160\pi^7E_0^3.
 \end{aligned}
 \label{eq:app-ODE-J-derivative-check}
\end{equation}
These are the endpoint values of the undeformed string-equation calculation~\cite{Ahmed:2025extended}. The agreement requires rescaling the density, the times and the pole coefficient together, which also preserves the normalization of both the derivatives and the volumes.

%-------------------------------------------------
\subsection{Resolvent and one-boundary amplitudes}\label{subsec:ODE-resolvent}

For the auxiliary Schr\"odinger operator
\begin{equation}
 {\cal H}=-\hbar^2\partial_\xi^2+u(\xi),
\end{equation}
we use the signed, $\hbar$-rescaled diagonal resolvent
\begin{equation*}
 \widehat R(\xi,\Ehat)
 =-\hbar\left\langle\xi\left|(\mathcal H-\Ehat)^{-1}\right|\xi\right\rangle
 =\hbar\left\langle\xi\left|(\Ehat-\mathcal H)^{-1}\right|\xi\right\rangle.
\end{equation*}
With position states normalized by $\langle\xi|\xi'\rangle=\delta(\xi-\xi')$, this convention fixes the negative leading term below and gives the Gel'fand--Dikii equation~\cite{Gelfand:1975rn}
\begin{equation}
 4(u-\Ehat)\widehat R^2-2\hbar^2\widehat R\widehat R''
 +\hbar^2(\widehat R')^2=1.
\label{eq:GD}
\end{equation}
Writing
\begin{equation}
 \widehat R(\xi,\Ehat)=\sum_{g\geq0}\hbar^{2g}\widehat R_g(\xi,\Ehat),
 \qquad
 \widehat R_0=-\frac{1}{2\sqrt{u_0-\Ehat}},
\end{equation}
we expand~\eqref{eq:GD} and use the string equation to eliminate the higher corrections to $u$. For $g\geq1$, this gives
\begin{equation}
 \widehat R_g=\frac{d}{d\xi}\widehat Q_g.
\end{equation}
The double-scaled prescription keeps the endpoint value at $\xi=\mu$~\cite{Ahmed:2025extended,AhmedNaskar:2026resolvent}. At genus one, the resolvent primitive is
\begin{equation}
 \widehat R_1=\frac{d}{d\xi}\widehat Q_1,
 \qquad
 \widehat Q_1=
 \frac{u_0'}{32(u_0-\Ehat)^{5/2}}-\frac{u_0''}{48u_0'(u_0-\Ehat)^{3/2}}.
\label{eq:Q1}
\end{equation}
The endpoint prescription then gives the one-boundary correlator,
\begin{equation}
 W^{\rm loc}_{g,1}(w)=-2w\widehat Q_g(\mu,\Ehat),
 \qquad \Ehat=\Ezero-w^2.
\label{eq:WQ}
\end{equation}
In particular, the genus-one endpoint data give
\begin{equation}
 W^{\rm loc}_{1,1}(w)=\frac{\Ezero}{8w^4}+\frac{16Z_1\Ezero-1}{8w^2}.
\label{eq:ODE-W11-check}
\end{equation}
The inverse Laplace transform is particularly simple, because $u_0(\mu)-\Ehat=w^2$. In our convention, we have
\begin{equation}
 w^{-2k}\quad\longleftrightarrow\quad
 \frac{b^{2k-2}}{(2k-1)!}
\label{eq:laplace-dictionary}
\end{equation}
after removing the overall factor of $b$ in the volume transform. Using~\eqref{eq:app-ODE-u12-explicit} gives~\eqref{eq:ODE-W11-check}, and hence~\eqref{eq:V11-local}, and substituting the five endpoint derivatives into the universal $\widehat Q_2$ likewise reproduces the full $(2,1)$ volume in Appendix~\ref{app:localvolumes}.

\subsection{Loop insertion and several boundaries}\label{sec:ODE-adding-boundaries}

Further boundaries are generated by loop insertion before evaluating at $\xi=\mu$~\cite{Lowenstein:2024fji,Lowenstein:2024gvz,Johnson:2026universal,Johnson:2026boundary,AhmedNaskar:2026resolvent}. In KdV notation, we have $t_0=\xi$, with the background endpoint value $t_0=\mu$.

To make this explicit, let us write the free energy and the potential as
\begin{equation}
 F=\sum_{g=0}^{\infty}\hbar^{2g-2}F_g+\cdots,
 \qquad
 u(\xi)=2\hbar^2\frac{\partial^2F}{\partial \xi^2}.
\label{eq:ODE-free-energy}
\end{equation}
On the decreasing real branch, $u_0^{\prime}<0$, so we choose the real genus-one free energy below. With the loop operator $\delta_{\Ehat}^{(u_0)}$ defined in~\eqref{eq:ODE-loop-action}, we have
\begin{equation}
 F_1=-\frac{1}{24}\log(-u_0'),
 \qquad
 \widehat Q_1(\xi,\Ehat)=\delta_{\Ehat}^{(u_0)}F_1,
\label{eq:ODE-F1-Q1}
\end{equation}
which reproduces~\eqref{eq:Q1}. The choice of logarithm differs from $\log u_0'$ by a constant phase, so the derivatives and stable amplitudes are unchanged.

We should be careful with the normalization of the time derivatives. With $R_k[u]=u^k+O(\hbar^2)$, the KdV flows are~\cite[Section II]{Johnson:2026boundary}
\begin{equation}
 (k+1)\frac{\partial u}{\partial t_k}=R'_{k+1}[u],
 \qquad \frac{\partial u_0}{\partial t_k}=u_0^k u_0'.
 \label{eq:ODE-normalized-flows}
\end{equation}
The second relation also follows by differentiating~\eqref{eq:app-ODE-tree} at fixed $\widetilde\Gamma_q^{\rm str}$. In these conventions, the energy-space loop operator is
\begin{equation}
 \delta_{\Ehat}=\sum_{k\geq0}\frac{(-1)^k(2k+1)!!}{2^{k+1}k!}\,
 \frac{1}{(-\Ehat)^{k+3/2}}\,\frac{\partial}{\partial t_k}.
\label{eq:ODE-loop-times}
\end{equation}
We write $\delta_{\Ehat}^{(u_0)}$ for its action on functions of the classical potential $u_0$, where the dispersionless flow in~\eqref{eq:ODE-normalized-flows} applies. Its coefficients are the Taylor coefficients of $\tfrac12(1+v)^{-3/2}$, so summing with $v=u_0/(-\Ehat)$, we find
\begin{equation}
 \delta_{\Ehat}^{(u_0)}u_0(\xi)=
 \frac{1}{2}\frac{u_0'(\xi)}{[u_0(\xi)-\Ehat]^{3/2}}.
\label{eq:ODE-loop-action}
\end{equation}
The operator commutes with $\xi$-derivatives, so that
\begin{equation}
 \delta_{\Ehat}^{(u_0)}u_0^{(m)}(\xi)
 =\frac{d^m}{d\xi^m}
 \left\{\frac{u_0'(\xi)}{2[u_0(\xi)-\Ehat]^{3/2}}\right\}.
\label{eq:ODE-loop-prolongation}
\end{equation}
Consequently, it acts by the ordinary chain rule on every differential polynomial in $u_0$, and the connected energy-space correlators are generated by repeated insertions,
\begin{equation}
 \widetilde W_{g,n}(\Ehat_1,\ldots,\Ehat_n)=
 \left.(-1)^{n-1}
 \delta_{\Ehat_n}^{(u_0)}\cdots\delta_{\Ehat_1}^{(u_0)}
 F_g[u_0',u_0'',\ldots]\right|_{\xi=\mu},
 \qquad g\geq1,
\label{eq:ODE-Wtilde}
\end{equation}
and the correlators in the local uniformizing variables $w_i^2=\Ezero-\Ehat_i$ are
\begin{equation}
 W^{\rm loc}_{g,n}(w_1,\ldots,w_n)
 =\prod_{i=1}^{n}(-2w_i)\,
 \widetilde W_{g,n}(\Ehat_1,\ldots,\Ehat_n).
\label{eq:ODE-Wgn}
\end{equation}
The insertions commute, and the stable sphere amplitudes follow from the same boundary-addition prescription applied to the genus-zero base data.

For the present curve, with $u_k\equiv u_0^{(k)}(\mu)$, this gives
\begin{equation}
 W^{\rm loc}_{0,3}(w_1,w_2,w_3)
 =-\frac{u_1}{2w_1^2w_2^2w_3^2},
\label{eq:ODE-W03}
\end{equation}
\begin{equation}
 \begin{split}
 W^{\rm loc}_{0,4}(w_1,w_2,w_3,w_4)={}&
 -\frac{u_2}{2w_1^2w_2^2w_3^2w_4^2}\\
 &+\frac{3u_1^2}{4}
 \frac{\displaystyle\sum_{i=1}^{4}\prod_{j\neq i}w_j^2}
 {w_1^4w_2^4w_3^4w_4^4},
 \end{split}
\label{eq:ODE-W04}
\end{equation}
and
\begin{equation}
\begin{split}
 W^{\rm loc}_{1,2}(w_1,w_2)=\frac{1}{w_1^2w_2^2}\Bigg[&
 \frac{u_3}{24u_1}-\frac{u_2^2}{24u_1^2}
 -\frac{u_2}{8}\left(\frac{1}{w_1^2}+\frac{1}{w_2^2}\right)\\
 &+\frac{5u_1^2}{32}\left(\frac{1}{w_1^4}+\frac{1}{w_2^4}\right)
 +\frac{3u_1^2}{32}\frac{1}{w_1^2w_2^2}\Bigg].
\end{split}
\label{eq:ODE-W12}
\end{equation}
Substituting~\eqref{eq:app-ODE-u12-explicit}--\eqref{eq:app-ODE-u3-explicit} and applying~\eqref{eq:laplace-dictionary}, we reproduce the local $V^q_{0,3}$, $V^q_{0,4}$ and $V^q_{1,2}$ in Appendix~\ref{app:localvolumes}. Together with the genus-two one-boundary result, these are independent checks of the lower local recursion.

The hard-edge string equation gives the finite $\Ezero\to0$ limits of these ratios and correlators, again in agreement with the recursion.

%-------------------------------------------------
\section{Theta-function form of the pole couplings}
\label{app:theta}
%-------------------------------------------------

Both couplings in Section~\ref{sec:BPS} are values of the function $\mathfrak g_q$ in~\eqref{eq:syk-coupling-product}. Using~\eqref{eq:syk-coupling-theta}, with $q=\e^{-\lambda}$ and $\tilde q=\e^{-\pi^2/\lambda}$, we can write
\begin{equation}
 \mathfrak g_q(v)=
 \frac{\sqrt{\pi/\lambda}\,q^{-v^2-1/4}}
 {8\pi(1-q)(q^2;q^2)_\infty^3}
 \vartheta_1(\pi v,\tilde q),
\end{equation}
and
\begin{equation}
 \begin{aligned}
 \widetilde\Gamma_q^{\rm glob}(\Ezero)&=\mathfrak g_q(|r|),
 & |r|&=\frac{2}{\lambda}\operatorname{arsinh}\sqrt{\Ezero},\\
 \widetilde\Gamma_q^{\rm loc}(\Ezero)&=\mathfrak g_q\big(\lambda^{-1}\operatorname{arsinh}(2\sqrt{\Ezero})\big).
 \end{aligned}
\end{equation}
The first line belongs to the global curve and includes the factor $1+\Ezero$ in~\eqref{eq:Gamma-global}, while the second belongs to the lower local curve~\eqref{eq:Gamma-top}. We derive these expressions from the product identity~\eqref{eq:syk-coupling-theta} in Appendix~\ref{app:SYK-continuum}. Their arguments differ at finite gap, but they agree in the correlated continuum limit, which gives~\eqref{eq:Gamma-limit}. Note that only the global coupling enters the microscopic BPS relation~\eqref{eq:BPS-residue} directly.

%-------------------------------------------------
\Needspace{12\baselineskip}
\section{Details of the intersection formula}
\label{app:intersection-proofs}

In this appendix, we prove the intersection formula~\eqref{eq:it-volume-class} and the boundary-addition identity~\eqref{eq:it-boundary-addition}, using the notation of Section~\ref{sec:intersection}.

\subsection{From the local curve to the intersection class}
\label{app:intersection-translation}

At nonzero gap, the one-branch intersection theorem applies~\cite{Eynard:2011intersection,DuninBarkowski:2012givental}. Since $x_{\rm loc}=w^2/2$ and $B$ has no regular part, the normalized coefficients of $y_{\rm loc}$ give the translation insertion~\cite{ChekhovNorbury:2019}
\begin{equation}
 \mathcal Y(v)=-\sum_{m\geq1}a_m(\Ezero,q)v^{m+1}.
 \label{eq:it-descendant-translation}
\end{equation}
The translation adds auxiliary markings carrying $\mathcal Y(\psi)$, and forgetting them converts descendants into $\kappa$-classes~\cite{MulaseSafnuk:2006}, through
\begin{equation}
 (\mathsf f_\ell)_*\prod_{j=1}^\ell\psi_{n+j}^{m_j+1}
 =
 \sum_{\mathfrak s\in\mathfrak S_\ell}\;
 \prod_{\substack{\text{cycles }(j_1\cdots j_l)\\ \text{of }\mathfrak s}}
 \kappa_{m_{j_1}+\cdots+m_{j_l}},\qquad m_j\geq1,
 \label{eq:it-cycle-pushforward}
\end{equation}
where $\mathsf f_\ell$ forgets the last $\ell$ markings~\cite{KaufmannManinZagier:1996,Norbury:2020enumerative}. Summing over the cycles, we get
\begin{equation}
 \sum_{\ell\geq0}\frac1{\ell!}(\mathsf f_\ell)_*
       \prod_{j=1}^\ell \mathcal Y(\psi_{n+j})
 =
 \exp\!\left[
  \sum_{m\geq1}
   [v^m]\!\left\{-\log\!\left(1-\frac{\mathcal Y(v)}v\right)\right\}\kappa_m
 \right]
 =
 \e^{\sum_{m\geq1}\sigma_m\kappa_m}.
 \label{eq:it-translation-proof}
\end{equation}
Here the $\ell=0$ term is $1$, and $\mathcal Y(v)=O(v^2)$ makes the sum finite in each degree.

Scaling the leading coefficient of $y$ from $1$ to $1/\Ezero$ multiplies the stable correlator by $\Ezero^{2g-2+n}$, so the translation theorem gives
\begin{equation}
 \omega^{\rm loc}_{g,n}(\mathbf w)=
 \sum_{\mathbf k\geq0}
 \left(\int_{\overline{\mathcal M}_{g,n}}
      \Omega^q_{g,n}(\Ezero)\prod_i\psi_i^{k_i}\right)
 \prod_i\frac{(2k_i+1)!!\,dw_i}{w_i^{2k_i+2}}.
 \label{eq:it-descendant-correlator}
\end{equation}
The boundary normalization then follows from the elementary Laplace transform
\[
 \int_0^\infty b\,\e^{-bw}\frac{b^{2k}}{2^k k!}\,db
 =\frac{(2k+1)!!}{w^{2k+2}},
\]
which proves~\eqref{eq:it-volume-class} in the convention~\eqref{eq:local-volume-def}, for all stable marked topologies.

\subsection{Adding and forgetting a marked point}
\label{app:intersection-boundary}

Let us take $\Ezero\ne0$ and a stable $(g,n)$, and let $\mathsf f:\overline{\mathcal M}_{g,n+1}\to\overline{\mathcal M}_{g,n}$ forget the last marking. Let $\mathcal D_i=\mathcal D_{i,n+1}$ be the divisor on which that marking and marking $i$ lie on a three-pointed rational tail. On $\overline{\mathcal M}_{g,n+1}$, with $\psi=\psi_{n+1}$, we have the standard relations
\begin{equation}
 \begin{gathered}
 \kappa_m^{(n+1)}=\mathsf f^*\kappa_m^{(n)}+\psi^m,\qquad
 \psi_i^{(n+1)}=\mathsf f^*\psi_i^{(n)}+\mathcal D_i,\\
 \mathcal D_i\psi=0,\qquad \mathcal D_i^2=-\mathcal D_i\mathsf f^*\psi_i,\qquad
 \mathcal D_i\mathcal D_j=0\quad(i\neq j).
 \end{gathered}
 \label{eq:it-forgetful-relations}
\end{equation}
Consequently, $\Omega^q_{g,n+1}=\Ezero\mathsf f^*\Omega^q_{g,n}\,\e^{\sum_{m\geq1}\sigma_m\psi^m}$. The term without a divisor pushes forward using $\mathsf f_*\psi^{m+1}=\kappa_m$, which gives the second line of~\eqref{eq:it-boundary-addition}. For a divisor term, we use
\begin{equation}
 \e^{\frac12b_i^2(\mathsf f^*\psi_i+\mathcal D_i)}
 =\e^{\frac12b_i^2\mathsf f^*\psi_i}
   +\mathcal D_i\,\frac{\e^{\frac12b_i^2\mathsf f^*\psi_i}-1}{\mathsf f^*\psi_i}.
 \label{eq:it-divisor-resummation}
\end{equation}
The quotient is a power series, equal to $\int_0^{b_i^2/2}\e^{v\mathsf f^*\psi_i}\,dv$. Since $\psi$ vanishes on the tail, the new-boundary insertion restricts to $1$, and $\mathsf f_*\mathcal D_i=1$ produces the first line. In particular, the factor $\Ezero$ is fixed by $\Omega^q_{g,n+1}|_{\mathcal D_i}=\Ezero\Omega^q_{g,n}$.

%-------------------------------------------------
\begingroup
% Keep each reference on a single page.
\raggedbottom \interlinepenalty=10000
\bibliographystyle{JHEP}
\bibliography{main}
\endgroup
\end{document}